\documentclass[a4paper,UKenglish,cleveref, autoref, thm-restate]{lipics-v2021}

\hideLIPIcs  %

\usepackage[bibliography=common]{apxproof}
\usepackage{multicol}
\usepackage{subcaption}
\usepackage{hyperref}
\usepackage{algpseudocode}
\usepackage{tikz}
\usepackage{leftidx}
\usepackage{pifont}
\usepackage{enumitem}
\usepackage{tikz-cd}
\usepackage{array}
\usepackage{mathtools}
\usepackage{xspace}
\usepackage{xcolor}

\DeclareMathAlphabet{\mymathbb}{U}{BOONDOX-ds}{m}{n}
\newcommand{\cmark}{\text{\ding{51}}}
\newcommand{\xmark}{\text{\ding{55}}}

\newcommand{\Out}{\textbf{\upshape{Out}}}
\newcommand{\Act}{\textbf{\upshape{Act}}}
\newcommand{\Id}{\text{\upshape{Id}}}
\newcommand{\At}{\text{\upshape{At}}}

\newcommand{\Exp}{\textbf{\upshape{Exp}}}
\newcommand{\BA}{\textsf{BA}}
\newcommand{\ex}{\text{\upshape{exp}}}
\newcommand{\BExp}{\textbf{\upshape{BExp}}}
\newcommand{\WExp}{\textbf{\upshape{WExp}}}
\newcommand{\TExp}{\textbf{\upshape{TExp}}}
\newcommand{\HCo}{$\mathcal{H}$\text{-coalgebra }}

\newcommand{\WC}[2]{\,\leftidx{_{#1}}{{\oplus}}{_{#2}}\,}
\newcommand{\Ax}[1]{\ensuremath{\mathsf{#1}}}
\newcommand{\Mon}{\mathcal{M}_\omega}
\newcommand{\bskip}{\mymathbb{1}}
\newcommand{\babort}{\mymathbb{0}}
\newcommand{\wgkat}{\textsf{\upshape{wGKAT}}\xspace}
\newcommand{\gkat}{\textsf{\upshape{GKAT}}\xspace}
\newcommand{\probgkat}{\textsf{\upshape{ProbGKAT}}\xspace}
\newcommand{\kat}{\textsf{\upshape{KAT}}\xspace}
\newcommand{\kawt}{\textsf{\upshape{KAWT}}\xspace}
\newcommand{\netkat}{\textsf{\upshape{NetKAT}}\xspace}
\newcommand{\coin}{\textcircled{{\scriptsize \$}}}
\newcommand{\exN}{\mathbb{N}^{+\infty}}
\newcommand{\exQp}{\mathbb{Q}_{\geq0}^{+\infty}}
\newcolumntype{L}{>{$}l<{$}} %

\DeclareUnicodeCharacter{0301}{*}
\allowdisplaybreaks

\nolinenumbers

\theoremstyle{plain}%
\newtheoremrep{thm}{Theorem}[section]
\newtheoremrep{claim}{Claim}[section]
\newtheoremrep{lem}[thm]{Lemma}
\newtheoremrep{prop}[thm]{Proposition}
\newtheoremrep{defn}{Definition}[section]
\newtheoremrep{conj}{Conjecture}[section]
\newtheoremrep{cor}{Corollary}[section]
\newtheoremrep{exmp}{Example}[section]
\newtheoremrep{rem}{Remark}[section]

\makeatletter
\newcommand{\bigplus}{%
	\DOTSB\mathop{\mathpalette\mattos@bigplus\relax}\slimits@
}
\newcommand\mattos@bigplus[2]{%
	\vcenter{\hbox{%
			\sbox\z@{$#1\sum$}%
			\resizebox{!}{0.9\dimexpr\ht\z@+\dp\z@}{\raisebox{\depth}{$\m@th#1+$}}%
	}}%
	\vphantom{\sum}%
}
\makeatother

\title{Weighted GKAT: Completeness and Complexity}

\author{Spencer Van Koevering}{Cornell University, United States of America \and \url{https://spencerkoevering.github.io/}} {sv493@cornell.edu}{https://orcid.org/0009-0008-5026-8060}{Partially supported by a Cornell Bowers CIS Deans Excellence Fellowship.}%
\author{Wojciech {Różowski}}{University College London, United Kingdom \and \url{https://wkrozowski.github.io/}}{w.rozowski@cs.ucl.ac.uk}{https://orcid.org/0000-0002-8241-7277}{Partially supported by ERC grant Autoprobe (grant agreement 101002697).}
\author{Alexandra Silva}{Cornell University, United States of America \and \url{https://alexandrasilva.org/}}{alexandra.silva@cornell.edu}{https://orcid.org/0000-0001-5014-9784}{Partially supported by ONR grant N68335-22-C-0411, ERC grant 101002697, and a Royal Society Wolfson fellowship.}

\authorrunning{S. Van Koevering, W. {Różowski}, A. Silva} 

\Copyright{Spencer Van Koevering, Wojciech {Różowski}, and Alexandra Silva} 

\EventEditors{Keren Censor-Hillel, Fabrizio Grandoni, Joel Ouaknine, and Gabriele Puppis}
\EventNoEds{4}
\EventLongTitle{52nd International Colloquium on Automata, Languages, and Programming (ICALP 2025)}
\EventShortTitle{ICALP 2025}
\EventAcronym{ICALP}
\EventYear{2025}
\EventDate{July 8--11, 2025}
\EventLocation{Aarhus, Denmark}
\EventLogo{}
\SeriesVolume{334}
\ArticleNo{160}


\category{} %

\relatedversion{} %

\begin{document}
	
	\maketitle

	\begin{abstract}
	 We propose Weighted Guarded Kleene Algebra with Tests (\wgkat), an uninterpreted weighted programming language equipped with branching, conditionals, and loops. We provide an operational semantics for \wgkat using a variant of weighted automata and introduce a sound and complete axiomatization. We also provide a polynomial time decision procedure for  bisimulation equivalence.	\end{abstract}

		\begin{toappendix}
			\section{Coalgebra}\label{AppCoalg}
			While we strove to use the language of automata in the body of this paper, the underlying arguments rely heavily upon the idea of universal coalgebra \cite{Rutten2000UniversalCA}. This idea allows us to discuss transition systems in a principled and simplified manner. Universal Coalgebra is based on the language of category theory, which we will assume the reader has a brief awareness of. We will recap some relevant category theoretic ideas, but otherwise refer the reader to \cite{awodey2010category}.
			
			\begin{defn}[\cite{Rutten2000UniversalCA}]
				Given a functor $F$, an $F$-coalgebra is a pair $(X, \beta)$ where $X$ is a set, and $\beta:X \to FX$. The set $X$ is called the carrier, or set of states, of the system and the function $\beta$ is called the transition function. 
			\end{defn}
			\begin{defn}
				Let $(X, \beta), (Y, \gamma)$ be $F$-coalgebras. A function $f: X \to Y$ is an $F$-coalgebra homomorphism if $F(f)\circ \beta = \gamma \circ f$. That is, if the following diagram commutes:
				\begin{center}
				\begin{tikzcd}
					X && Y \\
					\\
					FX && FY
					\arrow["f", from=1-1, to=1-3]
					\arrow["\beta"{description}, from=1-1, to=3-1]
					\arrow["\gamma"{description}, from=1-3, to=3-3]
					\arrow["Ff", from=3-1, to=3-3]
				\end{tikzcd}
			\end{center}
			\end{defn}
			
			Bisimulation is also a coalgebraic notion.
			\begin{defn}[\cite{Rutten2000UniversalCA}]
				Let $(X, \beta)$ and $(Y, \gamma)$ be $F$-coalgebras. A relation $R\subseteq X \times Y$ is a \emph{bisimulation} if there exists a function $\rho: R \to FR$ such that the projections $\pi_1: R\to X, \pi_2: R\to Y$ are $F$-coalgebra homomorphisms from $(R, \rho)$ to $(X, \beta)$ and $(Y, \gamma)$ respectively. We say the elements $x\in X, y\in Y$ are \emph{bisimilar} if there exists a bisimulation R such that $(x,y) \in R$. \label{bisimulation}
			\end{defn}
			
			Many of the arguments regarding bisimulation made require that the coalgebra functor $F$ preserve weak pullbacks \cite{Rutten2000UniversalCA}. In particular we will need this property in order to establish the existence of a bisimulation.
			\begin{defn}[\cite{awodey2010category},\cite{gumm2021freelatticefunctorsweaklypreserve}]
				In a category $C$, given morphisms $f$ and $g$ with $\text{cod}(f) =  \text{cod}(g)$, the \emph{pullback} of $f$ and $g$ is the unique object $P$ and morphisms $p_1$, $p_2$ as in the following diagram, such that $fp_1 = gp_2$. 
				\begin{center}\begin{tikzcd}
					P & A \\
					B & C
					\arrow["{p_2}"', from=1-1, to=1-2]
					\arrow["{p_1}", from=1-1, to=2-1]
					\arrow["g"', from=1-2, to=2-2]
					\arrow["f", from=2-1, to=2-2]
				\end{tikzcd}\end{center}
				
				Furthermore the pullback is universal with this property, that is, given any $z_1: Z \to A$ and $z_2: Z \to B$ with $fz_1=gz_2$ there exists a unique $u:Z\to P$ such that $z_1 = p_1u$ and $z_2 = p_2u$. If the mapping $u$ is not necessarily unique, then it is called a \emph{weak} pullback.
			\end{defn}
			
			\begin{defn}[\cite{GUMM2001185}]
			A functor $F$ preserves weak pullbacks if it transforms weak pullback diagrams into weak pullback diagrams.
			\end{defn}
		\end{toappendix}
		\section{Introduction}\label{intoductionS}
		
		Weighted programming is a recently introduced paradigm~\cite{batz2022weighted} in which computational branches of programs are associated with quantities. For example, the program $e \WC{\frac25}{\frac35} f$ does $e$ with weight $\frac25$ and $f$ with weight $\frac35$. Semantically, each branch is executed, its results scaled, and then results of different branches added. If we were to interpret weights as probabilities and use addition of real numbers, then the above example would be a classic randomized program: a biased coin flip. If we interpret the weights as costs rather than a probability, and addition as minimum rather than traditional addition, then $e \WC{\frac25}{\frac35} f$ would model selecting the cheapest resulting branch and be well suited for optimization problems.	
			
		Weighted programs can model problems from a variety of domains, including optimization, model checking, and combinatorics \cite{batz2022weighted}. However, reasoning about their behavior is difficult as different interpretations of the weights lead to different properties of the programs and their semantics, impacting, for instance, whether the equivalence of weighted programs is decidable. In this paper, we propose an algebraic approach to reason about equivalence of weighted programs and devise a decision procedure which applies to a large class of weights. 
		
		The starting point of our development is Guarded Kleene Algebra with Tests (\gkat) \cite{Smolka_2019}, an algebraic framework for reasoning about equivalence of uninterpreted imperative programs. Its equational axiomatization offers a simple framework for deductive reasoning, and it also has a decision procedure. We propose and axiomatize a \gkat-inspired weighted programming language, \wgkat, which enables reasoning about the equivalence of deterministic, uninterpreted programs with {\em weighted branching}. We design this language with the goals of offering an equational axiomatization and a polynomial-time decision procedure for the equivalence of expressions, both derived uniformly for a broad class of weights.
		
		One of the primary challenges we address is identifying a suitable class of semirings which is sufficient to prove that bisimulation is a complete proof technique for behavioural equivalence and, moreover, to define a operational semantics of loops. Using this class of semirings we were able to generalize traditional arguments for soundness and completeness from process algebra to our setting \cite{ToddThesis}. Another challenge is the decidability of equivalence for weighted programs, which hinges on properties of the semiring as well. We leverage existing work on monoid labeled transition systems to design the operational semantics and establish that equivalence up to bisimilarity is decidable for \wgkat.
		
		Our work was developed concurrently with another weighted variant of Kleene Algebra with Tests, called Kleene Algebra with Weights and Tests (\kawt) \cite{sedlar2023kleenealgebratestsweighted}. \kawt lacks a decision procedure and admits only a restrictive class of weights. The key difference in our work that enables us to significantly enlarge the class of weights we can consider is that \kawt extends \textsf{KAT} directly whereas our starting point is \gkat, the deterministic fragment of \textsf{KAT}.  This is what allows us to offer a decision procedure and a sound and complete axiomatization: by carefully selecting the class of weights we retain the general road map to soundness and completeness of \gkat as well as the efficient decision procedure. 
		
		In a nutshell, our contributions are as follows. 
		\begin{enumerate}
			\item We propose a weighted version of \gkat and equip it with an operational semantics using a variant of weighted automata (\cref{syntaxS,semanticsS}). 
			\item We identify a class of semirings for which the semantics of weighted loops can be computed when considering bisimulation equivalence (\cref{equivS}).
			\item We axiomatize bisimulation equivalence for \wgkat and provide a proof soundness (\cref{axiomsS}) and a proof of completeness (\cref{completenessS}).
			\item We design a decision procedure for bisimilarity which runs in $O(n^3\log^2n)$ time if the number of tests is fixed, where $n$ is the size of the programs considered (\cref{decisionS}). 
				\end{enumerate}
	We conclude and discuss directions for future work in \cref{conclS}. 
		\section{Syntax}\label{syntaxS}
		\begin{figure}[t]
			\begin{alignat*}{10}
			\\[-6ex]
	e,f \in \Exp &\coloneq& \; \textbf{p} \in \Act&\qquad& \text{do } \textbf{p}\\
			&\mid&b \in \BExp&\qquad\qquad\qquad\qquad\qquad& \text{assert } \textbf{b}\\
			&\mid&e +_b f &\qquad\qquad& \text{if }\textbf{b}\text{ do }e\text{ else do } f\\
			&\mid&e;f &\qquad\qquad& \text{do }e\text{ then do }f\\
			&\mid&e^{(b)}&\qquad\qquad& \text{while } \textbf{b} \text{ do } e\\
			&\mid&v \in \Out&\qquad\qquad& \text{return } \textbf{v}\\
			&\mid&e \WC{r}{s} f &\qquad\qquad& \text{do } e \text{ with weight } \textbf{r} \text{ and } f \text{ with weight } \textbf{s}\\[-6ex]
			\end{alignat*}
			\caption{Syntax of \wgkat}
			\label{semisyntax}
		\end{figure}
		In this section, we introduce the syntax of our language and examples on how to use it to model simple quantitative programs. 
	The syntax of \wgkat (\Cref{semisyntax}) is two-sorted consisting of the set $\Exp$ of expressions containing a sort $\BExp$ of Boolean expressions also called tests. Intuitively, $\Exp$ represents the overall syntax of weighted programs, while $\BExp$ is used to specify assertions, as well as conditional expressions appearing in the scope of \texttt{if-then-else} and \texttt{while} constructs. Given a finite set $T$ of \emph{primitive tests}, we define $\BExp$, the grammar of tests, to be given by the following:
			$$b,c \in \BExp\coloneq \babort \mid \bskip \mid t\in T \mid \bar b \mid b+c \mid bc$$
	In the above, $\babort$ and $\bskip$ respectively represent false and true, $\bar{\cdot}$ denotes negation, $+$ represents disjunction and juxtaposition is conjunction. We will write $\equiv_{\BA}$ to denote Boolean equivalence in $\BExp$. Logical entailment defines a preorder given by $b \leq_{\BA} c \iff b + c \equiv_{\BA} c$. The quotient of $\BExp$ by the relation $\equiv_{\BA}$, denoted by ${\BExp}/{\equiv_{\BA}}$ is in one-to-one correspondence with the Boolean Algebra freely generated by the set $T$. The entailment $[b]_{\equiv_{\BA}} \leq [c]_{\equiv_{\BA}} \iff b \leq_{\BA} c$ defines a partial order on ${\BExp}/{\equiv_{\BA}}$. The minimal nonzero elements in that order are called \emph{atoms} and we will write $\At$ for the set of them. 
			
	The syntax of \wgkat is parametric on two fixed sets: $\Act$ of uninterpreted program actions and $\Out$ of possible return values. The first five constructs in \Cref{semisyntax}, namely uninterpreted program actions, Boolean assertions, conditionals, sequential composition and while loops capture the syntax of \gkat. The \emph{return values} are inherited from \probgkat~\cite{rozowski2023probabilistic} and intuitively correspond to the \texttt{return} operation from imperative programming languages. 

		The new syntactic construct in \wgkat is \emph{weighted choice} denoted by $e \WC{r}{s} f$. Intuitively, given two programs $e$ and $f$, $e \WC{r}{s} f$ represents a program that runs $e$ with the weight of $r$ and $f$ with the weight of $s$. To lighten up notation, we will follow a convention that sequential composition binds tighter than weighted choice. Moreover, we define a scaling operation as the following syntactic sugar for a program that immediately continues with the weight of r.
		\begin{defn}[Scaling] \label{scalingdef}
			$\odot r \coloneq \bskip \WC{r}{0} \babort$
		\end{defn} 
		As a notational convention, scaling will take the highest precedence out of all $\wgkat$ operators.
		
		 Weights in the aforementioned syntactic constructs are drawn from the set $R$ equipped with a structure of a semiring. The usage of semirings for modeling weighted computation has been standard in automata theory, program semantics and verification communities \cite{weightedhandbook}. We start by recalling the appropriate definitions concerning semirings.

		\begin{defn}[\cite{golan2013semirings}]
			A semiring $(R, +, \cdot, 0, 1)$ is a nonempty carrier set $R$ equipped with two monoid structures interacting in the suitable way.
			\begin{enumerate}
				\item $(R,+, 0)$ is a commutative monoid with identity $0$
				\item $(R, \cdot, 1)$ is a monoid with identity element $1\neq 0$
				\item $a\cdot (b+c) = a\cdot b+a \cdot c$ and $(b+c)\cdot a = b\cdot a+c \cdot a$ for all $a,b,c \in R$
				\item $0\cdot a=0=a \cdot 0$ for all $a$ in $R$
			\end{enumerate}\label{semiringDef}
		\end{defn}
		We will abuse the notation and write multiplication in $(R,\cdot,1)$ as juxtaposition; additionally we will write $R$ to denote the whole algebraic structure, when the operations defined on the carrier are clear from the context. Note that additive and multiplicative identities of the semiring (respectively denoted by $0$ and $1$) are distinct from false and true tests (respectively denoted by $\babort$ and $\bskip$). 
		
		Classic examples of semirings include (non-negative) reals/rationals with addition and multiplication and the tropical semiring of extended non-negative natural numbers (denoted by $\exN$), where the addition is given by pairwise minimum and semiring multiplication is a sum of extended natural numbers.
		
		Later, we will consider semirings subject to some mild algebraic constraints, allowing us to meaningfully define an operational model for \wgkat programs, in particular with loops. We will make those constraints precise in further sections. Before we do so, we first give two examples of problems that could be encoded using the syntax that we have just described.
		
		\begin{exmp}[Ski Rental \cite{batz2022weighted}]
			Consider the problem of a person going on a ski trip for $n$ days. Each day they either rent skis for a cost of $1$, or they buy them for a cost of $y$ and no longer have the need to rent. We can encode the situation of making this choice on the trip lasting $n$ days, via the following \wgkat expression over the tropical semiring:
			$$(\bskip \WC{1}{y} v)^n;v$$
			Here, for positive $n \in \mathbb{N}$, we write $e^n$ the $n$-fold sequential composition of the expression $e \in \Exp$. 
			
			Essentially, we perform at most $n$ choices between paying $1$ and immediately terminating upon paying $y$. The immediate termination is represented using a return value $v$. \label{ski1}
		\end{exmp}
		\begin{exmp}[Coin Game]
			Consider the problem of playing a game where a coin is flipped and if it is heads the player wins a dollar and the game continues but if it is tails no money is won and the game ends. How much money should a player expect to make? To model expected value we will use weights to represent both values and probabilities. If we let $\coin$ be a return value representing the outcome of winning a dollar, then we can encode this situation via the following \wgkat expression over the semiring of extended non-negative rationals:
			$$\left( \left(\bskip \WC{1}{1} \left(\odot 1;\coin \right)\right) \WC{.5}{.5} \left(\odot 0;\coin\right)\right)^{(\bskip)}$$
			The while true makes the program continue to execute until some terminating output is reached. We interpret the choice weights as probabilities. The outer weighted choice captures the coin flip, so each branch executes with probability $.5$. The inner weighted choice is slightly different, it represents an ``\textit{and}''. Each branch executes with probability $1$. This can be thought of as similar to concurrently executing both branches rather than nondeterministically choosing one. The intuitive difference stems from the fact that addition in the rationals is not idempotent. The weights in the scalings represent the \emph{value} of the outcome $\coin$, it is \emph{how many} dollars the player wins. \label{coin1}
		\end{exmp}

	\section{Operational Semantics}\label{semanticsS}
	In this section, we present an operational semantics for \wgkat using a variant of weighted automata~\cite{weightedhandbook}, which are akin to GKAT automata~\cite{Smolka_2019}, where additionally each transition carries a weight taken from a semiring. Before we formally introduce our operational model, we recall the necessary definitions concerning weighted transitions.

	\newcommand{\supp}{\mathsf{supp}}

	For a set $X$ and a fixed semiring $\mathbb S$, define $\Mon(X) = \{ \varphi \mid \varphi \colon X \to \mathbb S, \supp(\varphi)\text{ is finite}\}$. Here, $\supp$ refers to the support of a function given by $\text{supp}(\nu) = \{x \in X\;|\;\nu(x) \neq 0\}$. Elements of the set $\Mon(X)$ are functions $\nu:X\to \mathbb S$ which we will refer to as weightings on $X$ over $\mathbb S$. We can view these  as formal sums over $X$ with coefficients from $\mathbb S$ or, alternatively, as a generalization of finitely supported distributions obtained by replacing probabilities in an unit interval with weights taken from a semiring and by dropping the normalization requirement. Because of this analogy, given an element $x \in X$, we can define an analog of Dirac delta specialized to weightings, namely $\delta_x : X \to \mathbb S$, which maps its argument to the semiring multiplicative identity $1$ if it is equal to $x$, and the semiring additive identity $0$ otherwise. Given a set $A \subseteq X$, we will write $\nu[A] \coloneq \sum_{x\in A} \nu(x)$ for the total weight of the set $A$ under the weighting $\nu$. Moreover, any function $h \colon X \to \Mon(Y)$ can be uniquely extended to a function $\mathsf{lin}(h) \colon \Mon(X) \to \Mon(Y)$, called the {\em linearization of $h$}, given by $$\mathsf{lin}(h)(\varphi)(y) = \sum\limits_{x\in X} \varphi(x)h(x)(y)$$
	Crucially, the above satisfies the property that $\mathsf{lin}(h) \circ \delta = h$, where $\delta: X \to \Mon(X)$ is defined by $\delta(x) = \delta_x$. A categorically minded reader might observe that weightings are precisely free semimodules over $\mathbb S$ and hence form a monad on the category of sets. The linerization $\mathsf{lin}(h)$ is the unique $\mathbb S$-semimodule homomorphism of the type $\Mon(X) \to \Mon(Y)$, that allows to factor $h$ through $\delta$. This is a consequence of a free-forgetful adjunction between the category of sets and the category of $\mathbb S$-semimodules.

	\begin{toappendix}
		
		$\Mon$ is a functor \cite{10.1007/978-3-030-30942-8_18} that takes sets and functions to finite weightings and morphisms between weightings. Furthermore $\Mon(X)$ is an $\mathbb S$ \textit{semimodule} when we have pointwise addition and pointwise scalar multiplication \cite{vanheerdt2021learningweightedautomataprincipal}. Where a semimodule is defined by:
		\begin{defn}[\cite{golan2013semirings}]\label{semimodule}
			Let $\mathbb S$ be a semiring. A (left) $\mathbb S$ semimodule is a commutative monoid $(M, +)$ with identity $0_M$, for which we have a function $\cdot: \mathbb S \times M \to M$ called scalar multiplication where for all $s,r \in \mathbb S, n,m \in M$
			\begin{align*}
				s\cdot 0_M &= 0_M &  0\cdot m &= 0_M &  1\cdot m &= m \\
				s\cdot(m+n) &= s\cdot m+s\cdot n &  (s+r)\cdot m &= s\cdot m+r\cdot m&  (sr)\cdot m &= s\cdot (r\cdot m) 
			\end{align*}
		\end{defn}	
		
		We require that semimodules have commutative addition, and the result that $\Mon(X)$ is a semimodule  does not \cite{vanheerdt2021learningweightedautomataprincipal}. Clearly the commutativity of addition in $\mathbb S$ lifts to component-wise addition of our formal sums. 		
		
		\begin{defn}[\cite{Golan2003}]\label{smhom}
			Let $\mathbb S$ be a semiring and let $M, M^\prime$ be $\mathbb S$ semimodules with $m, m^\prime \in M, r \in \mathbb S$, then a function $\alpha:M \to M^\prime$ is a homomorphism if and only if both of the following are true
			\begin{align*}
				\alpha(m+m^\prime) &= \alpha (m) + \alpha (m^\prime)\\
				\alpha(rm) &= r \alpha (m)
			\end{align*}
		\end{defn}
			
	\end{toappendix}
			
	\subparagraph*{\wgkat automata.} A \wgkat automaton is a pair $(X, \beta)$ consisting of a set of states $X$ and a transition function $\beta: X\to \Mon(2+\Out+\Act\times X)^{\At}$ that assigns to each state and a Boolean atom (capturing the current state of the variables) a weighting over three kinds of elements, representing different ways a state might perform a transition: 
	\begin{itemize}
	\item an element of $2 = \{\cmark, \xmark\}$, representing immediate acceptance or rejection;
	\item an element of  $\Out$, representing a return value;
	\item a pair consisting of an atomic program action $\Act$ and a next state in $X$, representing performing a labelled transition to a next state. 
	\end{itemize}
	To lighten up notation, we will write $\beta(x)_\alpha$ instead of $\beta(x)(\alpha)$.

		\begin{toappendix}
			We will model our \wgkat automata as coalgebras for the functor $\mathcal{H} = \Mon(2+\Out+\Act\times \Id)^{\At}$ when we reason about it.
			\begin{defn}
				An \HCo is a tuple $(X, \beta)$ where $X$ is a set of states, and $\beta$ is a function $\beta: X \to \Mon(2+\Out+\Act\times X)^{\At}$.
			\end{defn}
		\end{toappendix}
		
		\begin{exmp}\label{automatonexample}
			Let $X =\{ x_1, x_2, x_3\}, \mathbb S = (\exN, +, \cdot, 0, 1), \At = \{\alpha, \beta\}, \Act = \{p_1, p_2, p_3\},$ $\Out = \{v\}$. We consider a \wgkat automaton $(X, \tau)$, whose pictorial representation is below on the left, while the definition of $\tau$ is below to the right.

			\begin{minipage}{0.45\textwidth}
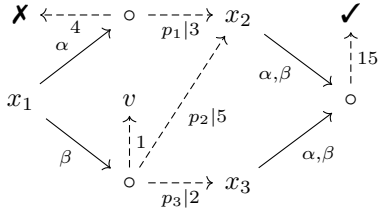
				\begin{tikzcd}
					\xmark & \circ & {x_2} & \cmark \\
					{x_1} & v && \circ \\
					& \circ & {x_3}
					\arrow["4", dashed, from=1-2, to=1-1]
					\arrow["{p_1 \mid 3}"', dashed, from=1-2, to=1-3]
					\arrow["{\alpha,\beta}"', from=1-3, to=2-4]
					\arrow["\alpha", from=2-1, to=1-2]
					\arrow["\beta"', from=2-1, to=3-2]
					\arrow["15"', dashed, from=2-4, to=1-4]
					\arrow["{p_2 \mid 5}"', dashed, from=3-2, to=1-3]
					\arrow["1"', dashed, from=3-2, to=2-2]
					\arrow["{ p_3 \mid 2}"', dashed, from=3-2, to=3-3]
					\arrow["{\alpha,\beta}"', from=3-3, to=2-4]
				\end{tikzcd}
			\end{minipage}
			\begin{minipage}{0.5\textwidth}
				\centering
				\begin{align*}
					\tau(x_1)_\alpha &= 4\delta_\xmark + 3\delta_{(p_1, x_2)}\\
					\tau(x_1)_\beta &= 1\delta_v + 5\delta_{(p_2, x_2)}+2\delta_{(p_3, x_3)}\\
					\tau(x_2)_\alpha &= \tau(x_2)_\beta = \tau(x_3)_\alpha = \tau(x_3)_\beta = 15\delta_\cmark
				\end{align*}
		\end{minipage}
	\end{exmp}

\subsection{Operational semantics of Expressions}
We are now ready to define the operational semantics of \wgkat expressions: we will construct a \wgkat automaton with transition function $\partial: \Exp \to \Mon(2+\Out+\Act\times \Exp)^\At$, whose states are \wgkat expressions. The semantics of an expression $e$, will be given by the behavior of the corresponding state in that automaton. This construction is reminiscent of the Brzozowski/Antimirov derivative construction for DFA and NFA respectively. For $\alpha \in \At$, we define $\partial(e)_\alpha$ by structural induction on $e$.
	
\subparagraph*{Tests, actions, output variables.} Suppose $b \in \BExp, v \in \Out, p\in \Act$, we let
	$$
		\partial(b)_{\alpha} = \begin{cases}\delta_{\cmark} &\alpha\leq_{\text{BA}}b\\
			\delta_{\xmark} &\alpha\leq_{\text{BA}}\bar {b}
		\end{cases}
		\qquad\partial(v)_\alpha = \delta_v \;\;\;\;
		\qquad\partial(p)_\alpha = \delta_{(p,\cmark)}			   
	$$
	A test either accepts or aborts with weight $1$ depending on whether the given atom entails the truth of the test. An output variable outputs itself with weight $1$ and then terminates. An action outputs the atomic program action and then accepts with weight $1$, allowing for further computation after the action has been taken. We note here that when we do sequential composition, we will rewire acceptance into the composed expression. For this reason $\delta_\cmark$ can also be thought of as skipping with weight $1$.
	
	\subparagraph*{Guarded Choice.} A guarded choice $+_b$ is an if-then-else statement conditioned on $b$. To capture this, we let $e+_b f$ have the outgoing edges of $e$ if $\alpha \leq_{\BA} b$ and $f$ otherwise. 
	$$
		\partial(e +_b f)_\alpha = 
		\begin{cases}
			\partial(e)_\alpha & \alpha \leq b\\
			\partial(f)_\alpha & \alpha \leq \bar b
		\end{cases}
	$$
	
	\subparagraph*{Weighted Choice.} The intuition here is that the automaton executes both arguments and scales their output by the given quantity. We define $e \WC{r}{s} f$ as the automaton with all outgoing edges of $e$, scaled by $r$, and all outgoing edges of $f$, scaled by $s$. The derivative is: $$\partial(e \WC{r}{s} f)_\alpha = r\cdot\partial(e)_\alpha + s\cdot\partial(f)_\alpha$$
	where these operations are the addition and scalar multiplication of the weightings.
	\subparagraph*{Sequential Composition.}
	We define the derivative of sequential composition as: 
	$$
		\partial(e;f)_\alpha = \partial(e)_\alpha \vartriangleleft_\alpha f
	$$
	where given $\alpha$ and $f$ we let $(\text{\textendash} \vartriangleleft_\alpha f):\Mon(2+\Out+\Act\times \Exp) \to \Mon(2+\Out+\Act\times \Exp)$  be the linearization of $c_{\alpha, f} \colon 2+\Out+\Act\times \Exp \to \Mon(2+\Out+\Act\times \Exp)$, given by: 
	$$
		c_{\alpha, f}(x) = \begin{cases}
			\partial(f)_\alpha & x=\cmark\\
			\delta_x & x\in\{\xmark\} \cup \Out\\
			\delta_{(p,e^\prime;f)}& x=(p,e^\prime)
		\end{cases}
	$$
	The intuition here is that if $e$ skips, then the behavior of the composition is the behavior of $f$. If $e$ aborts or returns a value then the derivative is the behavior of $e$. Finally, if $e$ executes an action and transitions to another state, then the behavior is the next step of $e$ composed with $f$ along with emitting the given action. 
	
	\subparagraph*{Guarded loops} A desired feature of the semantics of a program that performs loops is that each loop can be equivalently written as a guarded choice between acceptance and performing the loop body followed again by the loop. However, if the loop body could immediately accept, then the loop itself could non-productively accept any number of times before making a productive transition. 
	Each time the loop immediately accepts the computation is scaled again by the weight of acceptance. One can represent this unbounded accumulation of weights through the notion of a fixpoint. There is a well-studied class of semirings featuring such a construct, namely Conway semirings. We first recall the necessary definitions.
	 
	\begin{defn}[\cite{IterationSemirings}]
		A semiring $\mathbb S$ is a Conway Semiring if there exists a function $*:R \to R$ such that $\forall a,b \in R$\label{conwaySemiring}
		\begin{align}
				(a+b)^* &= a^*(ba^*)^* \label{denesting}\\
				(ab)^* &= 1+a(ba)^*b \label{ax2}
		\end{align}
	\end{defn} 
		
	We note that the  \emph{fixpoint rule} $aa^* +1 = a^*$ holds in all Conway Semirings  \cite[p.~6]{IterationSemirings}. 
		
	The properties of Conway semirings are enough to guarantee that we can meaningfully define semantics of loops.  We will see in \cref{axiomsS} that these additional identities, \Cref{denesting,ax2} in \cref{conwaySemiring},  play a crucial role in the proof of soundness.
	
	We now define the semantics of loops by iterating the weight of immediate acceptance: 
	$$
		\partial(e^{(b)})_\alpha(x) = 
		\begin{cases}
			1 &\quad x = \cmark \text{ and } \alpha \leq_{\BA} \bar b\\
			\partial(e)_\alpha(\cmark)^{*}\partial(e)_\alpha(x)&\quad x\in \{\xmark\} \cup \Out \text{ and } \alpha \leq_{\BA} b\\
			\partial(e)_\alpha(\cmark)^{*}\partial(e)_\alpha(p, e^\prime)&\quad x=(p, (e^\prime;e^{(b)})) \text{ and } \alpha \leq_{\BA} b\\
			0&\quad \text{otherwise}
		\end{cases}
	$$

	The loop could fail the Boolean test and then it just skips and does not execute the loop body (case 1). If the loop passes the test the body executes. Recall that the body might skip with some weight. So in this case the loop must either execute an action (case 3) or terminate (case 2) after some number of iterations. If the loop terminates (case 2) then we scale the weight of termination by the accumulated weight and terminate. If the loop takes an action (case 3), then we scale the weight of the action by the accumulated weight, emit the action, and compose the loop with the next step of $e$ to re-enter the loop once computation of $e$ has finished. We identify the behavior of divergent loops (case 4) with the \emph{zero weighting}, since they never take a productive action nor terminate. Such a weighting simply assigns the weight zero to all possible branches of the program. This is distinct from the program that asserts false, which immediately aborts with weight one.
	
	We now show that all expressions have a finite number of derivatives. This will be important for the completeness of our axiomatization and decidability of equivalence. We denote the set of states in the automaton reachable from expression $e$ as $\langle e\rangle_\partial$. 
	
	\begin{toappendix}
		$\langle e\rangle_\partial$ is the $\mathcal{H}$ subcoalgebra generated by $e$. This corresponds to the set of states which are reachable from state $e$ in the automaton.
	\end{toappendix}
	
	\begin{lemrep}
		For all $e \in \Exp$, $\left\langle e\right\rangle_\partial$ is finite. We bound it by $\#(e): \Exp \to \mathbb{N}$ where \label{lem7}
		\begin{align*}
			&\#(b) = 1,\;\;\;\#(v) = 1,\;\;\; \#(p) = 2,\;\;\; \#(e \WC{r}{s} f) = \#(e)+\#(f),\\
			&\#(e +_b f) = \#(e)+\#(f),\;\;\; \#(e;f) = \#(e)+\#(f),\;\;\; \#(e^{(b)}) = \#(e)
		\end{align*}
	\end{lemrep}
	\begin{appendixproof}
		We adapt the analogous proof for \probgkat \cite{rozowski2023probabilistic}. We set up the problem and show the cases for our new operators, but do not address the untouched cases. For any $e \in \Exp$ let $\left|\left\langle e\right\rangle_\partial\right|$ be the cardinality of the least subcoalgebra of $\Exp, \partial$ containing $e$. The proof that for all $e \in \Exp$ $\left|\left\langle e\right\rangle_\partial\right| \leq \#(e)$ is by induction.
		
		Assume that $\left|\left\langle e\right\rangle_\partial\right| \leq \#(e), \left|\left\langle f\right\rangle_\partial\right| \leq \#(f)$:
		
		$\#(e \WC{r}{s} f) = \#(e) +\#(f)$ because every derivative of $e \WC{r}{s} f$ is a derivative of $e$ or a derivative of $f$. Therefore $\left|\left\langle e\WC{r}{s}f \right\rangle_\partial\right| \leq \left|\left\langle e\right\rangle_\partial\right|+\left|\left\langle f \right\rangle_\partial\right| \leq \#(e) +\#(f)$.
		
		For sequential composition the behavior of $e;f$ is either termination, a derivative of $f$ or a derivative of $e$ follow by $f$. $|\langle e;f\rangle_\partial| = |\langle e\rangle_\partial \times \{f\}| + |\langle f \rangle_\partial| \leq \langle e\rangle_\partial + \langle f\rangle_\partial$.
		
		For the guarded loop, note that every derivative of $e^{(b)}$ is a derivative of $e$ composed with a derivative of $e^{(b)}$, in the action case in particular. Clearly then the number of derivatives is no greater than the number of derivatives of $e$. That is $\left|\left\langle e^{(b)}\right\rangle_\partial\right| \leq \left|\left\langle e\right\rangle_\partial\right| \leq \#(e)$.
	\end{appendixproof}
	\subsection{Characterizing Equivalence}\label{equivS}
	We will use {bisimilarity} as our notion of equivalence. Intuitively, two states of \wgkat automata are bisimilar if at each step of the continuing computation they are indistinguishable to an outside observer. This is a canonical notion of equivalence, that is more strict than language equivalence, and  will allow us to present an efficient decision procedure in \cref{decisionS}. We formalize what it means for states to be bisimilar in terms of automaton homomorphisms.
	\begin{defn}
		Given two \wgkat automata $(X, \beta), (Y, \gamma)$ a function $f: X\to Y$ is a \wgkat homomorphism if for all $x \in X$ and $\alpha \in \At$
		\begin{enumerate}
			\item For any $o \in 2+ \Out$\;\; $\gamma(f(x))_\alpha(o) = \beta(x)_\alpha(o)$
			\item For any $(p, y) \in \Act \times Y$\;\; $\gamma(f(x))_\alpha(p,y) = \beta(x)_\alpha[\{p\} \times f^{-1}(y)]$
		\end{enumerate}
	\end{defn}
	We concretely instantiate the definition of bisimulation \cite{Rutten2000UniversalCA} for our automata, for the general version see \cref{AppCoalg}.
	\begin{defn}
		Let $(X, \beta)$ and $(Y, \gamma)$ be \wgkat automata. A relation $R\subseteq X \times Y$ is a \textit{bisimulation} if there exists a transition function $\rho: R \to \Mon(2+\Out+\Act\times R)^{At}$ such that the projections $\pi_1: R\to X, \pi_2: R\to Y$ are \wgkat homomorphisms from $(R, \rho)$ to $(X, \beta)$ and $(Y, \gamma)$ respectively. We say the elements $x\in X, y\in Y$ are \textit{bisimilar} if there exists a bisimulation R such that $(x,y) \in R$. \label{bisimulationwatered}
	\end{defn}
	
	To be able to build the combined transition structure over $R$ in \cref{bisimulationwatered} we require some additional properties that enable us to refine the original weightings on each set to construct more granular weightings over the relation. We achieve this by requiring that the \textit{additive monoid} of the semiring has two properties: \emph{refinement}~\cite{refinementmonoids} and \emph{positivity}~\cite{gumm2009copower}.
	\begin{defn}\label{ref_pos_def}
		A monoid $M$ is a refinement monoid if $x+y=z+w$ implies that there exist  $s,t,u,v$ such that $s+t=x, s+u=z, u+v=y, t+v=w$. A monoid $M$ with identity $e$ is positive (conical) if for all $x,y\in M$, $x+y=e \implies x=e=y$.
	\end{defn} 
	 Refinement extends inductively to sums with any (finite) number of summands \cite{refinementmonoids}. Positivity says that no element other than the additive identity has an additive inverse. 
	 	\begin{toappendix}
		The reason that we need the semiring to be refinement and positive (\cref{ref_pos_def}) is that these two properties correspond precisely with the functor $\Mon$ preserving weak pullbacks.
		\begin{lem}
			$\Mon$ preserves weak pullbacks if and only if $\mathbb S$ is positive and refinement.
		\end{lem}
		\begin{proof}
			Since $\Mon$ is a set endofunctor it preserves weak pullbacks if and only if it weakly preserves pullbacks \cite[p.~20]{gumm2021freelatticefunctorsweaklypreserve}. Furthermore by \cite[Theorem 2.7]{typesandcoalg} $\Mon$ weakly preserves nonempty pullbacks if and only if it weakly preserves pullbacks. Finally, by \cite[Theorem 5.3]{GUMM2001185} the commutative monoid weighting functor $\Mon$ weakly preserves nonempty pullbacks if and only if the monoid is positive and refinement. The lemma follows immediately.
		\end{proof}
		\begin{lem}
			the functor $\mathcal{H}$ preserves weak pullbacks if the semiring over which weightings are taken is positive and refinement. \label{refinementIsPullbacks}
		\end{lem}
		\begin{proof}
			If $F$ and $G$ preserve weak pullbacks, then so do $F \circ G, F+G, F \times G$ \cite{PETERGUMM2000111}. Furthermore for any functor $F$: $F(X) = X^\Sigma$ for a fixed set $\Sigma$, $F(X) = X+A$, $F(X) = A$ for a fixed set $A$, and $\Id$, the identity functor, also preserve weak pullbacks \cite{PETERGUMM2000111}. Hence if $\Mon$ preserves weak pullbacks, then so does $\mathcal{H} = \Mon(2+\Out+\Act\times \Id)^{\At}$. $\Mon$ preserves weak pullbacks if $\mathbb{S}$ is positive and refinement. Hence $\mathcal H$ preserves weak pullbacks if $\mathbb{S}$ is positive and refinement.
		\end{proof}
	\end{toappendix}
	We now have a full characterization of the semirings we will consider: \emph{positive, Conway, and refinement}. We offer the following list of semirings as examples for interpretation due to their relevance to weighted programming \cite{batz2022weighted}, and assume that all semirings discussed from now on have these three properties.
	\begin{toappendix}
		\begin{lem}
			The following semirings are positive refinement monoids.
			\begin{multicols}{2}
				\begin{enumerate}[nosep]
					\item Tropical $(\exN, \min, +, \infty, 0)$
					\item Arctic  $(\mathbb N^{+\infty}_{-\infty}, \max, +, -\infty, 0)$
					\item Bottleneck  $(\mathbb R^{+\infty}_{-\infty}, \max, \min, -\infty, \infty)$
					\item Formal languages  $(2^{\Gamma^*}, \cup, \cdot, \emptyset, \{\varepsilon\})$
					\item Extended naturals  $(\mathbb \exN, +, \cdot, 0, 1)$
					\item Viterbi  $([0,1], \max, \cdot, 0, 1)$
					\item Boolean  $(\{0,1\}, \lor, \land, 0, 1)$
					\item[]
				\end{enumerate}
			\end{multicols}
			\begin{enumerate}[nosep]
				\setcounter{enumi}{7}
				\item Extended non-negative rationals $(\exQp, +, \cdot, 0, 1)$
			\end{enumerate}\label{posrefsemis}
		\end{lem}
		\begin{proof}
			We note that none of these refinements are necessarily unique nor do we attempt to prove uniqueness. We simply find one refinement in each semiring.
			\begin{enumerate}
				\item First: $\min(a,b) = \infty \implies a = \infty = b$. Suppose $\min(a,b) = \min(x,y)$, wlog suppose $a\geq b$ and $x\geq y$ (in the natural ordering on the semiring). Then we know $a=x$ and a refinement is:
				$$
				\begin{bmatrix}
					&\sum & \sum &\vline & \sum \\
					\sum&a & y& \vline & a \\
					\sum&b & \infty& \vline & b \\
					\hline
					\sum&x & y & \vline& & 
				\end{bmatrix}
				$$
				\item First: $\max(a,b) = -\infty \implies a = -\infty = b$. Suppose $\max(a,b) = \max(x,y)$, wlog suppose $a\geq b$ and $x\geq y$. Then we know $a=x$ and a refinement is:
				$$
				\begin{bmatrix}
					&\sum & \sum &\vline & \sum \\
					\sum&a & y& \vline & a \\
					\sum&b & -\infty& \vline & b \\
					\hline
					\sum&x & y & \vline& & 
				\end{bmatrix}
				$$
				\item This semiring is positive and refinement by the same argument as the Arctic semiring.
				\item First: $a \cup b = \emptyset \implies a = \emptyset = b$. Suppose $a \cup b = x\cup y$. We know that $a \subseteq x\cup y$ and the analogous fact is true for $b$, $x$, and $y$. Then the refinement is: 
				$$
				\begin{bmatrix}
					&\sum & \sum &\vline & \sum \\
					\sum&a\cap x & a \cap y& \vline & a \\
					\sum&b\cap x & b\cap y& \vline & b \\
					\hline
					\sum&x & y & \vline& & 
				\end{bmatrix}
				$$
				\item First: $a+b = 0 \implies a = 0 = b$. Suppose $a + b = x + y$. In the natural numbers we are guaranteed that if $a\geq b$ there exists some unique $c$ such that $b+c = a$. Let $a-b = c$ if $a\geq b$. We let $\textbf{m} \coloneq \min(a, x)$ The refinement is:
				$$
				\begin{bmatrix}
					&\sum & \sum &\vline & \sum \\
					\sum&\textbf{m} & \min(a-\textbf{m},y)& \vline & a \\
					\sum&\min(x-\textbf{m}, b) & \min(b-\min(x-\textbf{m}, b), y-\min(a-\textbf{m},y))& \vline & b \\
					\hline
					\sum&x & y & \vline& & 
				\end{bmatrix}
				$$
				Clearly this refinement is valid if $a=x$. wlog suppose $a > x$. Then $y > b$, as otherwise $a+b > x+y$ would hold. So then the refinement looks like: 
				$$
				\begin{bmatrix}
					&\sum & \sum &\vline & \sum \\
					\sum&x & \min(a-x,y)& \vline & a \\
					\sum&\min(0, b) & \min(b-\min(0, b), y-\min(a-x,y))& \vline & b \\
					\hline
					\sum&x & y & \vline& & 
				\end{bmatrix}
				$$
				$$
				\begin{bmatrix}
					&\sum & \sum &\vline & \sum \\
					\sum&x & \min(a-x,y)& \vline & a \\
					\sum&0 & \min(b, y-\min(a-x,y))& \vline & b \\
					\hline
					\sum&x & y & \vline& & 
				\end{bmatrix}
				$$
				Note that $a-x>y \implies a>x+y$ which contradicts the equality of the two original sums. Hence we can simplify to: 
				$$
				\begin{bmatrix}
					&\sum & \sum &\vline & \sum \\
					\sum&x & a-x& \vline & a \\
					\sum&0 & \min(b, y-(a-x))& \vline & b \\
					\hline
					\sum&x & y & \vline& & 
				\end{bmatrix}
				$$
				Note again that $y-(a-x) > b \implies y> b+(a-x) \implies y+x > b+a$ which contradicts our assumption. Hence we simplify again:
				$$
				\begin{bmatrix}
					&\sum & \sum &\vline & \sum \\
					\sum&x & a-x& \vline & a \\
					\sum&0 & b& \vline & b \\
					\hline
					\sum&x & y & \vline& & 
				\end{bmatrix}
				$$
				Note that $a+b=x+y \implies y = (a+b)-x = (a-x)+b$. So this refinement is valid.
				\item This proof is essentially the same as the arctic semiring.
				\item First: $a \lor b = 0 \implies a = 0 = b$. Suppose $a \lor b = x \lor y$. wlog suppose $a \geq b$, $x \geq y$. Note that $a = 1 \iff x=1$ and $a = 0 \iff x = 0$ and $a = 0 \implies b=0$ as well as $x=0 \implies y=0$. Hence the refinement looks like:
				$$
				\begin{bmatrix}
					&\sum & \sum &\vline & \sum \\
					\sum&a & y& \vline & a \\
					\sum&b & 0& \vline & b \\
					\hline
					\sum&x & y & \vline& & \rlap{\qedhere} 
				\end{bmatrix}
				$$
				\item This proof is essentially the same as the natural numbers.
			\end{enumerate}
		\end{proof}
		
		\begin{lem}
			Each of the following semirings is a Conway semiring. \label{conawysemis}
			\begin{multicols}{2}
				\begin{enumerate}[nosep]
					\item Tropical $(\exN, \min, +, \infty, 0)$
					\item Arctic  $(\mathbb N^{+\infty}_{-\infty}, \max, +, -\infty, 0)$
					\item Bottleneck  $(\mathbb R^{+\infty}_{-\infty}, \max, \min, -\infty, \infty)$
					\item Formal languages  $(2^{\Gamma^*}, \cup, \cdot, \emptyset, \{\varepsilon\})$
					\item Extended naturals  $(\mathbb \exN, +, \cdot, 0, 1)$
					\item Viterbi  $([0,1], \max, \cdot, 0, 1)$
					\item Boolean  $(\{0,1\}, \lor, \land, 0, 1)$
					\item[]
				\end{enumerate}
			\end{multicols}
			\begin{enumerate}[nosep]
				\setcounter{enumi}{7}
				\item Extended non-negative rationals $(\exQp, +, \cdot, 0, 1)$
			\end{enumerate}
		\end{lem}
		\begin{proof}
			\begin{enumerate}
				Semirings 1, 2, 4, 5, 7, 8 are known to be Conway Semirings \cite[p. 8]{weightedhandbook}

				For the Bottleneck semiring
				We define $\forall a \; a^* = \infty$. Then: 
				\begin{align*}
					(a+b)^* &= \infty = \infty\infty = a^*(ba^*)^*\\
					(ab)^* &= \infty = \infty+a(ba)^*b = 1_S+a(ba)^*b
				\end{align*}

				For the Viterbi semiring
				Let $\forall a\; a^* = 1$.
				\begin{align*}
					 (a+b)^* &= 1 = 11 = a^*(ba^*)^*\\
					 (ab)^* &= 1 = 1+a(ba)^*b = 1_S+a(ba)^*b&\rlap{\qedhere} 
				\end{align*}
			\end{enumerate}
		\end{proof}
	\end{toappendix}
	\begin{correp}
		The following semirings are positive, Conway, and refinement.
		\begin{multicols}{2}
			\begin{enumerate}[nosep]
				\item Tropical $(\exN, \min, +, \infty, 0)$
				\item Arctic  $(\mathbb N^{+\infty}_{-\infty}, \max, +, -\infty, 0)$
				\item Bottleneck  $(\mathbb R^{+\infty}_{-\infty}, \max, \min, -\infty, \infty)$
				\item Formal languages  $(2^{\Gamma^*}, \cup, \cdot, \emptyset, \{\varepsilon\})$
				\item Extended naturals  $(\mathbb \exN, +, \cdot, 0, 1)$
				\item Viterbi  $([0,1], \max, \cdot, 0, 1)$
				\item Boolean  $(\{0,1\}, \lor, \land, 0, 1)$
				\item[]
			\end{enumerate}
		\end{multicols}
		\begin{enumerate}[nosep]
			\setcounter{enumi}{7}
				\item Extended non-negative rationals $(\exQp, +, \cdot, 0, 1)$
		\end{enumerate}
	\end{correp}
	\begin{appendixproof}
		This follows from \cref{posrefsemis} and \cref{conawysemis}.
	\end{appendixproof}
	
	\begin{toappendix}
		With these three properties, positive, Conway, and refinement, in hand we can characterize $\mathcal H$-coalgebras.
		
		\begin{prop}
			There exists an $\mathcal{H}$-coalgebra $(Z, \zeta)$ which is final in Coalg$_{\mathcal{H}}$. Equivalently, For any $\mathcal{H}$-coalgebra $(X, \beta)$ there exists a unique homomorphism $!_{\beta}:X\to Z$ 
		\end{prop}
		\begin{proof}
			We first argue that $\Mon$ is bounded using the same argument as finitely supported probability distributions \cite{DEVINK}. Let $\langle s \rangle$ be the smallest subcoalgebra of $(X, \beta)$ containing $s$. It is constructed by repeatedly taking the set $\{\{s\} \cup \{s^\prime \mid \exists \alpha, p \; \beta(s)_\alpha(p, s^\prime) \neq 0 \}\}$. Note that since the support is finite, we have added finitely many elements a countable number of times. Hence the subcoalgebra is bounded by the cardinal number $\omega$. Boundedness is preserved under functor composition, binary products and binary coproducts \cite[Corollary 4.9]{gumm2002coalgebras}. Since $\At$ is finite $\Id^{\At}$ will preserve boundedness. Hence $\mathcal{H}$ is bounded and preserves weak pullbacks, so the existence of the final coalgebra follows from \cite[Theorem 10.3]{Rutten2000UniversalCA}.
		\end{proof}
		When it is clear from context we will omit the subscript on the map into the final coalgebra.
		\begin{prop}
			Let $(X, \beta)$ and $(Y, \gamma)$ be $\mathcal{H}$-coalgebras. The elements $x \in X, y \in Y$ are bisimilar iff $!_\beta(x) = !_\gamma(y)$. \label{bisimIffBeh}
		\end{prop}
		\begin{proof}
			$\mathcal{H}$ preserves weak pullbacks and as such by \cite[Theorem 9.3]{Rutten2000UniversalCA} $x$ and $y$ are bisimilar if and only if $!_\beta(x) = !_\gamma(y)$.
		\end{proof}
		The condition $!_\beta(x) = !_\gamma(y)$ is called behavioral equivalence.
	\end{toappendix} 
	
	\section{Axioms}\label{axiomsS}
	In this section, we present an axiomatization of bisimulation equivalence for \wgkat expressions (\Cref{axiomsFig}). We also derive some facts from the axioms which will be useful for both the examples and the proof of completeness (\Cref{sec:derivable}). 
	
	To explain the axioms in \Cref{axiomsFig}, we first need to define a property similar to Salomaa's empty word property \cite{salomaa}, which will be a familiar side condition for a fixpoint axiom.	
	\begin{defn}
		We define $E\colon \Exp \to \At \to \mathbb S$ inductively as follows
		
		\begin{minipage}[t]{0.45\textwidth}\vspace{-.65cm}
			\begin{align*}
				E(p)_\alpha &= E(v)_\alpha = 0\\
				E(b)_\alpha &= \begin{cases}1 &\alpha \leq_{\BA} b\\ 0 & \alpha \leq_{\BA} \bar b \end{cases}\\
				E(e \WC{r}{s} f)_\alpha &= rE(e)_\alpha+sE(f)_\alpha
			\end{align*}
		\end{minipage}
		\begin{minipage}[t]{0.45\textwidth}\vspace{-.65cm}
			\begin{align*}
				E(e +_b f)_\alpha &= \begin{cases}E(e)_\alpha &\alpha \leq_{\BA} b\\ E(f)_\alpha & \alpha \leq_{\BA} \bar b \end{cases}\\
				E(e;f)_\alpha &= E(e)_\alpha E(f)_\alpha\\
				E(e^{(b)})_\alpha &= E(\bar b)_\alpha
			\end{align*}
		\end{minipage}
	\end{defn}
	\begin{lemrep}
		Let $e \in \Exp, \alpha \in \At$. It holds that $E(e)_\alpha = \partial(e)_\alpha(\cmark)$\label{lem44}.
	\end{lemrep}
	\begin{appendixproof}
		We prove this by induction of the construction of $e$. The base cases hold immediately.
		\begin{align*}
			E(e\WC{r}{s}f) &= rE(e)_\alpha + sE(f)_\alpha& \\
			&= r\partial(e)_\alpha(\cmark) + s\partial(f)_\alpha(\cmark) & \text{Inductive hypothesis}\\
			&= \partial(e \WC{r}{s} f)_\alpha(\cmark) &
		\end{align*}\begin{align*}
			E(e +_b f) &= \begin{cases}E(e)_\alpha &\alpha \leq_{\BA} b\\ E(f)_\alpha & \alpha \leq_{\BA} \bar b \end{cases}&\\
			&= \begin{cases}\partial(e)_\alpha(\cmark) &\alpha \leq_{\BA} b\\ \partial(f)_\alpha(\cmark) & \alpha \leq_{\BA} \bar b \end{cases}& \text{Inductive hypothesis}\\
			&= \partial(e +_b f)_\alpha(\cmark) & 
		\end{align*}
		\begin{align*}
			E(e;f) &= E(e)_\alpha;E(f)_\alpha& \\
			&= \partial(e)_\alpha(\cmark)\partial(f)_\alpha(\cmark) & \text{Inductive hypothesis}\\
			&= \partial(e;f)_\alpha(\cmark) &
		\end{align*}
		\begin{align*}
			E(e^{(b)}) &= \begin{cases}1 & \alpha \leq_{\BA} \bar b\\0 & \alpha \leq_{\BA} b\end{cases}& \\
			&= \begin{cases}\partial(e^{(b)})_\alpha(\cmark) & \alpha \leq_{\BA} \bar b\\\partial(e^{(b)})_\alpha(\cmark) & \alpha \leq_{\BA} b\end{cases}& \\
			&= \partial(e^{(b)})_\alpha(\cmark)&\rlap{\qedhere} 
		\end{align*}
	\end{appendixproof}
	\begin{figure}[t!]\vspace{-.8cm}
		\begin{equation*}
			\begin{aligned}[t]
				&\textbf{\underline{Guarded Choice Axioms\phantom{q}}} \\
				\Ax{(G1)}\; &  e +_b e \equiv e \\
				\Ax{(G2)}\; &  e+_b f\equiv b;e+_{b}f\\
				\Ax{(G3)}\; &  e +_b f \equiv f+_{\bar b} e \\
				\Ax{(G4)}\; &  (e+_b f)+_cg\equiv e+_{bc}(f+_cg) \\
				&\textbf{\underline{Distributivity Axioms}} \\
				\Ax{(D1)}\; &  e\WC{r}{s} (f+_b g) \equiv (e \WC{r}{s} f) +_b (e \WC{r}{s} g) \\
				\Ax{(D2)}\; &  e\WC{r}{s} (f \WC{t}{u} g) \equiv e\WC{r}{1} (f \WC{st}{su} g) \\
				\Ax{(D3)}\; &  b;(e \WC{r}{s} f) \equiv b;(b;e \WC{r}{s} b;f)\\
				&\textbf{\underline{Sequencing Axioms}} \\
				\Ax{(S1)}\; &  \bskip; e \equiv e \equiv e; \bskip \\
				\Ax{(S2)}\; &  (e; f); g \equiv e; (f; g) \\
				\Ax{(S3)}\; &  \babort ; e\equiv \babort\\
				\Ax{(S4)}\; &  (e\WC{r}{s}f);g\equiv e;g \WC{r}{s} f;g \\
				\Ax{(S5)}\; &  (e+_bf);g\equiv e;g +_b f;g\\
				\Ax{(S6)}\; &  v ; e\equiv v \\
				\Ax{(S7)}\; &  b;c\equiv bc
			\end{aligned}\quad
			\begin{aligned}[t]
				&\textbf{\underline{Loop Axioms}} \\
				\Ax{(L1)}\; &  e^{(b)} \equiv e ; e^{(b)} +_b \bskip \\[2ex]
				\Ax{(L2)}\; &  \frac{e\equiv (f \WC{r}{s} \bskip)+_c g}{c;e^{(b)}\equiv c;((\odot s^*r ;f;(e)^{(b)})+_b \bskip)}\\[2ex]	
				&\textbf{\underline{Scaling Axioms}} \\
				\Ax{(C1)}\; &  \odot 1 \equiv \bskip\\
				\Ax{(C2)}\; &  \odot 0;e \equiv \odot 0\\ \\
				&\textbf{\underline{Weighted Choice Axioms}} \\
				\Ax{(W1)}\; &  e \WC{r}{s} e \equiv \odot (r+s);e \\
				\Ax{(W2)}\; &  e \WC{r}{s} f \equiv f \WC{s}{r} e \\
				\Ax{(W3)}\; &  e \WC{r}{s} (f \WC{t}{u} g ) \equiv (e \WC{r}{st} f) \WC{1}{su} g  \\ 
				\Ax{(W4)}\; &  e \WC{ru}{s} f \equiv (\odot u;e) \WC{r}{s} f \\[1.5ex]
				&\textbf{\underline{Fixpoint Rule} }   \\
				\Ax{(F1)}\; & \frac{ g\equiv e;g +_b f \;\;\; \forall \alpha\in\At\; E(e)_\alpha = 0}{g \equiv e^{(b)}; f}
			\end{aligned}
		\end{equation*}\vspace{-.45cm}
		\caption{The axiomatization of \wgkat bisimulation equivalence. In these statements we let $e,f,g \in \Exp, b,c \in \BExp, v\in\Out, \alpha \in \At, p \in \Act,\;  r,s,t,u \in \textbf{S}$. $1$ and $0$ weights refer to the multiplicative and additive identity in the chosen semiring.}
		\label{axiomsFig}
	\end{figure}
	\subparagraph*{Axiom Summary}
	
	We define ${\equiv} \subseteq {\Exp \times \Exp}$ as the smallest congruence satisfying the  axioms in \cref{axiomsFig}. Axioms \Ax{G1}-\Ax{G4} are inherited from \gkat \cite{Smolka_2019} and describe guarded choice. Axioms \Ax{D1} and \Ax{D2} state that weighted choice distributes over guarded choice and another weighted choice, respectively. \Ax{D3} describes the interaction of tests and weighted choice.

	Axioms \Ax{S1}-\Ax{S7} characterize sequential composition. \Ax{S4} is perhaps the most interesting; it states that composition right distributes over weighted choice just like guarded choice (\Ax{S5}). Since we equate the programs with respect to bisimilarity, which is a branching-time notion  of equivalence, the symmetric rules of left distributivity are not sound. \Ax{S1} and \Ax{S2} state that \wgkat expressions have the structure of a monoid with identity $\bskip$, and with $v \in \Out$ and $\babort$ being absorbent elements when composed on the left (\Ax{S6}, \Ax{S3}).
	
	\Ax{C1}-\Ax{C2} are about scaling. Scaling by the multiplicative identity does nothing, it is the same as skip. Scaling by $0$ annihilates any following expressions as it multiplies the weights by the multiplicative annihilator of the semiring. Though note that it is not the same as abort ($\babort$), since scaling by $0$ assigns weight $0$ to $\xmark$, while $\babort$ assigns weight $1$ to $\xmark$. In this way both expressions annihilate any following expressions, but are not the same.
	
	\Ax{W1}-\Ax{W4} describe weighted branching. Identical branches can be combined, \Ax{W1}. The order of branches doesn't matter, \Ax{W2}. A repeated branching can be re-nested by re-weighting, \Ax{W3}. Finally, weights can be pushed through the branching into the following expression, \Ax{W4}.
	
	\begin{rem}
		We have made a choice to present some of our axioms using the syntactic sugar $\odot$. The axiomatization could be equivalently expressed without it. For example, if we were to express \Ax{C1} without the syntactic sugar it would state $(\bskip \WC{1}{0} \babort);\bskip \equiv \bskip$. Stated explicitly, this axiom offers a way to remove weighted choices from expressions. Note that it is not derivable from other axioms of weighted choice as it is the only axiom with a weighted choice on the left-hand side but not on the right. The use of $\odot$ improves the readability and usability of the axioms. This is because it provides a compact representation of weighted choices with branches scaled by $0$ (which are irrelevant for the behavior).
	\end{rem}
	
	Finally we discuss the loop axioms. \Ax{F1} is the fixpoint rule. It states that fixpoints for productive loops are unique. We capture the notion of productivity of expressions, using an auxiliary operator $E$. \Ax{L1} is standard loop unrolling. A loop body is executed no times or one or more depending on the test, so it is the same as making a guarded choice between skipping or executing the loop body one or more times. \Ax{L2} allows unrolling and flattening the first set of iterations of a weighted loop if one of the branches is non-productive. The idea is that the loop can execute the nonproductive branch any number of times, but eventually must take the productive branch, after which it reenters the loop. Since a program action has now been taken, the boolean test $c$ could now fail. Hence when the loop body is re-entered it is of the original form. Intuitively the loop has been unrolled up to the first time it executes the productive branch and this nonproductive unrolled portion has been flattened into one scaling followed by the productive program action and then the original loop. 
	\begin{toappendix}
		We show that $\equiv$ is a bisimulation equivalence by first showing that any pair ${(e,f)} \in {\equiv}$ is behaviorally equivalent. Since $\mathcal H$ preserves weak pullbacks this  also implies that they are bisimilar by \cref{bisimIffBeh}, making $\equiv$ a bisimulation equivalence.
		
		Our proof of soundness via behavioral equivalence is inspired by process algebra \cite{ToddThesis}. We will first show that the congruence $\equiv$ is the kernel of a morphism, which will then allow us to express the morphisms to the final coalgebra as composite morphisms of the canonical quotient morphism and the morphism from the quotient coalgebra to the final coalgebra.
		
		First we must introduce some notation.
		\begin{defn}[\cite{rozowski2023probabilistic}]
			Let $A \subseteq \Exp$ and $f \in \Exp$, then $A/f = \{g \in \Exp \mid g;f \in A\}$.
		\end{defn}
		Furthermore we modify some lemmas from \probgkat about this piece of notation which remain sound with our new model. Their proofs are unchanged in the new model, so we omit them. 
		\begin{lem}[Lemma 41, \cite{rozowski2023probabilistic}]
			If $R \subseteq \Exp \times \Exp$ is a congruence relation with respect to \wgkat operators and $(e,f) \in R$ then $R(A/e) \subseteq R(A)/f$.
			\label{lem41}
		\end{lem}
		\begin{lem}[Lemma 42, \cite{rozowski2023probabilistic}]
			For all $\alpha \in \At$, $r,f \in \Exp$, $p \in \Act$, $A \subseteq \Exp$
			$$\partial(e;f)_\alpha[\{p\} \times A] = \partial(e)_\alpha[\{p\} \times A/f] + \partial(e)_\alpha(\cmark)\partial(f)_\alpha[\{p\}\times A]$$
			\label{lem42}
		\end{lem}
		\begin{lem}[Lemma 45, \cite{rozowski2023probabilistic}] \label{lem45}
			Let $R \subseteq \Exp \times \Exp$ be a congruence with respect to \wgkat operators such that $(e,f) \in R$, and let $Q \in \Exp/R$ be an equivalence class of $R$. It holds that $Q/e = Q/f$.
		\end{lem}
		\begin{lem}[Lemma 46, \cite{rozowski2023probabilistic}]
			Let $e,f \in \Exp$ and let $Q \in \Exp/\equiv$. Then $(Q/f)/e = Q/e;f$.  \label{lem46}
		\end{lem}
		\begin{lem}[Lemma 48, \cite{rozowski2023probabilistic}]
			Let $e \in \Exp$ and $b \in \BExp$. If $a \leq_{\BA} b$, then $\partial(b;e)_\alpha = \partial(e)_\alpha$. \label{lem48}
		\end{lem}
		
		Before proving soundness we recall two relevant properties of Conway semirings.
		\begin{lem}[\cite{IterationSemirings}] \label{fixpointrule}
			If $\mathbb S$ is a Conway semiring, the following equation holds in $\mathbb S$
			$$
				a^* = 1+aa^*
			$$ 
		\end{lem}
		\begin{lem}[\cite{IterationSemirings}] \label{slidingrulelem}
			If $\mathbb S$ is a Conway semiring, the following equation holds in $\mathbb S$
			$$
				(ab)^*a = a(ba)^*
			$$
		\end{lem}
		
		While Conway semirings are not the most general axiomatization of a semiring with a $^*$ operator, we need the $^*$ operator to obey \cref{denesting} in \cref{conwaySemiring}, and \cref{slidingrulelem} in addition to the fixpoint rule, \cref{fixpointrule}. We note that \ref{ax2} in \cref{conwaySemiring} is easily derivable from \cref{slidingrulelem} and the fixpoint rule so this class of semiring is not too strong.
		\begin{lem}
			If a semiring has a total operator $^*:R \to R$ with the properties $\forall a,b\;\; (ab)^*a = a(ba)^*$ and $\forall a\;\; a^* = 1+aa^*$, then $\forall a,b\;\; (ab)^* = 1+a(ba)^*b$.
		\end{lem}
		\begin{proof}
			\begin{align*}
				(ab)^* &= 1+ab(ab)^* & \cref{fixpointrule}\\
				&= 1+a(ba)^*b & \cref{slidingrulelem}\rlap{\qedhere} 
			\end{align*}
		\end{proof}
		Now we can prove the existence of the necessary homomorphism.
		\begin{lem}
			$\equiv$ is the kernel of a coalgebra homomorphism.\label{SoundnessMorphism}
		\end{lem}
		\begin{proof}
			This proof is inspired by the proof of \cite[lemma 4.4.3]{ToddThesis}. Consider the following commutative diagram
			\begin{center}\begin{tikzcd}
				\Exp && {\Exp/\equiv} \\
				\\
				{\mathcal{H}\Exp} && {\mathcal{H}\Exp/\equiv}
				\arrow["{[-]_\equiv}", from=1-1, to=1-3]
				\arrow["\partial"', from=1-1, to=3-1]
				\arrow["{\partial_\equiv}", dashed, from=1-3, to=3-3]
				\arrow["{\mathcal{H}[-]_\equiv}"', from=3-1, to=3-3]
			\end{tikzcd}\end{center}
			
			We want to show that there exists a morphism $\partial_\equiv$. Note that $[-]_\equiv$ is a surjection. We define $\partial_\equiv$ as $\partial_\equiv([e]_\equiv) = \mathcal{H}[-]_\equiv \partial(e)$ for any $e \in \Exp$. This is well defined if and only if $\text{ker}([-]_\equiv)\subseteq \text{ker}(\mathcal{H}[-]_\equiv\circ\partial)$ \cite[Lemma 3.17]{gumm1999elements}. Clearly $\equiv$ is the kernel of $[-]_\equiv$. So we show by induction $e \equiv f \implies \mathcal{H}[-]_\equiv\partial(e) = \mathcal{H}[-]_\equiv\partial(f)$.
			
			Consider the functor $W = \Mon(2+\Out+\Act\times \Exp)$, since $\At$ is finite $\mathcal{H} = \prod_{\alpha \in \At} W$. Therefore $(\mathcal{H}[-]_\equiv\circ\partial)(e) = (\mathcal{H}[-]_\equiv\circ\partial)(f) \iff \forall \alpha \in \At \;\; (W[-]_\equiv \circ \partial_\alpha)(e) = (W[-]_\equiv \circ \partial_\alpha)(f)$. Furthermore, $W[-]_\equiv$ is a semimodule homomorphism from $W\Exp \to \Exp/\equiv$. Consider $\nu_1, \nu_2 \in W\Exp, Q \in \Exp/\equiv$.
			\begin{align*}
				(W[-]_\equiv \nu_1) Q + (W[-]_\equiv \nu_2) Q &= \sum_{e \in \Exp \mid e \in Q} \nu_1(e) + \sum_{e \in \Exp \mid e \in Q} \nu_2(e)\\ 
				&= \sum_{e \in \Exp \mid e \in Q} \nu_1(e)+\nu_2(e)\\
				&= W[-]_\equiv (\nu_1+\nu_2) Q
			\end{align*}
			and consider a $\nu \in W\Exp, r \in \mathbb S, Q \in \Exp/\equiv$
			\begin{align*}
				(W[-]_\equiv r\nu) Q &= \sum_{e \in \Exp \mid e \in Q} r\nu(e)\\ 
				&= r\sum_{e \in \Exp \mid e \in Q} \nu(e)\\
				&= (r W[-]_\equiv \nu) Q
			\end{align*}

			We now examine the congruence part of the relation, then we examine the axioms.
			
			\subparagraph*{Case of: $\WC{}{}$}
			Suppose $e \equiv g, f\equiv h$, then $e \WC{r}{s} f \equiv g \WC{r}{s} h$, $\forall \alpha \in \At$  
			\begin{align*}
			 W[-]_\equiv \partial(e \WC{r}{s} f)_\alpha &=  W[-]_\equiv (r\partial(e)_\alpha + s\partial(f)_\alpha)\\
			 &= r\cdot W[-]_\equiv\partial(e)_\alpha+ s\cdot W[-]_\equiv\partial(f)_\alpha &\\
				&= r\cdot W[-]_\equiv\partial(g)_\alpha + s\cdot W[-]_\equiv\partial(h)_\alpha & e \equiv f, g \equiv h\\
				&= W[-]_\equiv(\partial (g \WC{r}{s} h)_\alpha)&
			\end{align*}

			\subparagraph*{Case of: $+_b$}
			Suppose $e \equiv g, f\equiv h$, then $e +_b f \equiv g +_b h$. For all $\alpha \in \At$\begin{align*}
				W[-]_\equiv\partial(e+_b f) &= W[-]_\equiv\begin{cases}
					\partial(e)_\alpha & \alpha \leq_{\BA} b\\
					\partial(f)_\alpha & \alpha \leq_{\BA} \bar b
				\end{cases}\\
				&= \begin{cases}
					W[-]_\equiv\partial(e)_\alpha & \alpha \leq_{\BA} b\\
					W[-]_\equiv\partial(f)_\alpha & \alpha \leq_{\BA} \bar b
				\end{cases}  \\
				&= \begin{cases}
					W[-]_\equiv\partial(g)_\alpha & \alpha \leq_{\BA} b\\
					W[-]_\equiv\partial(h)_\alpha & \alpha \leq_{\BA} \bar b
				\end{cases} & e \equiv g, f \equiv h\\
				&= W[-]_\equiv\begin{cases}
					\partial(g)_\alpha & \alpha \leq_{\BA} b\\
					\partial(h)_\alpha & \alpha \leq_{\BA} \bar b
				\end{cases} \\
				 &=W[-]_\equiv\partial(e+_b f)
			\end{align*}
			
			\subparagraph*{Case of: $;$}
			Suppose $e \equiv g, f\equiv h$, then $e;f \equiv g;h$. Before we begin, note that if $e \equiv g$ then $W[-]_\equiv\delta_{(p,e^\prime;e)} = \delta_{(p, [e^\prime;e]_\equiv)} = \delta_{(p, [e^\prime;g]_\equiv)} = W[-]_\equiv\delta_{(p,e^\prime;g)}$. As the fact that $\equiv$ is a congruence means that $[e^\prime;e]_\equiv = [e^\prime;g]_\equiv$. For all $\alpha \in \At$ \begin{align*}
				&W[-]_\equiv(\partial (e;f)_\alpha) = W[-]_\equiv(\partial(e)_\alpha \vartriangleleft_\alpha f)\\
				&= W[-]_\equiv\left(\sum_{x \in 2+\Out+\Act\times \Exp} \partial(e)_\alpha(x)c_{\alpha, f}(x)\right)\\
				&= \sum_{x \in 2+\Out+\Act\times \Exp} W[-]_\equiv\partial(e)_\alpha(x)W[-]_\equiv c_{\alpha, f}(x)&\\
				&=\quad\sum_{\mathclap{\substack{x \in 2+\Out+\Act\times \Exp}}} W[-]_\equiv\partial(e)_\alpha(x)
				\left(\begin{cases}
					W[-]_\equiv\partial(f)_\alpha & x=\cmark\\
					W[-]_\equiv\delta_x & x\in\{\xmark\} \cup \Out\\
					W[-]_\equiv\delta_{(p,e^\prime;f)}& x=(p,e^\prime)
				\end{cases}\right)(x)&\text{def of c}\\
				&=\quad\sum_{\mathclap{\substack{x \in 2+\Out+\Act\times \Exp}}} W[-]_\equiv\partial(g)_\alpha(x)
				\left(\begin{cases}
					W[-]_\equiv\partial(h)_\alpha & x=\cmark\\
					W[-]_\equiv\delta_x & x\in\{\xmark\} \cup \Out\\
					W[-]_\equiv\delta_{(p,e^\prime;h)}& x=(p,e^\prime)
				\end{cases}\right)(x)&e \equiv g, f \equiv h\\\\
				&= \sum_{x \in 2+\Out+\Act\times \Exp} W[-]_\equiv\partial(g)_\alpha(x)W[-]_\equiv c_{\alpha, h}(x)&\text{def of c}\\
				&= W[-]_\equiv(\partial(g)_\alpha \vartriangleleft_\alpha h) &\\
				&= W[-]_\equiv(\partial (g;h)_\alpha)
			\end{align*}
			
			\subparagraph*{Case of: $^{(b)}$}
			Suppose $e\equiv f$, then $e^{(b)} \equiv f^{(b)}$. For all $\alpha \in \At$ 
			\begin{align*}
				&W[-]_\equiv\partial (e^{(b)})_\alpha \\ 
				&= W[-]_\equiv \begin{cases}
					\bskip & \alpha \leq_{\BA} \bar b\\
					\partial(e)_\alpha(\cmark)^{*}\partial(e)_\alpha(x)& x\in \{\xmark\} \cup \Out, \alpha \leq_{\BA} b\\
					\partial(e)_\alpha(\cmark)^{*}\partial(e)_\alpha(p, e^\prime)& x=(p, (e^\prime;e^{(b)})), \alpha \leq_{\BA} b\\
					\babort& \text{Otherwise}\\
				\end{cases}\\
				&= \begin{cases}
					W[-]_\equiv\bskip & \alpha \leq_{\BA} \bar b\\
					W[-]_\equiv\partial(e)_\alpha(\cmark)^{*}\cdot W[-]_\equiv\partial(e)_\alpha(x)& x\in \{\xmark\} \cup \Out, \alpha \leq_{\BA} b\\
					W[-]_\equiv\partial(e)_\alpha(\cmark)^{*}\cdot W[-]_\equiv\partial(e)_\alpha(p, e^\prime)& x=(p, e^\prime;e^{(b)}), \alpha \leq_{\BA} b\\
					W[-]_\equiv\babort& \text{Otherwise}\\
				\end{cases} & \\
				&= \begin{cases}
					W[-]_\equiv\bskip & \alpha \leq_{\BA} \bar b\\
					W[-]_\equiv\partial(f)_\alpha(\cmark)^{*}\cdot W[-]_\equiv\partial(f)_\alpha(x)& x\in \{\xmark\} \cup \Out, \alpha \leq_{\BA} b\\
					W[-]_\equiv\partial(f)_\alpha(\cmark)^{*}\cdot W[-]_\equiv\partial(f)_\alpha(p, e^\prime)& x=(p, (e^\prime;f^{(b)})), \alpha \leq_{\BA} b\\
					W[-]_\equiv\babort& \text{Otherwise}\\
				\end{cases} & e \equiv f, e^b \equiv f^b\\
				&= W[-]_\equiv\begin{cases}
					\bskip & \alpha \leq_{\BA} \bar b\\
					\partial(f)_\alpha(\cmark)^{*}\partial(f)_\alpha(x)& x\in \{\xmark\} \cup \Out, \alpha \leq_{\BA} b\\
					\partial(f)_\alpha(\cmark)^{*}\partial(f)_\alpha(p, e^\prime)& x=(p, (e^\prime;f^{(b)})), \alpha \leq_{\BA} b\\
					\babort& \text{Otherwise}\\
				\end{cases} &\\
				&=W[-]_\equiv\partial (f^{(b)})\alpha & \text{def of }e^{(b)}
			\end{align*}
			Now we address the cases for each axiom. For each axiom we want to show that for any $\alpha$ if $e \equiv f$ by that axiom that $W[-]_\equiv (\partial(e)_\alpha) = W[-]_\equiv (\partial(f)_\alpha)$. We note that in some cases we will simply argue that $\forall \alpha \;\; \partial(e)_\alpha = \partial(f)_\alpha$ which clearly implies equivalence under the lifted quotient map. We also note that $W[-]_\equiv \partial_\alpha(e)(p, Q) = \partial_\alpha(e)[\{p\}\times Q]$ where $\alpha \in \At, e \in \Exp, p \in \Act, Q \in \Exp/\equiv$. So we examine the final cases with the latter syntax for the sake of readability.
			
			Then we can conclude from $\forall \alpha \in \At \;\; W[-]_\equiv (\partial(e)_\alpha) = W[-]_\equiv (\partial(f)_\alpha)$ that $\mathcal{H}[-]_\equiv\partial(e) = \mathcal{H}[-]_\equiv\partial(f)$.
			\begin{enumerate}
				\item[\Ax{G1}]
				\begin{align*}
					\forall \alpha \in \At\;\;\partial(e+_b e)_\alpha &= 
					\begin{cases}
						\partial(e)_\alpha & \alpha \leq_{\BA} b\\
						\partial(e)_\alpha & \alpha \leq_{\BA} \bar b
					\end{cases}\\
					&= \partial(e)_\alpha
				\end{align*}
				\item[\Ax{G2}]
				\begin{align*}
					\forall \alpha \in \At\;\;\partial(e+_b f)_\alpha
					&= \begin{cases}
						\partial(e)_\alpha & \alpha \leq_{\BA} b\\
						\partial(f)_\alpha & \alpha \leq_{\BA} \bar b
					\end{cases}\\
					&= \begin{cases}
						\partial(b;e)_\alpha & \alpha \leq_{\BA} b\\
						\partial(f)_\alpha & \alpha \leq_{\BA} \bar b
					\end{cases} &\cref{lem48}\\
					&= \partial(b;e +_b f)_\alpha
				\end{align*}
				\item[\Ax{G3}]
				\begin{align*}
					\forall \alpha \in \At\;\;\partial(e+_b f)_\alpha
					&= \begin{cases}
						\partial(e)_\alpha & \alpha \leq_{\BA} b\\
						\partial(f)_\alpha & \alpha \leq_{\BA} \bar b
					\end{cases}\\
					&= \partial(f+_{\bar{b}} e)_\alpha
				\end{align*}
				\item[\Ax{G4}]
				\begin{align*}
					\forall \alpha \in \At\;\;\partial((e+_b f)+_cg)_\alpha &= 
					\begin{cases}
						\partial(g)_\alpha & \alpha \leq_{\BA} \bar c\\
						\partial(e+_b f)_\alpha & \alpha \leq_{\BA} c
					\end{cases}\\
					&=\begin{cases}
						\partial(g)_\alpha & \alpha \leq_{\BA} \bar c\\
						\partial(e)_\alpha & \alpha \leq_{\BA} bc\\
						\partial(f)_\alpha & \alpha \leq_{\BA} \bar{b}c
					\end{cases}\\
					&=\begin{cases}
						\partial(g)_\alpha & \alpha \leq_{\BA} b\bar c\\
						\partial(g)_\alpha & \alpha \leq_{\BA} \bar b\bar c\\
						\partial(e)_\alpha & \alpha \leq_{\BA} bc\\
						\partial(f)_\alpha & \alpha \leq_{\BA} \bar{b}c
					\end{cases}\\
					&=\begin{cases}
						\partial(e)_\alpha & \alpha \leq_{\BA} bc\\
						\begin{cases}
							\partial(f)_\alpha & \alpha \leq_{\BA} c\\
							\partial(g)_\alpha & \alpha \leq_{\BA} \bar c
						\end{cases} & \alpha \leq_{\BA} \bar{bc}
					\end{cases}\\
					&=\begin{cases}
						\partial(e)_\alpha & \alpha \leq_{\BA} bc\\
						\partial(f+_c g) & \alpha \leq_{\BA} \bar{bc}
					\end{cases}\\
					&= \partial(e+_{bc}(f+_cg))_\alpha\\
				\end{align*}
				
				\item[\Ax{D1}]
				\begin{align*}
					\forall \alpha \in \At\;\;\partial(e \WC{r}{s} (f+_b g))_\alpha &= r\partial(e)_\alpha +s\partial(f+_b g)_\alpha\\
					&=\begin{cases}
						r\partial(e)_\alpha + s\partial(f)_\alpha & \alpha \leq_{\BA} b\\
						r\partial(e)_\alpha + s\partial(g)_\alpha & \alpha \leq_{\BA} \bar b\\
					\end{cases}\\
					&=\begin{cases}
						\partial(e \WC{r}{s} f)_\alpha & \alpha \leq_{\BA} b\\
						\partial(e \WC{r}{s} g)_\alpha & \alpha \leq_{\BA} \bar b\\
					\end{cases}\\
					& = \partial((e \WC{r}{s} f)+_b (e \WC{r}{s} g))_\alpha
				\end{align*}
				\item[\Ax{D2}]
				\begin{align*}
					\forall \alpha \in \At\;\;\partial(e\WC{r}{s} (f \WC{t}{u} g))_\alpha &= r\partial(e)+s\partial(f \WC{t}{u} g))_\alpha\\
					&= r\partial(e)+s(t\partial(f)_\alpha + u\partial(g)_\alpha)\\
					&= r\partial(e)+(st\partial(f)_\alpha + su\partial(g)_\alpha) & \cref{semimodule}\\
					&= r\partial(e)+\partial(f\WC{st}{su} g)_\alpha\\
					&= \partial(e \WC{r}{1} (f\WC{st}{su} g))_\alpha
				\end{align*}
				\item[\Ax{D3}]
				\begin{align*}
					\partial(b;(e \WC{r}{s} f))_\alpha &= \begin{cases}
						r\partial(e)_\alpha + s\partial(f)_\alpha & \alpha \leq b\\
						\delta_\xmark & \alpha \leq \bar b\\
					\end{cases}\\
					&= \begin{cases}
						r\partial(b;e)_\alpha + s\partial(b;f)_\alpha & \alpha \leq b\\
						\delta_\xmark & \alpha \leq \bar b\\
					\end{cases}& \cref{lem48}\\
					&= \begin{cases}
						\partial(b;e \WC{r}{s} b;f)_\alpha & \alpha \leq b\\
						\delta_\xmark & \alpha \leq \bar b\\
					\end{cases}\\
					&= \partial(b;(b;e \WC{r}{s} b;f))_\alpha 
				\end{align*}
				
				\item[\Ax{S1}]
				For all $\alpha \in \At$
				\begin{align*}
					\partial(\bskip;e)_\alpha &= \partial(e)_\alpha & \cref{lem48}\\ 
					&= \begin{cases}
						\partial(e)_\alpha(x) & x=\cmark \\
						\partial(e)_\alpha(x) & \text{else}
					\end{cases}\\
					&= \begin{cases}
						\partial(e)_\alpha(x)1 & x=\cmark \\
						\partial(e)_\alpha(x)+0 & \text{else}
					\end{cases}\\
					&= \begin{cases}
						\partial(e)_\alpha(x)\partial(\bskip)_\alpha(x) & x=\cmark \\
						\partial(e)_\alpha(x)+\partial(e)_\alpha(\cmark)\partial(\bskip)_\alpha(x) & \text{else}
					\end{cases}\\
					&= \partial(e;\bskip)_\alpha
				\end{align*}
				\item[\Ax{S2}]
				\begin{align*}
					\forall \alpha \in \At\;\;\partial((e;f);g)_\alpha(\cmark) &= \partial(e;f)_\alpha(\cmark)\partial(g)_\alpha(\cmark)\\
					&= \partial(e)_\alpha(\cmark)\partial(f)_\alpha(\cmark)\partial(g)_\alpha(\cmark)\\
					&= \partial(e)_\alpha(\cmark)(\partial(f;g)_\alpha(\cmark))\\
					&= \partial(e;(f;g))_\alpha(\cmark)
				\end{align*}
				
				For $o \neq \cmark$
				\begin{align*}	
					\partial((e;f);g)_\alpha(o) &= \partial(e;f)_\alpha(\cmark)\partial(g)_\alpha(o) + \partial(e;f)_\alpha(o)\\
					&= \partial(e)_\alpha(\cmark)\partial(f)_\alpha(\cmark)\partial(g)_\alpha(o)  + \partial(e)_\alpha(\cmark)\partial(f)_\alpha(o)+\partial(e)_\alpha(o)\\
					&= \partial(e)_\alpha(\cmark)(\partial(f)_\alpha(\cmark)\partial(g)_\alpha(o)  + \partial(f)_\alpha(o))+\partial(e)_\alpha(o)\\
					&= \partial(e)_\alpha(\cmark)(\partial(f;g)_\alpha(o))+\partial(e)_\alpha(o)\\
					&= \partial(e;(f;g))_\alpha(o)
				\end{align*}
				\item[\Ax{S3}]
				\begin{align*}
					\forall \alpha \in \At\;\;\partial(\babort;e)_\alpha(\xmark) &= \partial(\babort)_\alpha(\xmark) + \partial(\babort)_\alpha(\cmark)\partial(e)_\alpha(\xmark)\\
					&= \partial(\babort)_\alpha(\xmark)\\
					&= \bskip
				\end{align*}
				
				All other outputs have $0$ weight in both expressions, for all atoms by the definition of sequential composition.
				\item[\Ax{S4}]
				We will split up the cases of output. $\forall \alpha \in \At$ 
				
				if $o \in \{\xmark\} + \Out$.
				\begin{align*}
					&\partial( (e \WC{r}{s} f);g)_\alpha(o)\\ 
					&=\partial( e \WC{r}{s} f)_\alpha(o) + \partial( e \WC{r}{s} f)_\alpha(\cmark)\partial(g)_\alpha(o)\\ 
					&= r\partial(e)_\alpha(o)+s\partial(f)_\alpha(o) + r\partial(e)_\alpha(\cmark)\partial(g)_\alpha(o)+s\partial(f)_\alpha(\cmark)\partial(g)_\alpha(o)\\
					&= r(\partial(e)_\alpha(o)+\partial(e)_\alpha(\cmark)\partial(g)_\alpha(o))+s(\partial(f)_\alpha(o)+\partial(f)_\alpha(\cmark)\partial(g)_\alpha(o))\\
					&= r\partial(e;g)_\alpha(o) + s\partial(f;g)_\alpha(o)\\
					&= \partial(e;g \WC{r}{s} f;g)_\alpha(o)  
				\end{align*}
				Furthermore
				\begin{align*}
					\partial( (e \WC{r}{s} f);g)_\alpha(\cmark) &= \partial( e \WC{r}{s} f)_\alpha(\cmark)\partial(g)_\alpha(\cmark)\\ 
					&= r\partial_\alpha(e)(\cmark)\partial(g)_\alpha(\cmark) +  s\partial_\alpha(f)(\cmark)\partial(g)_\alpha(\cmark)\\ 
					&= r\partial_\alpha( e;g)_\alpha(\cmark) +  s\partial_\alpha( f;g)_\alpha(\cmark)\\ 
					&= \partial(e;g \WC{r}{s} f;g)_\alpha(\cmark)  
				\end{align*}
				Finally if $p \in \Act$ and $Q\in \Exp/\equiv$. Then $\forall \alpha \in \At$ 
				\begin{align*}
					&\partial( (e \WC{r}{s} f);g)_\alpha[\{p\} \times Q] \\
					&= \partial( e \WC{r}{s} f)_\alpha[\{p\} \times Q/g] + \partial( e \WC{r}{s} f)_\alpha(\cmark)\partial(g)_\alpha[\{p\} \times Q] &  \cref{lem42}\\
					&=r\partial( e)_\alpha[\{p\} \times Q/g] + s\partial(f)_\alpha[\{p\} \times Q/g] \\&\;+ s\partial(f)_\alpha(\cmark)\partial(g)_\alpha[\{p\} \times Q] + r\partial(e)_\alpha(\cmark)\partial(g)_\alpha[\{p\} \times Q]\\ 
					&= r(\partial( e)_\alpha[\{p\} \times Q/g]+\partial(e)_\alpha(\cmark)\partial(g)_\alpha[\{p\} \times Q])\\&\;+ s(\partial(f)_\alpha[\{p\} \times Q/g] + \partial(f)_\alpha(\cmark)\partial(g)_\alpha[\{p\} \times Q])\\
					&= r(\partial(e;g)_\alpha[\{p\} \times Q]) +s(\partial(f;g)_\alpha[\{p\} \times Q])& \cref{lem42}\\
					&= \partial(e;g \WC{r}{s} f;g)_\alpha[\{p\} \times Q]		
				\end{align*}
				
				\item[\Ax{S5}] 
				\begin{align*}
					\partial((e+_bf);g)_\alpha(\cmark) &= \partial(e+_bf)_\alpha(\cmark)\partial(g)_\cmark(\cmark)\\
					&= \begin{cases}
						\partial(e)_\alpha(\cmark)\partial(g)_\cmark(\cmark) & \alpha \leq_{\BA} b\\
						\partial(f)_\alpha(\cmark)\partial(g)_\cmark(\cmark) & \alpha \leq_{\BA} \bar b\\
					\end{cases}\\
					&= \partial(e;g +_b f;g)_\alpha(\cmark)
				\end{align*}		
				
				For $o \in \xmark + \Out$ 
				\begin{align*}
					\partial((e+_bf);g)_\alpha(o) &= \partial(e+_bf)_\alpha(o) + \partial(e+_bf)_\alpha(\cmark)\partial(g)_\alpha(o)\\
					&= \begin{cases}
						\partial(e)_\alpha(o)+\partial(e)_\alpha(\cmark)\partial(g)_\alpha(\cmark) & \alpha \leq_{\BA} b\\
						\partial(f)_\alpha(o)+\partial(f)_\alpha(\cmark)\partial(g)_\alpha(\cmark) & \alpha \leq_{\BA} \bar b\\
					\end{cases}\\
					&= \partial(e;g +_b f;g)_\alpha(o)
				\end{align*}
				Finally, for $p \in \Act, Q \in \Exp/\equiv$
				\begin{align*}
					&\partial((e+_bf);g)_\alpha[\text{p} \times Q]\\
					&= \partial(e+_bf)_\alpha[\text{p} \times Q/g] + \partial(e+_bf)_\alpha(\cmark)\partial(g)_\alpha[\text{p} \times Q]\\
					&= \begin{cases}
						\partial(e)_\alpha[\text{p} \times Q/g] + \partial(e)_\alpha(\cmark)\partial(g)_\alpha[\text{p} \times Q] & \alpha \leq_{\BA} b\\
						\partial(f)_\alpha[\text{p} \times Q/g] + \partial(f)_\alpha(\cmark)\partial(g)_\alpha[\text{p} \times Q] & \alpha \leq_{\BA} \bar b\\
					\end{cases}\\
					&= \partial(e;g +_b f;g)_\alpha[\text{p} \times Q]
				\end{align*}
				\item[\Ax{S6}]
				\begin{align*}
					\forall \alpha \in \At\;\;\partial(v;e)_\alpha(o) = \partial(v)_\alpha(\cmark)\partial(e)_\alpha(o) + \partial(v)_\alpha(o) = \partial(v)_\alpha(o)
				\end{align*}
				\item[\Ax{S7}]
				\begin{align*}
					\forall \alpha \in \At\;\;\partial(b;c)_\alpha &= \begin{cases}
						\partial(c)_\alpha & \alpha \leq_{\BA} b\\
						\delta_{\xmark} & \alpha \leq_{\BA} \bar b\\
					\end{cases} &\cref{lem48}\\ 
					&= \begin{cases}
						\delta_\cmark & \alpha \leq_{\BA} b \land \alpha \leq_{\BA} c\\
						\delta_{\xmark} & \text{else}\\
					\end{cases}\\
					&= \begin{cases}
						\delta_\cmark & \alpha \leq_{\BA} bc \\
						\delta_{\xmark} & \text{else}\\
					\end{cases}\\
					&= \partial(bc)_\alpha
				\end{align*}
				\item[\Ax{L1}]
				If $\alpha \leq_{\BA} \bar b$ then
				\begin{align*}
					\partial(e^{(b)})_\alpha(\cmark) = 1 = \partial(e;e^{(b)}+_b \bskip)_\alpha(\cmark) 
				\end{align*}
				And all other weights are $0$.
				This satisfies the required equality.
				If $\alpha \leq_{\BA} b$ then:
				\begin{align*}
					\partial(e;e^{(b)}+_b \bskip)_\alpha(\cmark) &= \partial(e;e^{(b)})_\alpha(\cmark)\\
					&= \partial(e)_\alpha(\cmark)\partial(e^{(b)})_\alpha(\cmark)\\
					&= 0\\
					&= \partial(e^{(b)})_\alpha(\cmark)
				\end{align*}
				If $o \in 2 + \Out$ then:
				\begin{align*}
					\partial(e;e^{(b)}+_b \bskip)_\alpha(o) &= \partial(e;e^{(b)})_\alpha(o)\\
					&= \partial(e)_\alpha(o)+\partial(e)_\alpha(\cmark)\partial(e^{(b)})_\alpha(o)\\
					&= \partial(e)_\alpha(o)+\partial(e)_\alpha(\cmark)\partial(e)_\alpha(\cmark)^*\partial(e)_\alpha(o)\\
					&= (1+\partial(e)_\alpha(\cmark)\partial(e)_\alpha(\cmark)^*)\partial(e)_\alpha(o)\\
					&= \partial(e)_\alpha(\cmark)^*\partial(e)_\alpha(o) &\cref{fixpointrule}\\
					&= \partial(e^{(b)})_\alpha(o)
				\end{align*}
				If $p \in \Act$ and $Q \in \Exp/\equiv$. Then: 
				\begin{align*}
					&\partial\left(e;e^{(b)}+_b \bskip\right)_\alpha[\{p\}\times Q] \\
					&= \partial\left(e;e^{(b)}\right)_\alpha[\{p\}\times Q]\\
					&= \partial\left(e\right)_\alpha[\{p\}\times Q/e^{(b)}] + \partial(e)_\alpha(\cmark)\partial(e^{(b)})_\alpha[\{p\}\times Q]\\
					&= \partial\left(e\right)_\alpha[\{p\}\times Q/e^{(b)}] + \partial(e)_\alpha(\cmark)\partial(e)_\alpha(\cmark)^*\partial(e)_\alpha[\{p\}\times Q/e^{(b)}]\\
					&= \partial(e)_\alpha(\cmark)^*\partial(e)_\alpha[\{p\}\times Q/e^{(b)}]\\
					&= \partial\left(e^{(b)}\right)_\alpha[\{p\}\times Q]
				\end{align*}
				\item[\Ax{L2}] 
				For our first case assume that $\alpha \leq_{\BA} \bar c$.
				\begin{align*}
					\partial(c;e^{(b)})_\alpha(\xmark) = \partial(c)_\alpha(\xmark) = 1 = \partial(c;((\odot s^*r ;f;(e)^{(b)})+_b \bskip))_\alpha(\xmark)
				\end{align*}
				Since all other weights are $0$ this is enough to prove the equality this case.
				
				We now consider the case where $\alpha \leq_{\BA} c$ and $\alpha \leq_{\BA} \bar b$.
				\begin{align*}
					\partial(c;e^{(b)})_\alpha(x) = \partial(e^{(b)})_\alpha(\cmark) = 1& = \partial(c;((\odot s^*r ;f;(e)^{(b)})+_b \bskip))(\cmark) & \cref{lem48}
				\end{align*}
				And all other weights are $0$. This verifies the required equality.
	
				Finally we check the case where $\alpha \leq_{\BA} c$ and $\alpha \leq_{\BA} b$.			
				Suppose $o \in \{\xmark\}\cup \Out$.
				\begin{align*}
					&\partial(c;e^{(b)})_\alpha(o) = \partial(((f \WC{r}{s} \bskip) +_c g)^{(b)})_\alpha(o)&\\
					&= \partial((f \WC{r}{s} \bskip)+_c g)_\alpha(\cmark)^*\partial((f \WC{r}{s} \bskip)+_c g)_\alpha(o)\\
					&= (r\partial(f)_\alpha(\cmark)+s)^*r\partial(f)_\alpha(o)\\
					&= s^*(r\partial(f)_\alpha(\cmark)s^*)^*r\partial(f)_\alpha(o) & \cref{conwaySemiring} (\ref{denesting})\\
					&= (s^*r\partial(f)_\alpha(\cmark))^*s^*r\partial(f)_\alpha(o) & \cref{slidingrulelem}\\
					&= s^*r\partial(f)_\alpha(o)+(s^*r\partial(f)_\alpha(\cmark))(s^*r\partial(f)_\alpha(\cmark))^*s^*r\partial(f)_\alpha(o) & \cref{fixpointrule}\\
					&= s^*r\partial(f)_\alpha(o)+s^*r(\partial(f)_\alpha(\cmark))\partial(((f \WC{r}{s} \bskip) +_c g)^{(b)})_\alpha(o)\\
					&=s^*r(\partial(f)_\alpha(o) + \partial(f)_\alpha(\cmark)\partial(e^{(b)})_\alpha(o))\\
					&=s^*r\partial(f;e^{(b)})_\alpha(o)\\
					&=\partial(\odot s^*r;f;e^{(b)})_\alpha(o)\\
					&=\partial((\odot s^*r;f;e^{(b)})+_b \bskip)_\alpha(o)\\
					&=\partial(c;((\odot s^*r;f;e^{(b)})+_b \bskip))_\alpha(o) & \cref{lem48}
				\end{align*}
				Furthermore
				\begin{align*}
					\partial(c;e^{(b)})_\alpha(\cmark) &= \partial(e^{(b)})_\alpha(\cmark) & \cref{lem48}\\
					&= 0\\
					&= \partial(c;((\odot s^*r;f;e^{(b)})+_b \bskip))_\alpha(\cmark)
				\end{align*}
				Finally let $p \in \Act$ and $Q \in \Exp/\equiv$.
				\begin{align*}
					&\partial(c;e^{(b)})_\alpha[\{p\} \times Q] \\
					&= \partial(e^{(b)})_\alpha[\{p\} \times Q] & \cref{lem48}\\
					&= \partial(((f \WC{r}{s} \bskip) +_c g)^{(b)})_\alpha[\{p\} \times Q]&\\
					&= \partial((f \WC{r}{s} \bskip)+_c g)_\alpha(\cmark)^*\\
						&\cdot\partial((f \WC{r}{s} \bskip)+_c g)_\alpha[\{p\} \times Q/((f \WC{r}{s} \bskip)+_c g)^{(b)}]\\
					&= (r\partial(f)_\alpha(\cmark)+s)^*r\partial(f)_\alpha[\{p\} \times Q/((f \WC{r}{s} \bskip)+_c g)^{(b)}]\\
					&= s^*(r\partial(f)_\alpha(\cmark)s^*)^*r\partial(f)_\alpha[\{p\} \times Q/((f \WC{r}{s} \bskip)+_c g)^{(b)}] & \cref{conwaySemiring} (\ref{denesting})\\
					&= (s^*r\partial(f)_\alpha(\cmark))^*s^*r\partial(f)_\alpha[\{p\} \times Q/((f \WC{r}{s} \bskip)+_c g)^{(b)}] & \cref{slidingrulelem}\\
					&= s^*r\partial(f)_\alpha[\{p\} \times Q/((f \WC{r}{s} \bskip)+_c g)^{(b)}]+(s^*r\partial(f)_\alpha(\cmark))\\&\cdot(s^*r\partial(f)_\alpha(\cmark))^*s^*r\partial(f)_\alpha[\{p\} \times Q/((f \WC{r}{s} \bskip)+_c g)^{(b)}] & \cref{fixpointrule}\\
					&= s^*r\partial(f)_\alpha[\{p\} \times Q/((f \WC{r}{s} \bskip)+_c g)^{(b)}]+s^*r(\partial(f)_\alpha(\cmark))\\&\partial(((f \WC{r}{s} \bskip) +_c g)^{(b)})_\alpha[\{p\} \times Q]\\
					&= s^*r(\partial(f)_\alpha[\{p\} \times Q/((f \WC{r}{s} \bskip) +_c g)^{(b)}] \\&+ \partial(f)_\alpha(\cmark)\partial((f \WC{r}{s} \bskip) +_c g)^{(b)})_\alpha[\{p\} \times Q])\\
					&= s^*r\partial(f;(f \WC{r}{s} \bskip) +_c g)^{(b)})_\alpha[\{p\} \times Q]\\
					&= s^*r\partial(f;e^{(b)})_\alpha[\{p\} \times Q]\\
					&= \partial(\odot s^*r;f;e^{(b)})_\alpha[\{p\} \times Q]\\
					&=\partial((\odot s^*r;f;e^{(b)})+_b \bskip)_\alpha[\{p\} \times Q]\\
					&=\partial(c;((\odot s^*r;f;e^{(b)})+_b \bskip))_\alpha[\{p\} \times Q] & \cref{lem48}
				\end{align*}
				\item[\Ax{C1}]
				We appeal to the properties of semimodules, \cref{semimodule}.
				\begin{align*}
					\forall \alpha \in \At\;\;\partial(\odot 1)_\alpha = 1\delta_\cmark = \delta_\cmark = \partial(\bskip)_\alpha
				\end{align*}
				\item[\Ax{C2}]
				We appeal to the properties of semimodules, \cref{semimodule}.
				\begin{align*}\forall \alpha \in \At\;\;\partial(\odot 0;e)_\alpha = 0\partial(e)_\alpha = 0 = 0 \partial(\bskip)_\alpha = \partial(\odot 0)_\alpha
				\end{align*}
				\item[\Ax{W1}]
				We appeal to the properties of semimodules, \cref{semimodule}.
				\begin{align*}
					\forall \alpha \in \At\;\;\partial(e \WC{r}{s} e)_\alpha = r\partial(e)_\alpha + s\partial(e)_\alpha = (r+s)\partial(e)_\alpha+0 = \partial(\odot(r+s);e)_\alpha
				\end{align*}
				\item[\Ax{W2}]
				We appeal to the commutativity of $+$ in semimodules (\cref{semimodule}).
				\begin{align*}
					\forall \alpha \in \At\;\;\partial(e \WC{r}{s} f)_\alpha = r\partial(e)_\alpha + s\partial(f)_\alpha = s\partial(f)_\alpha + r\partial(e)_\alpha = \partial(f \WC{s}{r} e)_\alpha
				\end{align*}
				\item[\Ax{W3}]
				\begin{align*}
					\forall \alpha \in \At\;\;\partial(e \WC{r}{s} (f \WC{t}{u} g))_\alpha &= r\partial(e)_\alpha + s\partial(f \WC{t}{u} g)_\alpha\\
					&= r\partial(e)_\alpha + s(t\partial(f)_\alpha + u\partial(g)_\alpha)\\
					&= r\partial(e)_\alpha + st\partial(f)_\alpha + su\partial(g)_\alpha & \cref{semimodule}\\
					&= \partial(e \WC{r}{st} f)_\alpha + su\partial(g)_\alpha\\
					&= \partial((e \WC{r}{st} f) \WC{1}{su}g)_\alpha
				\end{align*}
				\item[\Ax{W4}]
				\begin{align*}
					\forall \alpha \in \At\;\;\partial(e \WC{ru}{s} f)_\alpha &= ru\partial(e)_\alpha+s\partial(f)_\alpha\\
					&= r\partial(\odot u;e)_\alpha+s\partial(f)_\alpha &\cref{semimodule}\\
					&= \partial((\odot u;e) \WC{r}{s} (f))_\alpha
				\end{align*}
				\item[\Ax{F1}]
				First we consider the case where $\alpha \leq_{\BA} \bar b$.
				
				\begin{align*}
					\partial(e^{(b)};f)_\alpha &= \partial(f)_\alpha\\
					&= \partial(e;g +_b f)_\alpha
				\end{align*}
				
				We now consider the case where $\alpha \leq_{\BA} b$. Suppose $o \in \{\xmark\} \cup \Out$.
				\begin{align*}
					\partial(g)_\alpha(o) &= \partial(e;g+_b f)_\alpha(o) &\\
					&= \partial(e;g)_\alpha(o)\\
					&= \partial(e)_\alpha(o)+ \partial(e)_\alpha(\cmark)\partial(g)_\alpha(o)\\  
					&= \partial(e)_\alpha(o) & E(e) = 0\\
					&= 0^*\cdot\partial(e)_\alpha(o) & \cref{fixpointrule}\\
					&= \partial(e)_\alpha(\cmark)^*\partial(e)_\alpha(o) & E(e) = 0\\   
					&= \partial(e^{(b)})_\alpha(o)\\
					&= \partial(e^{(b)})_\alpha(o) +  \partial(e^{(b)})_\alpha(\cmark);\partial(f)_\alpha(0)\\
					&= \partial(e^{(b)};f)_\alpha(o)\\
				\end{align*}
				
				Finally consider $p\in \Act$ and $Q \in \Exp/\equiv$.
				\begin{align*}
					\partial(g)_\alpha[\{p\} \times Q] &= \partial(e;g+_b f)_\alpha[\{p\} \times Q] &\\
					&= \partial(e;g)_\alpha[\{p\} \times Q]\\
					&= \partial(e)_\alpha[\{p\} \times Q/g]+\partial(e)_\alpha(\cmark)\partial(g)_\alpha[\{p\} \times Q] & \cref{lem42}\\
					&= \partial(e)_\alpha[\{p\} \times Q/g] & E(e) = 0\\
					&= \partial(e)_\alpha(\cmark)^*\partial(e)_\alpha[\{p\} \times Q/g]\\
					&= \partial(e)_\alpha(\cmark)^*\partial(e)_\alpha[\{p\} \times Q/e^{(b)};f] & \cref{lem45}\\
					&= \partial(e)_\alpha(\cmark)^*\partial(e)_\alpha[\{p\} \times (Q/f)/e^{(b)}] & \cref{lem46}\\
					&= \partial(e^{(b)})_\alpha[\{p\} \times Q/f]\\
					&= \partial(e^{(b)})_\alpha[\{p\} \times Q/f] + \partial(e^{(b)})_\alpha(\cmark)\partial(f)_\alpha[\{p\} \times Q]\\
					&= \partial(e^{(b)};f)_\alpha[\{p\} \times Q]& \qedhere
				\end{align*}
			\end{enumerate} 
		\end{proof}
		With this in hand, proving soundness with respect to behavioral equivalence is straightforward.
		\begin{thm} \label{equivimpliesbeh}
			$e \equiv f \implies !_\partial e = !_\partial f$
		\end{thm}
		\begin{proof}
			Let $(\Exp_\equiv, \partial_\equiv)$ the unique coalgebra on $\Exp_\equiv$ which makes $[-]_\equiv$ a coalgebra homomorphism as per \cref{SoundnessMorphism}. Since $[-]_\equiv$ is an epimorphism this coalgebra structure is unique by \cite[Lemma 3.17]{gumm1999elements}. $(\Exp_\equiv, \partial_\equiv)$ has a unique morphism to $(Z, \zeta)$, call this $!_\equiv$. Note that $!_\equiv \circ [-]_\equiv: (\Exp, \partial) \to (Z,\zeta)$ is the unique morphism to the final coalgebra, $!_\partial$. Therefore $e\equiv f \implies [e]_\equiv = [f]_\equiv \implies (!_\equiv\circ[-]_\equiv) e = (!_\equiv\circ[-]_\equiv )f \implies !_\partial e = !_\partial f$ \cite{ToddThesis}.
		\end{proof}
	\end{toappendix}
	\begin{correp}[Soundness]
		For all $e,f \in \Exp\;\; e\equiv f \implies e \sim f$\label{soundness}
	\end{correp}
	\begin{appendixproof}
		$e \equiv f \implies$ $!e = !f$ by \cref{equivimpliesbeh}. Furthermore, $!e = !f \implies e \sim f$ by \cref{bisimIffBeh}. The result is immediate.
	\end{appendixproof}
	\subsection{Derivable facts}\label{sec:derivable}
	Some other useful identities are derivable from the axioms.
	\begin{lemrep}
		For all $b,c \in \BExp,\; e,f,g,h \in \Exp, \;r,s,t,u \in \mathbb S$\\
		\begin{tabular}{L L L L}
			 \Ax{DF1}& \odot t;(e \WC{r}{s} f) \equiv e \WC{tr}{ts} f&\Ax{DF8} &b;(e+_c f) \equiv b;e +_c b;f \\
			 \Ax{DF2}& e+_b(f+_c g) \equiv (e+_b f)+_{b+c} g&\Ax{DF9} & b;(e+_c f) \equiv b;(b;e+_c f)\\
			 \Ax{DF3}& e+_b \babort \equiv b;e&\Ax{DF10} & \odot r; \odot s \equiv \odot (rs)\\
			 \Ax{DF4}& b;(e+_b f) \equiv b;e&\Ax{DF11} & \odot r;(e+_b f) \equiv \odot r;e +_b \odot r;f\\
			 \Ax{DF5}& (e+_b f)\WC{r}{s} g \equiv (e \WC{r}{s} g) +_b (f \WC{r}{s} g)&\Ax{DF12} & g \WC{r}{s} \odot 0 \equiv \odot r; g\\
			 \Ax{DF6}& (e+_b f) \WC{r}{s} (g+_b h) \equiv (e \WC{r}{s} g) +_b (f \WC{r}{s} h)&\Ax{DF13} & b;(e \WC{r}{s} f) \equiv b;(b;e \WC{r}{s} f)\\
			 \Ax{DF7}& e +_\bskip f \equiv e&\Ax{DF14} & g \WC{s}{t} \left( e \WC{r}{u} \odot 0 \right) \equiv g \WC{s}{tr} e\\
		\end{tabular}
	\end{lemrep}
	\begin{appendixproof}
		We now prove each of these derived facts. \Ax{DF2}-\Ax{DF4}, \Ax{DF8} retain the same proof as \gkat, so we simply cite \gkat \cite{Smolka_2019} and do not restate the proof.
		\item[\Ax{DF1}]
		\begin{align*}
			\odot t;(e \WC{r}{s} f) &\equiv (\bskip \WC{t}{0} \babort);(e \WC{r}{s} f) &\\
			&\equiv \bskip;(e \WC{r}{s} f) \WC{t}{0} \babort;(e \WC{r}{s} f) & \Ax{S4}\\
			&\equiv \babort;(e \WC{r}{s} f)\WC{0}{t} \bskip;(e \WC{r}{s} f)  & \Ax{W2}\\
			&\equiv \babort;(e \WC{r}{s} f)\WC{0}{t} (e \WC{r}{s} f)  & \Ax{S1}\\
			&\equiv \babort;(e \WC{r}{s} f)\WC{0}{1} (e \WC{tr}{ts} f)  & \Ax{D2}\\
			&\equiv (e \WC{tr}{ts} f)\WC{1}{0} \babort;(e \WC{r}{s} f)  & \Ax{W2}\\
			&\equiv (e \WC{tr}{ts} f)\WC{1}{0} \babort;(e \WC{tr}{ts} f)  & \Ax{S3}\\
			&\equiv \bskip;(e \WC{tr}{ts} f)\WC{1}{0} \babort;(e \WC{tr}{ts} f)  & \Ax{S1}\\
			&\equiv (\bskip \WC{1}{0} \babort);(e \WC{tr}{ts} f) & \Ax{S4}\\
			&\equiv \odot 1;(e \WC{tr}{ts} f) &\\
			&\equiv e \WC{tr}{ts} f & \Ax{C1}, \Ax{S1}
		\end{align*}
		\item[\Ax{DF5}]
		\begin{align*}
			(e+_b f)\WC{r}{s} g &\equiv g \WC{s}{r} (e+_b f) &\Ax{W2}\\
			&\equiv (g \WC{s}{r}e+_b g \WC{s}{r}f) &\Ax{D1}\\
			&\equiv (e\WC{r}{s} g+_b f \WC{r}{s} g) &\Ax{W2}
		\end{align*}
		\item[\Ax{DF6}]
		\begin{align*}
			(e+_b f) \WC{r}{s} (g+_b h) &\equiv ((e+_b f)\WC{r}{s} g) +_b ((e +_b f)\WC{r}{s} h)&\Ax{D1}\\
			&\equiv ((e \WC{r}{s} g)+_b (f \WC{r}{s} g)) +_b ((e +_b f)\WC{r}{s} h)&\Ax{DF5}\\
			&\equiv b;((e \WC{r}{s} g)+_b (f \WC{r}{s} g)) +_b ((e +_b f)\WC{r}{s} h)&\Ax{G2}\\
			&\equiv (b;(e \WC{r}{s} g)) +_b ((e +_b f)\WC{r}{s} h)&\Ax{DF4}\\
			&\equiv (e \WC{r}{s} g) +_b ((e +_b f)\WC{r}{s} h)&\Ax{G2}\\
			&\equiv ((f +_{\bar b} e)\WC{r}{s} h) +_{\bar b} (e \WC{r}{s} g)&\Ax{G3}\\
			&\equiv ((f \WC{r}{s} h) +_{\bar b} (e \WC{r}{s} h)) +_{\bar b} (e \WC{r}{s} g)&\Ax{DF5}\\
			&\equiv \bar b((f \WC{r}{s} h) +_{\bar b} (e \WC{r}{s} h)) +_{\bar b} (e \WC{r}{s} g)&\Ax{G2}\\
			&\equiv \bar b(f \WC{r}{s} h) +_{\bar b} (e \WC{r}{s} g)&\Ax{DF4}\\
			&\equiv (f \WC{r}{s} h) +_{\bar b} (e \WC{r}{s} g)&\Ax{G2}\\
			&\equiv (e \WC{r}{s} g) +_{b} (f \WC{r}{s} h)&\Ax{G3}
		\end{align*}
		\item[\Ax{DF7}]
		\begin{align*}
			e+_\bskip f &\equiv \bskip;(e+_\bskip f) & \Ax{S1}\\
			&\equiv \bskip;e & \Ax{DF4}\\
			&\equiv e & \Ax{S1}
		\end{align*}
		\item[\Ax{DF9}]
		\begin{align*}
			b;(e +_c f) &\equiv b;e +_c b;f & \Ax{DF8}\\
			&\equiv bb;e +_c b;f & bb \equiv_{\BA} b\\
			&\equiv b;b;e +_c b;f & \Ax{S7}\\
			&\equiv b;(b;e +_c f) & \Ax{DF8}
		\end{align*}
		\item[\Ax{DF10}] 
		\begin{align*}
			\odot r; \odot s &\equiv \bskip \WC{r}{0} \babort ; \bskip \WC{s}{0} \babort &\\
			&\equiv \bskip;(\bskip \WC{s}{0} \babort) \WC{r}{0} \babort;(\bskip \WC{s}{0} \babort) & \Ax{S4}\\
			&\equiv (\bskip \WC{s}{0} \babort) \WC{r}{0} \babort & \Ax{S1,S3}\\
			&\equiv \babort \WC{0}{r} (\bskip \WC{s}{0} \babort) & \Ax{W2}\\
			&\equiv \babort \WC{0}{1} (\bskip \WC{rs}{0} \babort) & \Ax{D2}\\
			&\equiv \bskip;(\bskip \WC{rs}{0} \babort) \WC{1}{0} \babort;(\bskip \WC{rs}{0} \babort) & \Ax{W2, S3, S1}\\
			&\equiv \odot 1;(\bskip \WC{rs}{0} \babort) & \Ax{S4}\\
			&\equiv \bskip \WC{rs}{0} \babort & \Ax{C1,S1}\\
			&\equiv \odot(rs)
		\end{align*}
		\item[\Ax{DF11}]
		\begin{align*}
			\odot r;(e+_b f) &\equiv (\bskip \WC{r}{0} \babort);(e+_b f) &\\
			&\equiv (\bskip;(e+_b f) \WC{r}{0} \babort;(e+_b f)) & \Ax{S4}\\
			&\equiv (e+_b f) \WC{r}{0} \babort & \Ax{S1, S3}\\
			&\equiv \babort \WC{0}{r} (e+_b f) & \Ax{W2}\\
			&\equiv (\babort \WC{0}{r} e)+_b (\babort \WC{0}{r} f) & \Ax{D1}\\
			&\equiv (e \WC{r}{0} \babort)+_b (f \WC{r}{0} \babort) & \Ax{W2}\\
			&\equiv (\bskip;e \WC{r}{0} \babort)+_b (\bskip;f \WC{r}{0} \babort) & \Ax{S1}\\
			&\equiv (\bskip;e \WC{r}{0} \babort;e)+_b (\bskip;f \WC{r}{0} \babort;f) & \Ax{S3}\\
			&\equiv (\bskip \WC{r}{0} \babort);e+_b (\bskip \WC{r}{0} \babort);f & \Ax{S4}\\
			&\equiv \odot r;e +_b \odot r;f&
		\end{align*}
		\item[\Ax{DF12}]
		\begin{align*}
		g \WC{r}{s} \odot 0 &\equiv g \WC{r}{s} \odot 0;\babort&\Ax{C2}\\
		&\equiv g \WC{r}{0} \babort&\Ax{W2, W4}\\
		&\equiv \bskip;g \WC{r}{0} \babort;g&\Ax{S1, S3}\\
		&\equiv \left(\bskip \WC{r}{0} \babort\right);g&\Ax{S4}\\
		&\equiv \odot r;g&
		\end{align*}
		\item[\Ax{DF13}] 
		\begin{align*}
		b;(e \WC{r}{s} f) &\equiv b;(b;e \WC{r}{s} b;f) & \Ax{D3}\\
		&\equiv b;((bb);e \WC{r}{s} b;f) & \text{BA}\\
		&\equiv b;(b;b;e \WC{r}{s} b;f) & \Ax{S7}\\
		&\equiv b;(b;e \WC{r}{s} f) & \Ax{D3}
		\end{align*}
		\item[\Ax{DF14}]
		\begin{align*}
			g \WC{s}{t} \left( e \WC{r}{u} \odot 0 \right) &\equiv
			g \WC{s}{t} \left( \odot r; e \right)& \Ax{DF12}\\
			&\equiv \left( \odot r; e \right) \WC{t}{s} g& \Ax{W2}\\
			&\equiv e \WC{tr}{s} g& \Ax{W4} \rlap{\qedhere} 
		\end{align*}
	\end{appendixproof}
	With the axioms and derived facts in hand, we now revisit our original examples. In the \emph{ski rental} problem the relevant property is the minimum total cost that could be paid over these sequential days. An appropriate choice of semiring to weight over is the \emph{tropical semiring}, as it allows us to reason about the cheapest branch of computation using sum to produce branches and minimum to compare branches. 
	\begin{exmp}[Ski Rental]
		We now prove that the optimal cost is  buying skis if the trip is longer than $n$ days and renting skis otherwise. Recall the encoding was $(\bskip \WC{1}{y} v)^n;v$ where $n$ is the number days of the trip, $y$ the cost of skis, and $v$ terminates the process. 
		
		We first prove by induction on $n$ the lemma $(\bskip \WC{1}{y} v)^n \equiv (\bskip \WC{n}{y} v)$. The base case where $n=1$ is trivial. Suppose it is true for $n=k$.
		\begin{align*}
			(\bskip \WC{1}{y} v)^{k+1} &\equiv (\bskip \WC{1}{y} v)^{k};(\bskip \WC{1}{y} v) \tag{Definition of $e^n$}\\
			&\equiv (\bskip \WC{k}{y} v);(\bskip \WC{1}{y} v) \tag{Induction hypothesis}\\
			&\equiv (\bskip;(\bskip \WC{1}{y} v) \WC{k}{y} v;(\bskip \WC{1}{y} v)) \tag{\textsf{S4}}\\
			&\equiv ((\bskip \WC{1}{y} v) \WC{k}{y} v) \tag{\textsf{S1, S6}}\\
			&\equiv (v \WC{y}{k} (v \WC{y}{1} \bskip)) \tag{\textsf{W2}}\\
			&\equiv ((v \WC{y}{ky} v) \WC{0}{k1} \bskip) \tag{\textsf{W3}}\\
			&\equiv (( \odot(y+ky);v) \WC{0}{k1} \bskip) \tag{\textsf{W1}}\\
			&\equiv (( \odot y;v) \WC{0}{k1} \bskip) \tag{$y+ky=y$}\\
			&\equiv (v \WC{y}{k1} \bskip) \tag{\textsf{W4}}\\
			&\equiv (\bskip \WC{k1}{y} v) \tag{\textsf{W2}}
		\end{align*}
		We now use this fact to simplify the expression of the ski rental problem:
		\begin{align*}
			(\bskip \WC{1}{y} v)^n;v &\equiv (\bskip \WC{n}{y} v);v \\
			&\equiv (\bskip;v\WC{n}{y} v;v) \tag{\textsf{S4}}\\
			&\equiv (v\WC{n}{y} v) & \tag{\textsf{S1, S6}}\\
			&\equiv \odot(n+y);v & \tag{\textsf{W1}}
		\end{align*}
		This is termination with weight equal to the minimum of $n$ and $y$. 
	\end{exmp}
	\begin{exmp}[Coin Game]
		Recall the encoding is $((\bskip \WC{1}{1} (\odot 1;\coin)) \WC{.5}{.5} (\odot 0;\coin))^{(\bskip)}$, where $\coin$ is the outcome of winning a dollar. We are interested in expected value, so we choose the \emph{extended non-negative rational semiring} because this can capture both the probabilities and the values of outcomes. We choose rationals rather than the reals to avoid having to compare infinite precision numbers. We now simplify the encoding:
		\begin{align*}
			&((\bskip \WC{1}{1} (\odot 1;\coin )) \WC{.5}{.5} (\odot 0;\coin))^{(\bskip)}\\
			&\qquad\equiv ((\bskip \WC{1}{1} (\odot 1;\coin)) \WC{.5}{.5} (\odot 0;\babort))^{(\bskip)} \tag{\textsf{C2}}\\
			&\qquad\equiv ((\bskip \WC{1}{1} (\odot 1;\coin)) \WC{.5}{0} \babort)^{(\bskip)} \tag{\textsf{W4}}\\
			&\qquad\equiv (\odot .5(\bskip \WC{1}{1} (\odot 1;\coin)) )^{(\bskip)} \tag{\cref{scalingdef}}\\
			&\qquad\equiv (\odot .5(\bskip \WC{1}{1} \coin) )^{(\bskip)} \tag{\textsf{C1, S1}}\\
			&\qquad\equiv (\bskip \WC{.5}{.5} \coin)^{(\bskip)} & \tag{\textsf{DF1}}\\
			&\qquad\equiv (\coin \WC{.5}{.5} \bskip)^{(\bskip)} & \tag{\textsf{W2}}\\
			&\qquad\equiv (\coin \WC{.5}{.5} \bskip +_\bskip \babort)^{(\bskip)} & \tag{\textsf{DF7}}\\
			&\qquad\equiv \bskip;\left(\left(\odot (.5^*\cdot.5);\coin;\left((\coin \WC{.5}{.5} \bskip) +_\bskip \babort\right)^{(\bskip)}\right) +_\bskip \bskip\right) & \tag{\textsf{L2}}\\
			&\qquad\equiv\odot (.5^*\cdot.5);\coin;(\coin \WC{.5}{.5} \bskip +_\bskip \babort)^{(\bskip)} & \tag{\textsf{S1, DF7}}\\
			&\qquad\equiv\odot (.5^*\cdot.5);\coin & \tag{\textsf{S6}}\\
			&\qquad\equiv\odot (2\cdot.5);\coin \\
			&\qquad\equiv\odot 1;\coin &
		\end{align*}
		Hence, the player expects to win one dollar in this game.
	\end{exmp}
	
	\section{Completeness}\label{completenessS}
	It is natural to ask if the set of axioms presented in \cref{axiomsS} is \emph{complete}, i.e., whether bisimilarity of any pair of \wgkat expressions can be established solely through the means of axiomatic manipulation. As we will see later on, the problem of completeness essentially relies on representing \wgkat automata syntactically as certain systems of equations and uniquely solving them. This is similar to proofs of completeness one might encounter in the literature concerning process algebra and Kleene Algebra. 
	
	Unfortunately, similarly to \gkat~\cite{Smolka_2019} and its probabilistic exension~\cite{rozowski2023probabilistic}, this poses a few technical challenges. Firstly, unlike Kleene Algebra, we do not have a perfect correspondence between the operational model (\wgkat automata) and the expressivity of the syntax. Namely, there exist \wgkat automata that cannot be described by a \wgkat expression. This simply stems from the fact that our syntax is able to only describe programs with well-structured control flow using while loops and if-then-else statements, which is strictly less expressive than using arbitrary non-local flow involving constructs like goto expressions in an uninterpreted setting.
	Secondly,  classic proofs of completeness of Kleene Algebra strictly rely on the left distributivity of sequential composition over branching, which is unsound under bisimulation semantics. This is crucial, as it allows for reduction of the problem of solving an arbitrary system of equations to the unary case that can be handled through the inference rule similar to \textsf{F1} of our system. 
	
	In the presence of these problems, we circumvent those issues by extending our inference system with a powerful axiom schema generalizing \textsf{F1} to arbitrary systems of equations (hence avoiding the issue of left distributivity) that states that each such system has \emph{at most one solution} (thus solving the problem of insufficient expressivity of our syntax). We will follow the \gkat terminology and refer to this axiom schema as the \emph{Uniqueness axiom} \textsf{(UA)}~\cite{Smolka_2019}, but the general idea is standard in process algebra and was explored by Bergstra under the name \emph{Recursive Specification Principle} \textsf{(RSP)}~ \cite{10.1007/3-540-16444-8_1}.  
	
	To obtain the completeness result, we will rely on the following roadmap:
	\begin{enumerate}
		\item First, we show that every \wgkat expression possesses a certain normal form that ties it to its operational semantics.
		\item Then, we define systems of equations akin to the normal form the step above and define what it means to solve them.
		\item Relying on the above two steps, we prove the correspondence between solutions to systems of equations and \wgkat automata homomorphisms.
		\item We establish the soundness of the Uniqueness Axiom. As a consequence of the correspondence from the step above, we show that two expressions being related by some bisimulation can be seen as them being solutions to the same system of equations. This combined with the use of \textsf{UA} guarantees completeness. 
	\end{enumerate}

	\subsection{Fundamental Theorem}
	In order to prove completeness we will use an intermediate result which states that each expression can be simplified to a normal form, obtained from the operational semantics of the expression. Following the terminology from the Kleene Algebra community we call this the \textit{fundamental theorem} \cite{RUTTEN20031}.
		\begin{defn}[\cite{rozowski2023probabilistic}]\label{gdedsum}
			Given $\Phi \subseteq \At$ and an indexed collection $\{e_\alpha\}_{\alpha \in \Phi}$ where $\forall \alpha\in\Phi, e_\alpha \in \Exp$ we define a generalized guarded sum inductively as
			$$
				\bigplus_{\alpha \in \Phi} = \babort \text{ if } \Phi = \emptyset\;\;\;\;\;\;\;\;\;\;\; \bigplus_{\alpha \in \Phi} = e_\gamma +_\gamma \bigplus_{\alpha \in \Phi \setminus\gamma} e_\alpha\;\; \gamma \in \Phi
			$$
		\end{defn}
		\begin{toappendix}
			\label{FTapx}
		\begin{lem} [\cite{rozowski2023probabilistic}, 52]
			Generalized guarded sums are well-defined up to $\equiv$ \label{lem52}
		\end{lem}
		\begin{lem}[\cite{rozowski2023probabilistic}, 53]
			Let $b,c \in \BExp$ and let $\{e_\alpha\}_{\alpha \in \At}$ be an indexed collection such that $\forall \alpha \;\; e_\alpha \in \Exp$. Then:
			$$
				c;\bigplus_{\alpha \leq_{\BA} b} e_\alpha \equiv \bigplus_{\alpha \leq_{\BA} bc}
			$$
		\end{lem}
		Where we take $\alpha \leq_{\BA} b$ in this case to actually be the set $\{\alpha \mid \alpha \leq_{\BA} b\}$ similar to \cite{Smolka_2019, rozowski2023probabilistic}
		\begin{lem}[\cite{rozowski2023probabilistic}, 54]
			Let $\Phi \subseteq \At$ and $\{e_\alpha\}_{\alpha \in \At}$ be an indexed collection such that $\forall \alpha \;\; e_\alpha \in \Exp$. Then:
			$$
				\bigplus_{\alpha \in \Phi} e_\alpha \equiv \bigplus_{\alpha \in \Phi} \alpha;e_\alpha
			$$\label{lem54}
		\end{lem}
		\begin{lem}[\cite{rozowski2023probabilistic}, 55]
			Let $\Phi \subseteq \At$ and $\{e_\alpha\}_{\alpha \in \At}$ be an indexed collection such that $\forall \alpha \;\; e_\alpha \in \Exp$ and let $f \in \Exp$. Then:
			$$
				\left( \bigplus_{\alpha \in \Phi} e_\alpha \right);f  \equiv \bigplus_{\alpha \in \Phi} e_\alpha;f
			$$\label{lem55}
		\end{lem}
		\begin{lem}\label{lem56}
			Let $\Phi \subseteq \At$ and let $\{e_\alpha\}_{\alpha \in \Phi}$ and $\{f_\alpha\}_{\alpha \in \Phi}$ be indexed collections such that $e_\alpha, f_\alpha \in \Exp$ for each $\alpha \in \Phi$ and let $r,s \in \mathbb S$. Then:
			$$\left(\bigplus_{\alpha \in \Phi}e_\alpha\right) \WC{r}{s} \left(\bigplus_{\alpha \in \Phi}f_\alpha\right) \equiv \bigplus_{\alpha \in \Phi}(e_\alpha \WC{r}{s} f_\alpha)$$
		\end{lem}
		\begin{proof}
			We prove this by induction on the size of $\Phi$. For the base case we take a $\Phi$ of size $0$.
			$$\left(\bigplus_{\alpha \in \Phi}e_\alpha\right) \WC{r}{s} \left(\bigplus_{\alpha \in \Phi}f_\alpha\right) \equiv \babort \WC{r}{s} \babort \equiv \bigplus_{\alpha \in \Phi}(e_\alpha \WC{r}{s} f_\alpha)$$
			For our inductive step:
			\begin{align*}
				&\left(\bigplus_{\alpha \in \Phi}e_\alpha\right) \WC{r}{s} \left(\bigplus_{\alpha \in \Phi}f_\alpha\right) \\&\equiv \left(e_\gamma +_\gamma \bigplus_{\alpha \in \Phi\setminus\{\gamma\}}e_\alpha\right) \WC{r}{s} \left(f_\gamma +_\gamma \bigplus_{\alpha \in \Phi\setminus\{\gamma\}}f_\alpha\right)&\cref{gdedsum}\\
				&\equiv \left(e_\gamma \WC{r}{s} f_\gamma\right) +_\gamma \left(\left(\bigplus_{\alpha \in \Phi\setminus\{\gamma\}}e_\alpha\right) \WC{r}{s} \left(\bigplus_{\alpha \in \Phi\setminus\{\gamma\}}f_\alpha\right)\right)&\Ax{DF6}\\
				&\equiv \left(e_\gamma \WC{r}{s} f_\gamma\right) +_\gamma \left(\bigplus_{\alpha \in \Phi\setminus\{\gamma\}}(e_\alpha \WC{r}{s} f_\alpha)\right)&\text{Inductive hypothesis}\\
				&\equiv \bigplus_{\alpha \in \Phi}(e_\alpha \WC{r}{s} f_\alpha)&\rlap{\qedhere} 
			\end{align*}
		\end{proof}
	\end{toappendix}
		
		\begin{defn} 
			\label{defBigOplus}
			Given a non-empty finite index set I and indexed collections $\{r_i\}_{i\in I}$ and $\{e_i\}_{i\in I}$ such that $e_i \in \Exp$ and $r_i \in \mathbb S$ for all $i\in I$ we define a sum  expression inductively:
			$$\bigoplus_{i\in I}r_i \cdot e_i = e_j \WC{r_j}{1} \left( \bigoplus_{i\in I\setminus\{j\}}r_i\cdot e_i \right)$$
		\end{defn}
		Both \cref{gdedsum} and \cref{defBigOplus} are well defined up to the order they are unrolled. See \cref{FTapx} for details.
	\begin{toappendix}
		
		We note that by \Ax{DF14} we can omit the final branch of the generalized weighted sum which has an empty index set and thus is the zero weighting. We sometimes opt to use \Ax{DF12} and express the sum with the scaling notation $\odot r;e$ rather than $e \WC{r}{1} \odot 0$.
		
		\begin{lem}
			Generalized weighted sums are well defined up to $\equiv$ \label{lem59} 
		\end{lem}
		\begin{proof}
			We prove this by induction similar to the generalized convex sums of \probgkat \cite{rozowski2023probabilistic}. The case where $I=\{i\}$ is immediate as $\bigoplus_{i\in \{I\}} r_i\cdot e_i \equiv e_i \WC{r}{1} \odot 0$.
			 
			Suppose $I =\{k,j\}$ then:
			\begin{align*}
				e_j \WC{r_j}{1} \left( \bigoplus_{i\in I\setminus\{j\}}r_i\cdot e_i \right) &\equiv e_j \WC{r_j}{1} (e_k \WC{r_k}{1} \odot 0) &\cref{defBigOplus}\\
				&\equiv e_j \WC{r_j}{1} (\odot r_k;e_k) &\Ax{DF12}\\
				&\equiv \odot (r_j; e_j) \WC{1}{1} (\odot r_k;e_k) &\Ax{W4}\\
				&\equiv \odot (r_k; e_k) \WC{1}{1} (\odot r_j;e_j) &\Ax{W2}\\
				&\equiv e_k \WC{r_k}{1} \left( \bigoplus_{i\in I\setminus\{k\}}r_i\cdot e_i \right)&\text{symmetry}
			\end{align*}
			For our inductive case we assume that $\{k,j\}\subset I$:
			\begin{align*}
				e_j \WC{r_j}{1} \left( \bigoplus_{i\in I\setminus\{j\}}r_i\cdot e_i \right) &\equiv e_j \WC{r_j}{1} \left(e_k \WC{r_k}{1}\left( \bigoplus_{i\in I\setminus\{j,k\}}r_i\cdot e_i \right)\right)&\cref{defBigOplus}\\
				& \equiv\left( e_j \WC{r_j}{r_k} e_k\right) \WC{1}{1}\left( \bigoplus_{i\in I\setminus\{j,k\}}r_i\cdot e_i \right)&\Ax{W3}\\
				& \equiv\left( e_k \WC{r_k}{r_j} e_j\right) \WC{1}{1}\left( \bigoplus_{i\in I\setminus\{j,k\}}r_i\cdot e_i \right)&\Ax{W2}\\
				& \equiv e_k \WC{r_k}{1} \left(e_j \WC{r_j}{1}\left( \bigoplus_{i\in I\setminus\{j,k\}}r_i\cdot e_i \right)\right)&\Ax{W3}\\
				& \equiv e_k \WC{r_k}{1} \left(\bigoplus_{i\in I\setminus\{k\}}r_i\cdot e_i \right)&\rlap{\qedhere} 
			\end{align*} 
		\end{proof}
		\begin{rem}
		We note here that we will sometimes write $$\bigoplus_{i \in I} r_i\cdot e_i \WC{}{} \bigoplus_{j \in J} s_j\cdot f_j$$ for the generalized weighted sum $$\bigoplus_{i \in I\cup J}t_k\cdot g_k \;\; t_k = \begin{cases}
			r_k & k \in I\\ s_k & k \in J 
		\end{cases},\quad g_k = \begin{cases}
		e_k & k \in I\\ f_k & k \in J
		\end{cases}$$ where $I$ and $J$ are finite and nonempty and $I \cap J = \emptyset$. This is useful to separate parts of a generalized weighted sum.
		\end{rem}
		\begin{lem}
			Let $I$ be a non-empty finite index set, $\{r_i\}_{i\in I}$ and $\{e_i\}_{i\in I}$ indexed collections such that $r_i \in \mathbb S$ and $e_i \in \Exp$ for all $i$ and $f \in \Exp$ then:
			$$\left(\bigoplus_{i\in I}r_i\cdot e_i\right);f \equiv \left(\bigoplus_{i\in I}r_i\cdot e_i;f\right)$$ \label{lem61}
		\end{lem}
		\begin{proof}
			By induction on the size of $I$. For the base case let $I = \{j\}$.
			\begin{align*}
				\left(\bigoplus_{i\in I}r_i\cdot e_i\right);f &\equiv (\odot r_j;e_j);f&\Ax{DF12}\\ 
				&\equiv (\odot r_j;e_j;f)&\Ax{S4}\\ 
				&\equiv \left(\bigoplus_{i\in I}r_i\cdot e_i;f\right) &\Ax{DF12}
			\end{align*}
			For the inductive case suppose $\{j\} \subset I$.
			\begin{align*}
				\left(\bigoplus_{i\in I}r_i\cdot e_i\right);f &\equiv  \left(e_j \WC{r_j}{1} \left(\bigoplus_{i\in I\setminus\{j\}}r_i\cdot e_i\right)\right);f&\cref{defBigOplus}\\
				&\equiv  \left(e_j;f \WC{r_j}{1} \left(\bigoplus_{i\in I\setminus\{j\}}r_i\cdot e_i\right);f\right)&\Ax{S4}\\
				&\equiv  \left(e_j;f \WC{r_j}{1} \left(\bigoplus_{i\in I\setminus\{j\}}r_i\cdot e_i;f\right)\right)&\text{Inductive hypothesis}\\
				&\equiv  \left(\bigoplus_{i\in I}r_i\cdot e_i;f\right)& \rlap{\qedhere} 
			\end{align*}
		\end{proof}
		\begin{lem}
			Let $I$ be a non-empty finite index set, $\{r_i\}_{i\in I}$ and $\{e_i\}_{i\in I}$ indexed collections such that $r_i \in \mathbb S$ and $e_i \in \Exp$ for all $i\in I$. Let $E = \bigcup_{i \in I}\{e_i\}$. Then: 
			
			$$\bigoplus_{i \in I} r_i\cdot e_i \equiv \bigoplus_{e \in E}\left(\sum_{e_i=e}r_i\right)\cdot e$$ \label{lem63}
		\end{lem}
		\begin{proof}
			We prove this by induction on the size of $I$. If $|I| = 1$ then the sum is $$e_1 \WC{r_1}{0} \babort$$ which satisfies the lemma.
			
			We now present the inductive case. Let $I = k+1$. We will pull out some $e_k$. Either $e_k$ is equal to another other element in $\{e_i\}_{i \in I}$ or not. In the first case:
			\begin{align*}
				&\bigoplus_{i \in I} r_i\cdot e_i \\
				&\equiv e_k \WC{r_k}{1}\bigoplus_{j \in I \setminus \{k\}}r_j\cdot e_j&\cref{defBigOplus}\\
				&\equiv e_k \WC{r_k}{1}\bigoplus_{e \in E}\left(\sum_{\{i \mid e_i=e, i \neq k\}}r_i\right)\cdot e&\text{Ind. hypothesis}\\
				&\equiv e_k \WC{r_k}{1}\left(e_k \WC{\sum_{\{i \mid e_i=e_k, i \neq k\}} r_ i}{1} \bigoplus_{e \in E \setminus \{e_k\}}\left(\sum_{\{i \mid e_i=e\}}r_i\right)\cdot e\right)&\cref{defBigOplus}\\
				&\equiv \left( e_k \WC{r_k}{\sum_{\{i \mid e_i=e_k, i \neq k\}} r_ i}e_k \right) \WC{1}{1} \bigoplus_{e \in E \setminus \{e_k\}}\left(\sum_{\{i \mid e_i=e\}}r_i\right)\cdot e&\Ax{W3}\\
				&\equiv \left( \odot \sum_{\{i \mid e_i=e_k\}} r_i;  e_k\right) \WC{1}{1} \bigoplus_{e \in E \setminus \{e_k\}}\left(\sum_{\{i \mid e_i=e\}}r_i\right)\cdot e&\Ax{W1}\\
				&\equiv e_k \WC{\sum_{\{i \mid e_i=e_k\}} r_i}{1} \bigoplus_{e \in E \setminus \{e_k\}}\left(\sum_{\{i \mid e_i=e\}}r_i\right)\cdot e& \Ax{W4}\\
				&\equiv \bigoplus_{e \in E}\left(\sum_{e_i=e}r_i\right)\cdot e&\cref{defBigOplus}
			\end{align*}
			In the second case:
			\begin{align*}
				\bigoplus_{i \in I} r_i\cdot e_i &\equiv e_k \WC{r_k}{1}\bigoplus_{j \in I \setminus \{i\}}r_j\cdot e_j&\cref{defBigOplus}\\
				&\equiv e_k \WC{r_k}{1}\bigoplus_{e \in E \setminus \{e_k\}}\left(\sum_{\{i \mid e_i=e\}}r_i\right)\cdot e&\text{Inductive hypothesis}\\
				&\equiv e_k \WC{\sum_{i \mid e_i=e_k}r_i}{1}\bigoplus_{e \in E \setminus \{e_k\}}\left(\sum_{\{i \mid e_i=e\}}r_i\right)\cdot e&\\
				&\equiv \bigoplus_{e \in E}\left(\sum_{e_i=e}r_i\right)\cdot e&\rlap{\qedhere} 
			\end{align*} 
		\end{proof}
		\begin{lem}
			Let $I$ be a non-empty finite index set, $\{r_i\}_{i \in I}$ and $\{e_i\}_{i \in I}$ indexed collections such that $r_i \in S$ and $e_i \in \Exp$ for all $i \in I$. Let $E = \bigcup_{i \in I}\{e_i\}$ then: 
			$$\bigoplus_{i\in I}r_i\cdot e_i \equiv \bigoplus_{[e]_\equiv \in E/\equiv}\left(\sum_{e_i \equiv e} r_i\right)\cdot e$$
			\label{lem64}
		\end{lem}
		\begin{proof}
			The proof is similar to \cref{lem63} using $\equiv$ in place of $=$.
		\end{proof}
	\end{toappendix}
	\begin{thmrep}[Fundamental Theorem]
		For every $e \in \Exp$ it holds that $$e \equiv \bigplus_{\alpha \in \At} \left(\bigoplus_{d \in \text{supp}(\partial(e)_\alpha)}\partial(e)_\alpha(d) \cdot \ex(d)\right)$$
		where \text{exp} defines a function $2+\Out+\Act\times\Exp \to \Exp$ given by
		$$\exp(\xmark) = \babort \;\;\;\exp(\cmark) = \bskip \;\;\;\ex(v) = v \;\;\;\exp((p, f)) = p;f \;\;\; (v \in \Out, p \in \Act, f \in \Exp)$$\label{fundamentaltheorem}
	\end{thmrep}
	\begin{appendixproof}
	The proof is by induction on the construction of $e$. The base cases are trivial, we illustrate one of them. We also omit the case of guarded choice as it is unchanged from \probgkat \cite{rozowski2023probabilistic}.
	\subparagraph*{Case of: $e\equiv p$}
	\begin{align*}
		p &\equiv \bigplus_{\alpha \in \At} p & \Ax{G1}\\
		&\equiv \bigplus_{\alpha \in \At} p;\bskip & \Ax{S1}\\
		&\equiv \bigplus_{\alpha \in \At}\left(\bigoplus_{d \in \text{supp}(\partial(e)_\alpha)} \partial(e )_\alpha(d)\cdot \text{exp}(d)\right) &
	\end{align*}
	\subparagraph*{Case of: $e \equiv f \WC{r}{s} g$}
	\begin{align*}
		f \WC{r}{s} g &\equiv \bigplus_{\alpha \in \At} \left(\bigoplus_{d \in \text{supp}(\partial(f)_\alpha)}\partial(f)_\alpha(d) \cdot \text{exp}(d)\right) \\&\;\;\WC{r}{s} \bigplus_{\alpha \in \At} \left(\bigoplus_{d \in \text{supp}(\partial(g)_\alpha)}\partial(g)_\alpha(d) \cdot \text{exp}(d)\right)& \text{Ind. hypothesis}\\
		&\equiv\bigplus_{\alpha \in \At} \Biggl( \bigoplus_{d \in \text{supp}(\partial(f)_\alpha)}\partial(f)_\alpha(d) \cdot \text{exp}(d)\\&\;\; \WC{r}{s} \bigoplus_{d \in \text{supp}(\partial(g)_\alpha)}\partial(g)_\alpha(d) \cdot \text{exp}(d)\Biggl)& \cref{lem56}
	\end{align*}
	For each $\alpha$ we can rearrange this inner expression using \Ax{W3} into a single generalized sum. Furthermore, let $I_\alpha \coloneq \text{supp}(\partial(f)_\alpha) + \text{supp}(\partial(g)_\alpha)$, where $+$ is the coproduct. Then let $$r_i = \begin{cases}
		r\partial(f)_\alpha(i)& i \in \text{supp}(\partial(f)_\alpha)\\s\partial(g)_\alpha(i)& i \in \text{supp}(\partial(s)_\alpha)
	\end{cases}$$
	Hence we can define $$f \WC{r}{s} g = \bigplus_{\alpha \in \At}\left(\bigoplus_{i\in I_\alpha} r_i\cdot \text{exp}(i)\right)$$
	So by \cref{lem63} we can combine terms which gives us
	$$\bigplus_{\alpha \in \At}\left(\bigoplus_{d \in \text{supp}(\partial(f)_\alpha) \cup \text{supp}(\partial(g)_\alpha)} \left(r\partial(f)_\alpha(d) + s\partial(g)_\alpha(d)\right)\cdot \text{exp(d)}\right)$$
	Which is the same as
	$$\bigplus_{\alpha \in \At}\left(\bigoplus_{d \in \text{supp}(\partial(f \WC{r}{s} g)_\alpha)} \partial(f \WC{r}{s} g)_\alpha(d)\cdot \text{exp(d)}\right)$$
	
	\paragraph*{Case of: $e= f;g$}
	\begin{align*}
		f;g &\equiv \left(\bigplus_{\alpha \in \At} \left(\bigoplus_{d \in \text{supp}(\partial(f)_\alpha)}\partial(f)_\alpha(d) \cdot \text{exp}(d)\right)\right);g& \text{Inductive hypothesis}\\
		&\equiv \left(\bigplus_{\alpha \in \At} \left(\bigoplus_{d \in \text{supp}(\partial(f)_\alpha)}\partial(f)_\alpha(d) \cdot \text{exp}(d)\right);g\right)& \cref{lem55}\\
		&\equiv \left(\bigplus_{\alpha \in \At} \left(\bigoplus_{d \in \text{supp}(\partial(f)_\alpha)}\partial(f)_\alpha(d) \cdot \text{exp}(d);g\right)\right)& \cref{lem61}
	\end{align*}
	We split this term to get
	\begin{align*}
		&\bigplus_{\alpha \in \At} \left(\partial(f)_\alpha(\cmark) \cdot \bskip;g \WC{}{} \bigoplus_{{\substack{d \in \{\xmark\} \cup \Out}}}\partial(f)_\alpha(d)\cdot \text{exp}(d);g  \right. \\ &\left. \;\;\WC{}{} \bigoplus_{{\substack{d \in \text{supp}(\partial(f)_\alpha)\cap \Act\times\Exp}}}\partial(f)_\alpha(d) \cdot \text{exp}(d);g\right)\\
		&\equiv \bigplus_{\alpha \in \At} \left(\partial(f)_\alpha(\cmark) \cdot g \WC{}{}  \bigoplus_{{\substack{d \in \{\xmark\} \cup \Out}}}\partial(f)_\alpha(d) \cdot \text{exp}(d);g \right. \\ &\left.\;\;\WC{}{} \bigoplus_{{\substack{d \in \text{supp}(\partial(f)_\alpha)\cap \Act\times\Exp}}}\partial(f)_\alpha(d) \cdot \text{exp}(d);g\right) & \Ax{S1}\\
		&\equiv \bigplus_{\alpha \in \At} \left(\partial(f)_\alpha(\cmark) \cdot g \WC{}{}  \bigoplus_{{\substack{d \in \{\xmark\} \cup \Out}}}\partial(f)_\alpha(d) \cdot \text{exp}(d) \right. \\ &\left.\;\;\WC{}{} \bigoplus_{{\substack{d \in \text{supp}(\partial(f)_\alpha)\cap \Act\times\Exp}}}\partial(f)_\alpha(d) \cdot \text{exp}(d);g\right) & \Ax{S3, S6}\\
		&\equiv \bigplus_{\alpha \in \At} \alpha;\left(\partial(f)_\alpha(\cmark) \cdot g \WC{}{}  \bigoplus_{{\substack{d \in \{\xmark\} \cup \Out}}}\partial(f)_\alpha(d) \cdot \text{exp}(d) \right. \\ &\left.\;\;\WC{}{} \bigoplus_{{\substack{d \in \text{supp}(\partial(f)_\alpha)\cap \Act\times\Exp}}}\partial(f)_\alpha(d) \cdot \text{exp}(d);g\right) & \Ax{G2}\\
		&\equiv \bigplus_{\alpha \in \At} \alpha;\left(\partial(f)_\alpha(\cmark) \cdot \alpha;g \WC{}{}  \bigoplus_{{\substack{d \in \{\xmark\} \cup \Out}}}\partial(f)_\alpha(d) \cdot \text{exp}(d) \right. \\ &\left.\;\; \WC{}{} \bigoplus_{{\substack{d \in \text{supp}(\partial(f)_\alpha)\cap \Act\times\Exp}}}\partial(f)_\alpha(d) \cdot \text{exp}(d);g\right) & \Ax{DF13}\\
		&\equiv \bigplus_{\alpha \in \At} \left(\partial(f)_\alpha(\cmark) \cdot \alpha;g \WC{}{} \bigoplus_{{\substack{d \in \{\xmark\} \cup \Out}}}\partial(f)_\alpha(d) \cdot \text{exp}(d) \right. \\ &\left.\;\;\WC{}{} \bigoplus_{{\substack{d \in \text{supp}(\partial(f)_\alpha)\cap \Act\times\Exp}}}\partial(f)_\alpha(d) \cdot \text{exp}(d);g\right) & \Ax{G2}
	\end{align*}
	We can apply the inductive hypothesis and \Ax{DF4} to say 
	$$\alpha;g = \alpha;\bigoplus_{d\in \text{supp}(\partial(g)_\alpha)}\partial(g)_\alpha(d)\cdot \text{exp}(d)$$ The $\alpha$ will be removed by \Ax{G2} as soon as we substitute this back into the above expression.
	We substitute and can express this as one sum. Let
	\begin{align*}I_\alpha &= \text{supp}(\partial(g)_\alpha)+(\text{supp}(\partial(f)_\alpha)\cap (\{\xmark\} \cup \Out))\\&+\{(a, f^\prime;g) \mid (a,f^\prime) \in \text{supp}(\partial(f)_\alpha) \cap \Act\times\Exp\}\end{align*}
	For each $\alpha \in \At$ let $\{r_i\}_{i\in I_\alpha}$ such that 
	$$ r_i  = \begin{cases}
		\partial(f)_\alpha(\cmark)\partial(g)_\alpha(i)& i\in \text{supp}(\partial(g)_\alpha)\\
		\partial(f)_\alpha(i)&i\in \text{supp}(\partial(f)_\alpha\cap(\{\xmark\} \cup \Out))\\
		\partial(f)_\alpha(a,f^\prime)&i=(a,f^\prime;g) \text{ and } (a,f^\prime) \in \text{supp}(\partial(f)_\alpha \cap \Act\times\Exp)
	\end{cases}$$
	Hence we have
	$$f;g \equiv \bigplus_{\alpha \in \At}\left(\bigoplus_{i \in I_\alpha}r_i \cdot \text{exp}(i)\right)$$
	Finally we apply \cref{lem63} to arrive at
	$$f;g \equiv \bigplus_{\alpha \in \At}\left(\bigoplus_{d \in \text{supp}(\partial(f;g)_\alpha)} \partial(f;g)_\alpha(d) \cdot \text{exp}(d)\right)$$
	\paragraph*{Case of: $e \equiv f^{(b)}$}
	We want to show that $$f^{(b)} \equiv \bigplus_{\alpha \in \At}\left(\bigoplus_{d \in \text{supp}(\partial(f^{(b)})_\alpha)}\partial\left(f^{(b)}\right)_\alpha(d)\cdot\ex(d)\right)$$
	Note that due to \cref{lem54} this is equivalent to proving $$\bigplus_{\alpha \in \At} \alpha;f^{(b)} \equiv \bigplus_{\alpha \in \At}\alpha;\left(\bigoplus_{d \in \text{supp}(\partial(f^{(b)})_\alpha)}\partial\left(f^{(b)}\right)_\alpha(d)\cdot\ex(d)\right)$$
	Which is a consequence of the statement $\forall \alpha \in \At$
	$$\alpha;f^{(b)} \equiv \alpha;\left(\bigoplus_{d \in \text{supp}(\partial(f^{(b)})_\alpha)}\partial\left(f^{(b)}\right)_\alpha(d)\cdot\ex(d)\right)$$
	We first consider an atom $\gamma \leq_{\BA} \bar b$. It holds that
	\begin{align*}
		\gamma;f^{(b)} &\equiv \gamma;(f;f^{(b)}+_b \bskip)&\Ax{L1}\\
		&\equiv (f;f^{(b)}+_b \bskip)+_\gamma \babort&\Ax{DF3}\\
		&\equiv f;f^{(b)}+_{b\gamma}(\bskip +_\gamma \babort)&\Ax{G4}\\
		&\equiv f;f^{(b)}+_{\babort}(\bskip +_\gamma \babort)&\gamma \leq_{\BA} \bar b\\
		&\equiv (\bskip+_\gamma \babort)+_{\bskip} f;f^{(b)}&\Ax{G4}\\
		&\equiv (\bskip+_\gamma \babort)&\Ax{DF7}\\
		&\equiv \gamma;\bskip&\Ax{DF3}\\
		&\equiv \gamma;\left(\bigoplus_{d \in \text{supp}(\partial(f^{(b)})_\gamma)}\partial\left(f^{(b)}\right)_\gamma(d)\cdot\ex(d)\right)&
	\end{align*}
	Consider an atom $\gamma \leq_{\BA} b$.
	\begin{align*}
		f &\equiv \bigplus_{\alpha \in \At}\left(\bigoplus_{d \in \text{supp}(\partial(f)_\alpha)}\partial\left(f\right)_\alpha(d)\cdot\ex(d)\right) & \text{Inductive hypothesis}\\
		&\equiv \bigoplus_{d \in \text{supp}(\partial(f)_\gamma)}\partial\left(f\right)_\gamma(d)\cdot\ex(d) \\&\;\;+_\gamma \bigplus_{\alpha \in \At\setminus \gamma}\left(\bigoplus_{d \in \text{supp}(\partial(f)_\alpha)}\partial\left(f\right)_\alpha(d)\cdot\ex(d)\right)\\
		&\equiv \left(\bigoplus_{d \in \text{supp}(\partial(f)_\gamma) \setminus \cmark}\partial\left(f\right)_\gamma(d)\cdot\ex(d) \WC{1}{\partial(f)_\gamma(\cmark)} \bskip\right) \\&\;\;+_\gamma \bigplus_{\alpha \in \At\setminus \gamma}\left(\bigoplus_{d \in \text{supp}(\partial(f)_\alpha)}\partial\left(f\right)_\alpha(d)\cdot\ex(d)\right)
	\end{align*}
	We continue with \Ax{L2} to get
	\begin{align*}
		\gamma;f^{(b)} &\equiv \gamma;\left(\left(\odot{\partial(f)_\gamma(\cmark)}^*;\bigoplus_{\mathclap{\substack{d \in \text{supp}(\partial(f)_\gamma)\setminus \cmark}}}\partial\left(f\right)_\gamma(d)\cdot\ex(d)\right);f^{(b)}+_b \bskip\right)\\
		&\equiv \gamma;\left(\left(\odot{\partial(f)_\gamma(\cmark)}^*;\bigoplus_{\mathclap{\substack{d \in \text{supp}(\partial(f)_\gamma)\setminus \cmark}}}\partial\left(f\right)_\gamma(d)\cdot\ex(d);f^{(b)}\right)+_b \bskip\right)&\cref{lem61}\\
		&\equiv \left(\left(\odot 	{\partial(f)_\gamma(\cmark)}^*;\bigoplus_{\mathclap{\substack{d \in \text{supp}(\partial(f)_\gamma)\setminus \cmark}}}\partial\left(f\right)_\gamma(d)\cdot\ex(d);f^{(b)}\right)+_b \bskip\right)+_\gamma \babort&\Ax{DF3}\\
		&\equiv \left(\odot 	{\partial(f)_\gamma(\cmark)}^*;\bigoplus_{\mathclap{\substack{d \in \text{supp}(\partial(f)_\gamma)\setminus \cmark}}}\partial\left(f\right)_\gamma(d)\cdot\ex(d);f^{(b)}\right)+_{\gamma b} (\bskip +_\gamma \babort)&\Ax{G4}\\
		&\equiv \left(\odot 	{\partial(f)_\gamma(\cmark)}^*;\bigoplus_{\mathclap{\substack{d \in \text{supp}(\partial(f)_\gamma)\setminus \cmark}}}\partial\left(f\right)_\gamma(d)\cdot\ex(d);f^{(b)}\right)+_{\gamma} (\bskip +_\gamma \babort)&\gamma \leq_{\BA} b\\
		&\equiv (\babort +_{\bar\gamma} \bskip) +_{\bar\gamma} \left(\odot {\partial(f)_\gamma(\cmark)}^*;\bigoplus_{\mathclap{\substack{d \in \text{supp}(\partial(f)_\gamma)\setminus \cmark}}}\partial\left(f\right)_\gamma(d)\cdot\ex(d);f^{(b)}\right)&\Ax{G3}\\
		&\equiv (\bar\gamma;(\babort +_{\bar\gamma} \bskip))+_{\bar\gamma} \left(\odot {\partial(f)_\gamma(\cmark)}^*;\bigoplus_{\mathclap{\substack{d \in \text{supp}(\partial(f)_\gamma)\setminus \cmark}}}\partial\left(f\right)_\gamma(d)\cdot\ex(d);f^{(b)}\right)&\Ax{G2}\\
		&\equiv (\bar\gamma;\babort) +_{\bar\gamma} \left(\odot {\partial(f)_\gamma(\cmark)}^*;\bigoplus_{\mathclap{\substack{d \in \text{supp}(\partial(f)_\gamma)\setminus \cmark}}}\partial\left(f\right)_\gamma(d)\cdot\ex(d);f^{(b)}\right)&\Ax{DF4}\\
		&\equiv (\babort+_{\bar\gamma} \babort) +_{\bar\gamma} \left(\odot {\partial(f)_\gamma(\cmark)}^*;\bigoplus_{\mathclap{\substack{d \in \text{supp}(\partial(f)_\gamma)\setminus \cmark}}}\partial\left(f\right)_\gamma(d)\cdot\ex(d);f^{(b)}\right)&\Ax{DF3}\\
		&\equiv \babort +_{\bar\gamma} \left(\odot {\partial(f)_\gamma(\cmark)}^*;\bigoplus_{\mathclap{\substack{d \in \text{supp}(\partial(f)_\gamma)\setminus \cmark}}}\partial\left(f\right)_\gamma(d)\cdot\ex(d);f^{(b)}\right)&\Ax{G1}\\
		&\equiv \left(\odot {\partial(f)_\gamma(\cmark)}^*;\bigoplus_{\mathclap{\substack{d \in \text{supp}(\partial(f)_\gamma)\setminus \cmark}}}\partial\left(f\right)_\gamma(d)\cdot\ex(d);f^{(b)}\right)+_\gamma \babort&\Ax{G3}\\
		&\equiv \gamma;\left(\odot {\partial(f)_\gamma(\cmark)}^*;\bigoplus_{\mathclap{\substack{d \in \text{supp}(\partial(f)_\gamma)\setminus \cmark}}}\partial\left(f\right)_\gamma(d)\cdot\ex(d);f^{(b)}\right)&\Ax{DF3}
	\end{align*}
	Separating the support gives
	\begin{align*}
		\gamma;f^{(b)} &\equiv \gamma;\biggl(\odot {\partial(f)_\gamma(\cmark)}^*;\biggl(\bigoplus_{d \in \{\xmark\} \cup \Out}\partial(f)_\gamma(d)\cdot\ex(d);f^{(b)} \\&\;\;\oplus \bigoplus_{d \in \text{supp}(\partial(f)_\gamma) \cap \Act \times \Exp}\partial(f)_\gamma(d)\cdot\ex(d);f^{(b)}\biggr)\biggr)\\
		&\equiv \gamma;\biggl(\bigoplus_{d \in \{\xmark\} \cup \Out}\partial(f)_\gamma(\cmark)^*\partial\left(f\right)_\gamma(d)\cdot\ex(d);f^{(b)} \\&\;\;\oplus \bigoplus_{d \in \text{supp}(\partial(f)_\gamma) \cap \Act \times \Exp}\partial(f)_\gamma(\cmark)^*\partial\left(f\right)_\gamma(d)\cdot\ex(d);f^{(b)}\biggr)&\Ax{DF1}\\
		&\equiv \gamma;\biggl(\bigoplus_{d \in \{\xmark\} \cup \Out}\partial(f)_\gamma(\cmark)^*\partial(f)_\gamma(d)\cdot\ex(d) \\&\;\;\oplus \bigoplus_{d \in \text{supp}(\partial(f)_\gamma) \cap \Act \times \Exp}\partial(f)_\gamma(\cmark)^*\partial(f)_\gamma(d)\cdot\ex(d);f^{(b)}\biggr)&\Ax{S6, S3}\\
		&\equiv \gamma;\left(\bigoplus_{d \in \text{supp}(\partial(f^{(b)})_\gamma)}\partial(f^{(b)})_\gamma(d)\cdot\ex(d) \right) \rlap{\qedhere} 
	\end{align*}
	\end{appendixproof}
	\subsection{Solutions to Systems of Equations}
	
	\subparagraph*{System of equations} We are now ready to characterize systems of equations and their solutions. We first introduce a three-sorted grammar to mimic the way that atoms index weighted choices in the fundamental theorem.
	\begin{align*}
		\Exp(X) &\coloneq \bigplus_{\alpha \in \At}w_\alpha & (\forall \alpha \in \At \;\; w_\alpha \in \WExp(X))\\
		\WExp(X) &\coloneq \bigoplus_{i \in I} r_i \cdot t_i &(t_i \in \TExp, r_i \in \mathbb S, I \text{ finite})\\
		\TExp(X) &\coloneq \textbf{f} \mid \textbf{g}x &(\textbf{f, g} \in \Exp, x \in X)
	\end{align*}
	Note that the above is well defined only if we take a congruence containing our axioms of both guarded and weighted choice, as the generalized version of guarded choice is defined up to skew-associativity and skew-commutativity. That is, an arbitrary order may be picked when selecting an element of $\Exp(X)$, but after substituting \wgkat expressions for indeterminates $x \in X$ it won't matter which order was selected. See \cref{FTapx} for details.
	\begin{defn}[\cite{rozowski2023probabilistic}]\label{salsys}
		A system of equations is a pair $(X, \tau: X \to \Exp(X))$ consisting of a finite set $X$ of indeterminates and a function $\tau$. If for all $x \in X, \alpha \in \At$ in each $\tau(x)$ the \textbf{g}x subterms satisfy $E(g)_\alpha = 0$, the system of equations is called \emph{Salomaa}.
	\end{defn}
	We now define the system of equations associated with a given \wgkat automaton.
	\begin{defn}
		Let $(X, \beta)$ be a \wgkat automaton. A system of equations associated with $(X, \beta)$ is a Salomaa system $(X, \tau)$, with $\tau: X \to \Exp(X)$ defined by
		 $$
			\tau(x) = \bigplus_{\alpha \in \At}\left(\bigoplus_{d \in \text{supp}(\beta(x))_\alpha} \beta(x)_\alpha(d)\cdot \text{sys}(d)\right)
		$$ 
		where $\text{sys}: 2+\Out+\Act\times\Exp \to \WExp(X)$ is given by
		 $$
			\text{sys}(\xmark) = \babort\;\;\text{sys}(\cmark)=\bskip\;\; \text{sys}(v) = v\;\;\text{sys}(p, x) = px
		$$
	\end{defn}
	\begin{exmp}
		The $\tau$ (up to $\equiv$) associated with the automaton from \cref{automatonexample} is
		\begin{align*}
			x_1 \mapsto& (\babort \WC{4}{3} (p_1x_2)) +_\alpha (v \WC{1}{1} ((p_2x_2)\WC{5}{2}(p_3x_3))+_\beta \babort) \\
			x_2 \mapsto& \bskip \WC{15}{0} \babort +_\alpha (\bskip \WC{15}{0} \babort +_\beta \babort)\\
			x_3 \mapsto& \bskip \WC{15}{0} \babort +_\alpha (\bskip \WC{15}{0} \babort +_\beta \babort)
		\end{align*} \label{systemofequations} \vspace{-.5cm}
	\end{exmp}
	\begin{defn}[\cite{rozowski2023probabilistic}]
		Given a function $h: X \to \Exp$ that assigns a value to each indeterminate in $X$ we derive a \wgkat expression $h^{\#}(e)$ for each $e \in \Exp(X)$ inductively:
		
		\begin{minipage}[t]{.45\textwidth}
			\begin{align}
				h^{\#}(f) &= f \label{HSharpDef1}\\
				h^{\#}(e+_b f) &= h^{\#}(e) +_b h^{\#}(f)\label{HSharpDefGuarded}
			\end{align}
		\end{minipage}
		\begin{minipage}[t]{.5\textwidth}
			\begin{align}
				h^{\#}(w_1 \WC{r}{s} w_2) &= h^{\#}(w_1) \WC{r}{s} h^{\#}(w_2)\label{HSharpDef2}\\
				h^{\#}(\textbf{g}x) &= \textbf{g};h(x)\label{HSharpDef3}
			\end{align}\label{HSharpDef}\vspace{-.5cm}
		\end{minipage}
	\end{defn}
	
	We now characterize a solution to the Salomaa system of equations. 
	\begin{defn}[\cite{rozowski2023probabilistic}]
		Let $R \subseteq \Exp \times \Exp$ be a congruence. A solution up to $R$ of a system $(X, \tau)$ is a map $h: X \to \Exp$ such that for all $x \in X$ that $(h(x), h^{\#}(\tau(x))) \in R$
	\end{defn}
	\begin{exmp}
		A solution up to $\equiv$ of \cref{systemofequations} would satisfy
		\begin{align*}
			h(x_1) &\equiv (\babort \WC{4}{3} p_1;h(x_2)) +_\alpha (v \WC{1}{1} (p_2;h(x_2)\WC{5}{2}p_3;h(x_3))+_\beta \babort) \\
			h(x_2) &\equiv \bskip \WC{15}{0} \babort +_\alpha (\bskip \WC{15}{0} \babort +_\beta \babort)\\
			h(x_3) &\equiv \bskip \WC{15}{0} \babort +_\alpha (\bskip \WC{15}{0} \babort +_\beta \babort)
		\end{align*}
		
		So in this case we can select 
		\begin{align*}
			h(x_2) = h(x_3) &= \odot 15\\
			h(x_1) &=  (\babort \WC{4}{3} (p_1;\odot 15)) +_\alpha (v \WC{1}{1} ((p_2;\odot 15)\WC{5}{2}(p_3;\odot 15))+_\beta \babort)
		\end{align*}
	\end{exmp}
	We now can show the connection between solutions to the Salomaa system of equations and automaton homomorphisms using the equivalence classes of states.
	\begin{toappendix}
		\begin{lem}[{{\cite[5.8]{Rutten2000UniversalCA}}}]\label{uniquemapbisim}
			Let $(X, \beta)$ be an $\mathcal H$-coalgebra, $R \subseteq X \times X$ a bisimulation equivalence on $(X, \beta)$ and $[-]_R:X \to X/R$ the canonical quotient map. Then, there exists a unique transition map $\bar\beta: X/R \times \At \to \Mon(2+\Out+\Act\times X/R)$ which makes $[-]_R$ into an $\mathcal{H}$-coalgebra homomorphism from $(X, \beta)$ to $(X/R, \bar\beta)$. \label{lem65}
		\end{lem}
	\end{toappendix}
	\begin{thmrep}
		Let $(X, \beta)$ be a finite \wgkat automaton. The map $h:X\to \Exp$ is a solution up to $\equiv$ of the system associated with $(X, \beta)$ if and only if $[-]_\equiv \circ h$ is a \wgkat automaton homomorphism from $(X, \beta)$ to $(\Exp/\equiv, \bar\partial)$. Where $\bar\partial$ is the unique transition function on $\Exp/\equiv$ which makes the quotient map $[-]_\equiv:\Exp \to \Exp/\equiv$ a \wgkat automaton homomorphism from $(\Exp, \partial)$ to $(\Exp/\equiv, \bar\partial)$.  \label{thm21}
	\end{thmrep}
	\begin{appendixproof}		
		This theorem and its proof are akin to \cite[Theorem 21]{rozowski2023probabilistic}. The theorem follows from \cref{uniquemapbisim}. Like \cref{SoundnessMorphism} we appeal to products over all $\alpha \in \At$ to construct $\mathcal H$, since $\At$ is finite.
		
		Let $\bar\partial: \Exp/\equiv \to \mathcal{H}(\Exp/\equiv)$ be the unique $\mathcal{H}$-coalgebra structure map from \cref{lem65} which makes the quotient map $[-]_\equiv:\Exp\to\Exp/\equiv$ into an $\mathcal H$-coalgebra homomorphism. Then given a function $h: X \to \Exp$, $[-]_\equiv \circ h$ is a \HCo homomorphism from $(X, \beta)$ to $(\Exp/\equiv, \bar\partial)$ if and only if the following diagram commutes.
		\begin{center} \begin{tikzcd}
			X && {\Exp/\equiv} \\
			\\
			{\mathcal{H}X} && {\mathcal{H}(\Exp/\equiv)}
			\arrow["{[-]_\equiv\circ h}"{description}, from=1-1, to=1-3]
			\arrow["\beta", from=1-1, to=3-1]
			\arrow["\bar\partial", from=1-3, to=3-3]
			\arrow["{\mathcal{H}([-]_\equiv\circ h)}"{description}, from=3-1, to=3-3]
		\end{tikzcd} \end{center}
		
		That is: $\bar\partial \circ [-]_\equiv \circ h = \mathcal{H}([-]_\equiv\circ h) \circ \beta$. Since $[-]_\equiv$ is a coalgebra homomorphism $\mathcal H[-]_\equiv \circ \partial = \bar\partial \circ [-]_\equiv$ and the previous statement holds if and only if $$\mathcal{H}[-]_\equiv \circ \partial \circ h = \mathcal{H}([-]_\equiv \circ h) \circ \beta$$ That is if for all $x \in X, \alpha \in \At, o \in 2+\Out, (a,[e^\prime]) \in \Act \times \Exp/\equiv$
		\begin{align}
			\beta(x)_\alpha(o) &= \partial(h(x))_\alpha(o) \label{cond1.21}\\
			\sum_{h(x^\prime) \equiv e^\prime}\beta(x)_\alpha(a, x^\prime) &= \sum_{f \equiv e^\prime}\partial(h(x))_\alpha(a,f)\label{cond2.21}
		\end{align}
		We start by proving the converse. Assume that $[-]_\equiv \circ h$ is an \HCo homomorphism. Let $(X, \tau)$ be the Salomaa system associated with $(X, \beta)$. To show $h$ is a solution we will show that $h(x)\equiv (h^{\#} \circ \tau)(x)$ for all $x \in X$. By \cref{fundamentaltheorem} it holds that for all $x \in X$ that
		\begin{align*}
			h(x) \equiv \bigplus_{\alpha\in\At}\left(\bigoplus_{d \in \text{supp}(\partial(h(x))_\alpha)} \partial(h(x))_\alpha(d)\cdot\exp(d)\right)
		\end{align*}
		We can unroll this and use \cref{lem64} to obtain for all $x \in X$
		\begin{align}
			h(x) = \bigplus_{\alpha \in \At}\left(\bigoplus_{d \in 2 + \Out} \partial(h(x))_\alpha(d)\cdot d \oplus \bigoplus_{(a, [e^\prime]_\equiv)\in E_\alpha} \left(\sum_{f \equiv e^\prime} \partial(h(x))_\alpha(a,f)\right)\cdot a;e\right) \label{exp1.21}
		\end{align}
		Where for each $\alpha \in \At$ $$E_\alpha = \{(a, Q) \in \Act \times \Exp/\equiv \mid \exists e^\prime\;\; e^\prime \in Q, \partial(h(x))_\alpha(a, e^\prime) \neq 0\}$$
		We can write the other side of the equation, for all $x \in X$
		$$(h^{\#}\circ \tau)(x) \equiv \bigplus_{\alpha \in \At}\left(\bigoplus_{d \in 2+\Out} \beta(x)_\alpha(d)\cdot d \oplus \bigoplus_{(a,f) \in \text{supp}(\beta(x)_\alpha)} \beta(x)_\alpha(a, f)\cdot a;h(f)\right)$$
		Applying \cref{lem64} we get 
		\begin{align}
			(h^{\#}\circ \tau)(x) \equiv \bigplus_{\alpha \in \At}\left(\bigoplus_{d \in 2+\Out} \beta(x)_\alpha(d)\cdot d \oplus \bigoplus_{[e]_\equiv \in F_\alpha}\left(\sum_{h(x^\prime) \equiv e\prime} \beta(x)_\alpha(a,x^\prime)\right)\cdot a;e^\prime\right) \label{exp2.21} 
		\end{align}
		Where for all $\alpha \in \At$
		$$F_\alpha = \{(a,Q) \in \Act\times \Exp/\equiv \mid\exists x^\prime\;\; h(x^\prime) \in Q, \beta(x)_\alpha(a, x^\prime)\neq0\}$$
		We argue that \cref{exp1.21} and \cref{exp2.21} are the same. We do this by showing that the indexed collections that form each sum are equivalent for any fixed $x \in X$ and $\alpha \in \At$.
		For any $\alpha \in \At$ given $d \in 2+\Out$ if $d \in \text{supp}(\partial(h(x)))_\alpha$, then because of \cref{cond1.21} $d \in \text{supp}(\beta(x)_\alpha)$. The converse holds by symmetry.  Furthermore, for all such $d$, again by condition \cref{cond1.21} $\beta(x)_\alpha(d) = \partial(h(x))_\alpha(d)$.
		Consider $(a, Q) \in \Act\times\Exp/\equiv$. If $(a, Q) \in E_\alpha$ then there is some $e^\prime \in \Exp$ such that $\partial(h(x))_\alpha(a,e^\prime) \neq 0$. Because of \cref{cond2.21} there exists some $x^\prime \in X$ satisfying $h(x^\prime) \equiv e^\prime$ and $\beta(x)_\alpha(a,x^\prime) \neq 0$. The converse holds by symmetry. Hence $F_\alpha = E_\alpha$. Similarly also by \cref{cond2.21} the weights associated with each $(a,Q) \in \Act\times \Exp/\equiv$ are the same.
		Therefore $h(x) = (h^{\#}\circ \tau)(x)$ for all $x \in X$
		To prove the other direction we assume $h$ is a solution up to $\equiv$. First we show that $(\partial\circ h^{(\#)}\circ \tau)(x) = (\mathcal{H}(h) \circ \beta)(x)$ for all $x \in X$. For all $x \in X$
		$$(h^{\#}\circ \tau)(x) \equiv \bigplus_{\alpha \in \At}\left(\bigoplus_{d \in 2+\Out} \beta(x)_\alpha(d)\cdot d \oplus \bigoplus_{(a,f) \in \text{supp}(\beta(x)_\alpha)} \beta(x)_\alpha(a, f)\cdot a;h(f)\right)$$
		When we apply the transition map, we can split into two cases for all $\alpha \in \At$.
		\begin{enumerate}
			\item For $o \in 2+\Out$, $(\partial_\alpha \circ h^{\#}\circ \tau)(x)(o) = \beta(x)_\alpha(o) = (\mathcal H (h) \circ \beta)(x)_\alpha(o)$
			\item For all $(a, e^\prime) \in \Act \times \Exp$, 
			$$(\partial_\alpha \circ h^{\#}\circ \tau)(x)(a,e^\prime) = \sum_{h(x^\prime) = e^\prime}\beta(x)_\alpha(a,x^\prime) = (\mathcal H (h) \circ \beta)(x)_\alpha(a,e^\prime)$$
		\end{enumerate}
		Hence $(\partial \circ h^{\#}\circ \tau) = (\mathcal H (h) \circ \beta)$. Composing with $\mathcal{H}[-]_\equiv$ gives
		\begin{align*}
			\mathcal{H}[-]_\equiv \circ \partial \circ h^{\#}\circ \tau &= \mathcal{H}[-]_\equiv \circ \mathcal{H}(h) \circ \beta\\
			\mathcal{H}[-]_\equiv \circ \partial \circ h^{\#}\circ \tau &= \mathcal{H}([-]_\equiv \circ h) \circ \beta & \text{Functor Associativity}\\
			\bar\partial \circ [-]_\equiv \circ h^{\#}\circ \tau &= \mathcal{H}([-]_\equiv \circ h) \circ \beta & [-]_\equiv \text{ is a homomorphism}\\
			\bar\partial \circ [-]_\equiv \circ h &= \mathcal{H}([-]_\equiv \circ h) \circ \beta & h^{\#}\circ \tau \equiv h
		\end{align*}
		Therefore $[-]_\equiv \circ h$ is an \HCo  homomorphism from $(X, \beta)$ to $(\Exp/\equiv \bar\partial)$.
	\end{appendixproof}

	\subsection{Bisimilarity implies Provable Equivalence} 
	Having characterized systems of equations and their solutions we now present the uniqueness axiom and establish completeness. Let $\dot\equiv \subseteq \Exp \times \Exp$ be the least congruence that contains $\equiv$, and satisfies the following axiom:
	\begin{align*}
		&\frac{(X, \tau) \text{ is a Salomaa system}\;\;f,g:X\to\Exp \text{ are solutions of }(X, \tau) \text{ up to } \dot\equiv}{f(x) \dot\equiv g(x)\; \text{for all } x\in X} \tag{\textsf{UA}}
	\end{align*}
	
	We establish the soundness of the axiom in the same way as before.
	\begin{toappendix}
		\begin{lem}
			The congruence $\dot\equiv$ is the kernel of a coalgebra homomorphism \label{dotKernel}
		\end{lem}
		\begin{proof}
			Like in the previous proof of soundness we need to show the inclusion of the kernel. Specifically we need to show that $e \dot\equiv f \implies \mathcal{H}([-]_{\dot\equiv}) \circ \partial(e) = \mathcal{H}([-]_{\dot\equiv})\circ \partial(f)$. This proof by induction is identical to \cref{SoundnessMorphism} except for one new case.
			Let $f,g$ be solutions up to $\dot\equiv$ of the Salomaa system $(X, \tau)$. Then it is true that $\forall x \in X \;\; f(x) \dot\equiv g(x)$.
			Furthermore, by definition of being a solution $\forall x \in X\;\; f(x) = f^{\#}(\tau(x))$  (with symmetry for $g$) hence $f = f^{\#} \circ \tau$, and clearly
			\begin{align*}
				\mathcal{H}([-]_{\dot\equiv}) \circ \partial \circ f &= \mathcal{H}([-]_{\dot\equiv}) \circ \partial \circ f^{\#} \circ \tau\\
				\mathcal{H}([-]_{\dot\equiv}) \circ \partial \circ g &= \mathcal{H}([-]_{\dot\equiv}) \circ \partial \circ g^{\#} \circ \tau
			\end{align*}
			From this we need to show
			\begin{align*}
				\mathcal{H}([-]_{\dot\equiv}) \circ \partial \circ f = \mathcal{H}([-]_{\dot\equiv}) \circ \partial \circ g
			\end{align*}
			Which we will do by showing that $\forall \alpha \in \At$ $$W[-]_{\dot\equiv} \circ \partial_\alpha \circ f = W[-]_{\dot\equiv} \circ \partial_\alpha \circ g$$
			Let $(X, \tau)$ be a Salomaa system. $\forall \alpha \in \At, x \in X$ 
			\begin{align*}
				&(W[-]_{\dot\equiv} \circ \partial_\alpha \circ f)(x) = \left(W[-]_{\dot\equiv} \circ \partial_\alpha \circ f^\#\circ \tau\right)(x)\\
				&=(W[-]_{\dot\equiv} \circ \partial_\alpha \circ f^\#)\left(\bigplus_{\alpha \in \At}\bigoplus_{i\in I_\alpha} r_i\cdot \begin{cases}\textbf{g}x_{i}\\\textbf{f}\end{cases}\right)&\\
				&=(W[-]_{\dot\equiv} \circ \partial_\alpha)\left(\bigplus_{\alpha \in \At}\bigoplus_{i\in I_\alpha} r_i\cdot f^{\#}\begin{cases}\textbf{g}x_{i}\\\textbf{f}\end{cases}\right)&\\
				&=(W[-]_{\dot\equiv} \circ \partial_\alpha)\left(\bigplus_{\alpha \in \At}\bigoplus_{i\in I_\alpha} r_i\cdot \begin{cases}\textbf{g};f(x_i)\\\textbf{f}\end{cases}\right)&\\
				&=(W[-]_{\dot\equiv})\left(\sum_{i \in I_\alpha} r_i\cdot \partial_\alpha\begin{cases}\textbf{g};f(x_i)\\\textbf{f}\end{cases}\right)&
			\end{align*}
			In the first case
			\begin{align*}
				&(W[-]_{\dot\equiv})\left(\sum_{i \in I_\alpha} r_i\cdot \partial(\textbf{g};f(x_i))_\alpha\right)&\\
				&=W[-]_{\dot\equiv}
				r_i\sum_{i \in I_\alpha}\bigg( \partial(g)_\alpha(\cmark)\partial(f(x_i))_\alpha \\&\;\;
				+ \sum_{o \in \{\xmark\} + \Out}\partial(\textbf{g})_\alpha(o)\delta_o 
				 + \sum_{(p, e^\prime) \in \Act\times\Exp} \partial(\textbf{g})_\alpha(p, e^\prime)\delta_{p, e^\prime;f(x_i)}\bigg)&\\
				&=r_i\sum_{i \in I_\alpha}\bigg( \partial(g)_\alpha(\cmark)W[-]_{\dot\equiv}\partial(f(x_i))_\alpha \\&\;\;+ \sum_{o \in \{\xmark\} + \Out}\partial(\textbf{g})_\alpha(o)W[-]_{\dot\equiv}\delta_o \\&\;\;+ \sum_{(p, e^\prime) \in \Act\times\Exp} \partial(\textbf{g})_\alpha(p, e^\prime)W[-]_{\dot\equiv}\delta_{p, e^\prime;f(x_i)}\bigg)&\\
				&=r_i\sum_{i \in I_\alpha} \bigg(\sum_{o \in \{\xmark\} + \Out}\partial(\textbf{g})_\alpha(o)W[-]_{\dot\equiv}\delta_o \\&\;\;+ \sum_{(p, e^\prime) \in \Act\times\Exp} \partial(\textbf{g})_\alpha(p, e^\prime)W[-]_{\dot\equiv}\delta_{p, e^\prime;f(x_i)}\bigg)&\cref{salsys}\\
				&=r_i\sum_{i \in I_\alpha} \bigg(\sum_{o \in \{\xmark\} + \Out}\partial(\textbf{g})_\alpha(o)W[-]_{\dot\equiv}\delta_o \\&\;\;+ \sum_{(p, e^\prime) \in \Act\times\Exp} \partial(\textbf{g})_\alpha(p, e^\prime)W[-]_{\dot\equiv}\delta_{p, e^\prime;g(x_i)}\bigg)&\\
				&=(W[-]_{\dot\equiv})\left(\sum_{i \in I_\alpha} r_i\cdot \partial(\textbf{g};g(x_i))\right)& \text{symmetry}
			\end{align*}
			Hence
			\begin{align*}
				&(W[-]_{\dot\equiv})\left(\sum_{i \in I_\alpha} r_i\cdot \partial\begin{cases}\textbf{g};f(x_i)\\\textbf{f}\end{cases}\right)\\
				&=(W[-]_{\dot\equiv})\left(\sum_{i \in I_\alpha} r_i\cdot \partial\begin{cases}\textbf{g};g(x_i)\\\textbf{f}\end{cases}\right)\\
				&=(W[-]_{\dot\equiv} \circ \partial_\alpha \circ g)(x) & \text{symmetry}
			\end{align*}
	
			Hence $\forall \alpha \in \At\; \forall x \in X\;\; (W[-]_{\dot\equiv}\circ \partial_\alpha \circ f) x = (W[-]_{\dot\equiv}\circ \partial_\alpha \circ g) x$ hence $\forall \alpha \in \At \;\; W[-]_{\dot\equiv}\circ \partial_{\alpha} \circ f = W[-]_{\dot\equiv}\circ \partial_{\alpha} \circ g$, which implies that $\prod_{\alpha\in\At} W[-]_{\dot\equiv}\circ \partial_{\alpha} \circ f = \prod_{\alpha\in\At}W[-]_{\dot\equiv}\circ \partial_{\alpha} \circ g$ which of course means that 
			$$\mathcal{H}([-]_{\dot\equiv})\circ \partial \circ f = \mathcal{H}([-]_{\dot\equiv})\circ \partial \circ g\qedhere$$
		\end{proof}
	\end{toappendix}
	\begin{toappendix}
		\begin{cor} \label{behEquivUA}
			For $e,f \in \Exp$, $e \dot\equiv f \implies !_\partial e=!_\partial f$
		\end{cor}
		\begin{proof}
			Follows from \cref{dotKernel} with the same argument as \cref{equivimpliesbeh}
		\end{proof}
	\end{toappendix}
	
	\begin{thmrep}[Soundess of $\dot\equiv$]
		For all $e,f \in \Exp,\;\; e \dot\equiv f \implies e \sim f$
	\end{thmrep}
	\begin{appendixproof}
		Immediate from \cref{behEquivUA} and \cref{bisimIffBeh} 
	\end{appendixproof}
	
	Finally, we can prove completeness using the uniqueness axiom.
	
	\begin{thmrep}[Completeness]
		For all $e,f \in \Exp,\;\; e \sim f\implies e \dot\equiv f$
	\end{thmrep}
	\begin{appendixproof}
		This fact now follows for the same reason as \probgkat \cite{rozowski2023probabilistic}.
		There exists some bisimulation $(R, \delta)$ on $\langle e\rangle_\partial \times \langle f \rangle_\partial$ which witnesses their equivalence. Clearly $R$ is finite as $\langle e\rangle_\partial$ and $\langle f\rangle_\partial$ are both finite due to \ref{lem7}. Let $\pi_1, \pi_2$ be the projection mappings from $R$ to $(\langle e\rangle_\partial, \partial)$ and $(\langle f\rangle_\partial, \partial)$. Let $j, k$ be the inclusion maps from $(\langle e\rangle_\partial, \partial)$ and $(\langle f\rangle_\partial, \partial)$ into $(\Exp, \partial)$. By \cref{thm21} $j \circ \pi_1$, and $k \circ \pi_2$ are solutions of the Salomaa system associated with $(R, \delta)$ up to $\equiv$, since they may be composed with $[-]_\equiv$. They are also solutions up to $\dot\equiv$ as $\equiv \subseteq \dot\equiv$. Hence by \Ax{UA} $\forall (g,h) \in R \;\;(j\circ \pi_1)(g,h) \dot\equiv (k \circ \pi_2)(g,h)$. So it follows that
		$$e \dot\equiv j(e)\dot\equiv (j\circ \pi_1)(e,f) \dot\equiv (k\circ \pi_2)(e,f) \dot\equiv k(f) \dot\equiv f\qedhere$$
	\end{appendixproof}
	\section{Decidability and Complexity}\label{decisionS}
	In practice, one often cares about a mechanistic way to check program equivalence. We show that bisimulation equivalence for \wgkat expressions is indeed decidable, and we offer an efficient procedure to decide it. This sets our language apart from other weighted \kat variants, in particular \kawt, which does not offer a decision procedure for equivalence and is likely not decidable due to using trace equivalence. The key insight is that bisimulation equivalence lets us appeal to known theory of monoid-weighted transition systems and achieve a decision procedure which, remarkably, is polynomial time.
	
	Similar to \probgkat \cite{rozowski2023probabilistic}, we rely on a \emph{coalgebraic} generalization \cite{10.1007/978-3-030-30942-8_18} of classic partition refinement \cite{KanPartRefine, partrefine} to minimize the automaton under bisimilarity which essentially computes the largest bisimulation.  
	Coalgebraic partition refinement offers a procedure to minimize any coalgebra that fits a grammar of supported functors and polynomial constructors \cite{10.1007/978-3-030-30942-8_18}. As mentioned before, \wgkat can be viewed as coalgebras. So in order to establish the decidability of bisimulation equivalence for \wgkat automata we show how to express the equivalent coalgebra in the grammar which can be decided by coalgebraic partition refinement. However, in order to determine a bound on the runtime of this procedure we must bound some of the internal operations of the coalgebraic partition refinement algorithm. 
	
	To bound the runtime of the coalgebraic partition refinement algorithm we encode our coalgebra in a format with a uniform label set, and then bound the comparison of labels, the \textbf{init} step, and the \textbf{update} step \cite{10.1007/978-3-030-30942-8_18}. The \textbf{init} step essentially takes the set of labels on edges from a state and computes the total weight of those edges. The \textbf{update} step basically uses a state's labels into a set of states $S$, and the weight of the same state's edges into another set of states $B$ and computes the weight of edges into $S$ and $C/S$, as well as computing transitions into blocks of related states \cite{10.1007/978-3-030-30942-8_18, dorsch2017efficientcoalgebraicpartitionrefinement}. Once we bound these operations we can leverage a complexity result \cite[Theorem 3.4]{10.1007/978-3-030-30942-8_18} of coalgebraic partition refinement to achieve the desired runtime bound. The details of both of these arguments are in \cref{DecideAppendix}.
	
	\begin{toappendix} 
		\label{DecideAppendix}
		As with \gkat and \probgkat we fix the set of atoms beforehand to avoid being able to encode boolean unsatisfiability.
	\begin{thm}[Decidability]
		$e\dot\equiv f$ is decidable
	\end{thm}
	\begin{proof}
		$\mathcal H$ is expressible using the polynomial functor and monoid weighting functor. Hence it conforms to the grammar of functors for which we know bisimulation equivalence in their associated coalgebras is decidable via coalgebraic partition refinement \cite{10.1007/978-3-030-30942-8_18}. 
	\end{proof}

		In order to characterize the complexity of deciding bisimilarity, we examine some of the internals of the partition refinement algorithm by defining some mappings. The runtime of coalgebraic partition refinement on a given coalgebra is determined by the runtime of two mappings \textbf{init} and \textbf{update}, as well as how much time it takes to compare weights. In order to characterize the required mappings and weights we need to determine a refinement interface for $\mathcal{H}$. This is done by encoding the functor and its associated coalgebra \cite[Section 8]{wissmann2020efficient}. 
		
		\begin{defn}[\cite{10.1007/978-3-030-30942-8_18}]The encoding of the functor $F$ is $(A, \flat)$, where $A$ is a set of labels and $\flat$ is is a family of maps $\flat:FX \to \mathcal{B}_\omega (A \times X)$, one for every set $X$.
		\end{defn}
		$\flat$ can be interpreted as creating the labeled edges from the codomain of the transition function. 
		For each $X$ let $!$ be the unique map $X\to 1$, we let $1$ be a fixed singleton set $\{*\}$ \cite{10.1007/978-3-030-30942-8_18}. Furthermore:
		
		\begin{defn}[\cite{10.1007/978-3-030-30942-8_18}]The encoding of the $F$-coalgebra $(X, \psi)$ is given by $\langle F!,\flat\rangle\cdot \psi: X \to F1 \times \mathcal B_\omega(A \times X)$.
		\end{defn}
		
		We say that the coalgebra has $|X|$ states, and $\sum_{x\in X} |\flat(c(x))|$ edges \cite{10.1007/978-3-030-30942-8_18}. The refinement interface will be defined in reference to the following mappings which essentially sort elements of the coalgebra according to their relationship to the subsets $S$ and $B$. 
		
		\begin{defn}[\cite{10.1007/978-3-030-30942-8_18}]
			If $S\subseteq B \subseteq X$ then let $\chi^B_S:X \to 3$ such that
	$$			\chi^B_S(x \in S) = 2 \;\;\;\;
				\chi^B_S(x \in B\setminus S) = 1 \;\;\;\;
				\chi^B_S(x \in X \setminus B) = 0$$
		\end{defn}
		\begin{defn}[Refinement Interface \cite{10.1007/978-3-030-30942-8_18}]
			Given an encoding $(A, \flat)$ of the set functor $F$, a refinement interface for $F$ is the triple $(W, \textbf{init}, \textbf{update})$. Where $W$ is a set of weights, $\textbf{init}:F1 \times \mathcal{B}_\omega A \to W$, $\textbf{update}:\mathcal{B}_\omega A \times W \to W \times F3 \times W$ such that there exists a family of weight maps $w: \mathcal{P}X \to (FX \to W)$ such that for all $t \in FX$, and $S\subseteq B \subseteq X$
			\begin{align*}
				w(X)(t) &= \textbf{init}(F!(t), \mathcal{B}_\omega\pi_1(\flat(t)))\\
				(w(S)(t), F\chi_S^B(t), w(B\setminus S)(t)) &= \textbf{update}(\{a \mid (a, x) \in \flat(t), x \in S\}, w(B)(t))
			\end{align*}
		\end{defn}
		The maps $w$ can be thought of as tracking the weight into a set of states from the chosen element of the coalgebra. For example $w(B)(t)$ for $B \subseteq X, t \in FX$ is the weight into $B$ from $t$. 
		
		Since $\mathcal H$ is a composite functor we can obtain the refinement interface by flattening and desorting $\mathcal H$ \cite[Section 8]{wissmann2020efficient} to obtain a set functor which is a coproduct of basic functors in the provided grammar. By doing this, the refinement interface can be obtained by piece-wise combining the refinement interfaces of the basic functors. 
		
		In our case $\mathcal H$ flattens into three polynomial functors and a monoid weighting functor. Consider the \HCo $(Q,\tau)$. We first flatten $\mathcal H$ to the multisorted $Set^4$ functor
		$$\bar H(X_1, X_2, X_3, X_4) = (X_2^\At, \Mon^{X_3}, 2+\Out+X_4, \Act\times X_1)$$ We desort this multisorted coalgebra into a coalgebra for $\bar X =X_1 + X_2 + X_3 +X_4$ and coproduct the resulting refinement interfaces to obtain a refinement interface for the desorted coalgebra \cite[Section 8]{wissmann2020efficient}. Each of the basic subfunctors has an existing refinement interface \cite{wissmann2020efficient, 10.1007/978-3-030-30942-8_18}. So we simply count states and transitions. There are $|\At||Q|, |\At||Q|, |\Act||Q|, |Q|$ states in $X_2, X_3, X_4$ and $X_1$ respectively. There are $|Q||\At|, |Q||\At|, k, |\Act||Q|$ edges from $X_1 \to X_2, X_2 \to X_3, X_3 \to X_4, \text{ and } X_4 \to X_1$  respectively. Where $$k = \sum_{q \in Q, \alpha \in \At, o \in 2+\Out+\Act \times Q} [\tau(q)_\alpha(o) \neq 0]$$
		
		In order to know the complexity we must combine the refinement interfaces as per \cite[8.18]{wissmann2020efficient} which allows us to bound the runtime of \textbf{init} and \textbf{update}, as well as comparisons of terms of type $W$ and then apply \cite[Theorem 3.4]{10.1007/978-3-030-30942-8_18}. The labels and weights are the coproducts of the labels and weights of each subfunctor \cite[8.18]{wissmann2020efficient}. Furthermore, the maps $\flat$ and $w$ are component-wise functions build from the original $\flat_i, w_i$ maps from the subfunctors \cite[8.18]{wissmann2020efficient}. The same is true for \textbf{init} and \textbf{update} \cite[8.18]{wissmann2020efficient}. So we simply bound the runtimes of \textbf{init} and \textbf{update} by bounding the worst-case refinement interface of the underlying functions.
		
		\begin{thm}
			If $\At$ is fixed, then the bisimilarity of states in a \wgkat automaton $(Q, \tau)$ is decidable in time $O((|Q|\cdot|\Out|+|\Act|\cdot|Q|^2)(\log^2(|\Out|+|\Act|\cdot|Q|)))$  \label{runtimeThm}
		\end{thm}
		\begin{proof}
			The bound for comparison of elements of type $W$, and the runtime of \textbf{init} and \textbf{update} is dominated by the case of the monoid weighted functor. This is because the comparison of types, and the computation on \textbf{init} and \textbf{update} for the (bounded) polynomial functor runs in constant time \cite[Example 6.11]{wissmann2020efficient}. So we will bound by this worst case. Let $n$ and $m$ be the number of states and transitions in our flattened coalgebra.
			
			We have $p(n,m) = O(\log m)$ for a general monoid \cite[Proof of Corollary 5.5]{10.1007/978-3-030-30942-8_18}  which dominates the constant $p(n,m)$ factor of our other subfunctors \cite[Section 3.2]{10.1007/978-3-030-30942-8_18}. We  use these bounds to achieve an overall runtime of $O((m+n)\log(n)\log(m))$ \cite[Theorem 3.4]{10.1007/978-3-030-30942-8_18}. $m$ dominates $n$ as it is strictly larger, since we can trim unreachable states. Therefore the overall runtime is bounded by $O(m\log^2(m))$. Additionally $m$ is dominated by $|k|$ in the worst case. So we substitute $|k|$ from $m$ and simplify.
			\begin{align*}
				&O(|Q|\cdot|\At|)(2+|\Out|+|\Act|\cdot|Q|) \log(m)^2 = O(|Q|\cdot|\Out|+|\Act|\cdot|Q|^2) \log(m)^2\\
				&\leq O((|Q|\cdot|\Out|+|\Act|\cdot|Q|^2)(\log^2 ((|Q|\cdot|\At|)(2+|\Out|+|\Act|\cdot|Q|)))\\
				&= O((|Q|\cdot|\Out|+|\Act|\cdot|Q|^2)(\log(|Q|\cdot|\At|) + \log(2+|\Out|+|\Act|\cdot|Q|))^2)\\
				&\leq O((|Q|\cdot|\Out|+|\Act|\cdot|Q|^2)(\log(|Q|) + \log(|\At|)+\log(2+|\Out|+|\Act|\cdot|Q|))^2)\\
				&= O((|Q|\cdot|\Out|+|\Act|\cdot|Q|^2)(\log(|Q|)+\log(2+|\Out|+|\Act|\cdot|Q|))^2)\\
				&\leq O((|Q|\cdot|\Out|+|\Act|\cdot|Q|^2)\log(2+|\Out|+|\Act|\cdot|Q|)^2)\\
				&\leq O((|Q|\cdot|\Out|+|\Act|\cdot|Q|^2)(\log(2|\Out|+2|\Act|\cdot|Q|))^2)\\
				&= O((|Q|\cdot|\Out|+|\Act|\cdot|Q|^2)(\log(2)+\log(|\Out|+|\Act|\cdot|Q|))^2)\\
				&= O((|Q|\cdot|\Out|+|\Act|\cdot|Q|^2)\\&\;\;\cdot(\log^2(2)+2\log(2)\log(|\Out|+|\Act|\cdot|Q|)+\log^2(|\Out|+|\Act|\cdot|Q|)))\\
				&= O((|Q|\cdot|\Out|+|\Act|\cdot|Q|^2)(\log^2(|\Out|+|\Act|\cdot|Q|)))\rlap{\qedhere} 
			\end{align*}
		\end{proof}
		\end{toappendix}
		\begin{correp}[Decidability]
			If $e, f \in \Exp$, $n = \# e + \# f$, and $\At$ is fixed, then \wgkat equivalence of $e$ and $f$ is decidable in time $O(n^3\log^2(n))$ \label{runtime}
		\end{correp}
		\begin{appendixproof}
			Unreachable states need not be included in a bisimulation, therefore, if it exists, we can find the required bisimulation in the smallest subcoalgebra of $(\Exp, \partial)$ containing the two expressions we are deciding equivalence of.
			
			By \cref{lem7}, $\langle e, f\rangle_\partial$ has no more than $\# e + \# f$ reachable states. Furthermore, the number of distinct actions and outputs from $\langle e, f\rangle_\partial$ are each no more than $n$ also per \cref{lem7}. Hence $|\Out|, |\Act|, |Q| \leq n$, and hence $|Q|^2\cdot|\Act| \leq n^3$. We assume that $\At$ is fixed. Hence, from \cref{runtimeThm}. 
			\begin{align*}
				O((|Q|\cdot|\Out|+|\Act|\cdot|Q|^2)&\cdot\log^2 (|\Out|+|\Act|\cdot|Q|))\\ 
				&\leq O((n^2+n^3)\cdot\log^2 (n+n^2))\\
				&= O(n^3\log^2 (n+n^2))\\
				&\leq O(n^3\log^2 (2n^2))\\
				&= O(n^3(\log(2n^2))^2)\\
				&= O(n^3(\log(n^2)+\log(2))^2)\\
				&= O(n^3(\log^2(n^2)+2\log(n^2)\log(2)+\log^2(2)))\\
				&= O(n^3(\log^2(n^2)))\\
				&= O(n^34\log^2(n))\\
				&= O(n^3\log^2(n))\rlap{\qedhere} 
			\end{align*}
		\end{appendixproof}
	\begin{rem}The astute reader will know that the equivalence of rational power series in the tropical semiring is undecidable \cite{tropical_undecideable} and might reasonably doubt the decidability of equivalence of \wgkat expressions. This does not pose an issue for \wgkat. Our notion of equivalence is one of bisimilarity, and hence single step equivalence, not trace equivalence. Due to this, deciding equivalence of \wgkat terms over the tropical semiring reduces to deciding equivalence of terms over the \emph{additive monoid} of the tropical semiring, not the tropical semiring itself. For this reason equivalence of \wgkat automata can be decided without deciding equivalence of rational power series. For this reason equivalence with respect to bisimilarity of \wgkat expressions is decidable even over the tropical semiring. \end{rem}
	\section{Related Work}\label{relatedS}
	Our work comes in a long line of work on Kleene Algebra with Tests (\kat) \cite{KAT}. This is an algebraic framework for reasoning about programs based on Kleene Algebra, the algebra of regular expressions. The original model was one with the axioms of Kleene Algebra combined with a boolean algebra of tests, giving control flow to a algebraic model. A restriction of this language to the deterministic fragment led to \gkat which allows for decision of equivalence in near linear time \cite{Smolka_2019}. Following this, a deterministic probabilistic language, \probgkat, was developed which was also efficiently decidable \cite{rozowski2023probabilistic}. Our work is directly inspired by \probgkat, although it is orthogonal as an extension of GKAT with semiring weightings. \probgkat extended \gkat with convex sums in order to capture probability, in contrast to the general weighted sums we introduce in this paper which do not enforce convexity \cite{rozowski2023probabilistic}. As a result the semantics and axiomatizations differ considerably in form but also in terms of the techniques needed to prove results. We took inspiration for our arguments from the general theory of effectful process algebras \cite{ToddThesis}.
	
 	Both idempotent and non-idempotent semirings have been examined and used as models for computation \cite{KOZEN1994366,weightedhandbook}. While our work is directly inspired by \cite{KAT}, the work on various non-idempotent $^*$-semirings was invaluable to our characterization of semirings appropriate for weighting over. In particular we considered both iteration semirings \cite{IterationSemirings} and inductive $^*$-semirings \cite{InductiveSemirings}. The work on classifying $^*$-semirings is quite thorough and may offer useful classes for researchers working on weighted computation. Semirings in general are well studied and we refer to this theory at various points for our arguments \cite{golan2013semirings}.
	
	We would be remiss to not mention the work on Kleene Algebra with Weights and Tests (\kawt) \cite{sedlar2023kleenealgebratestsweighted}. This is an extension of \kat with weights. It works by weighting the nondeterministic branches of computation in \kat and summing weights. The syntax is quite similar to our own. \kawt is based on a relational semantics which assigns weights to guarded strings in contrast to our coalgebra-based semantics. The syntax of \kawt introduces another class of expression to \kat, the weightings. This class behaves like the Boolean tests in \kat, allowing for both identities, addition and multiplication. These terms are then subterms of programs. For a semiring to be used as the weights in a \kawt it must be idempotent and complete. This allows for a definition of iteration, while keeping each \kat a \kawt. An idempotent semiring is one in which $\forall s \in \mathbb S, \;\; s+s=s$ and a complete semiring is one in which there is an infinitary sum operator $\sum_{I}$ for any index set $I$, and the sum obeys some natural notions of sums and multiplication.
	
	Two key differences between \wgkat and \kawt languages are decidability and the admissible class of semirings. In particular no decidability results exist for \kawt~\cite{sedlar2023kleenealgebratestsweighted} and idempotent and complete semirings are considered. We require the semiring be positive, Conway and refinement. All complete semirings are positive \cite[22.28]{golan2013semirings}, and all complete $^*$-semirings are Conway semirings \cite[Theorem 3.4]{weightedhandbook}. So roughly speaking, we impose the new restriction of being a refinement monoid, but in exchange we drop the requirement of idempotence and loosen the requirement of completeness. We find that this trade includes several new useful semirings, like the natural numbers with infinity under addition and multiplication, or even the non-negative extended reals under addition and multiplication, while ruling out only pathological examples.
	
	Finally our work on the \wgkat automaton leverages the existing work on the general theory of coalgebra \cite{Rutten2000UniversalCA, gumm2002coalgebras, PETERGUMM2000111, gumm2009copower} and in particular monoid-weighted transition systems \cite{GUMM2001185}. We found the property of refinement in monoids to particularly useful and the existing literature to be quite thorough \cite{refinementmonoids, GUMM2001185}. These ideas of mutual refinements are relevant to the semantics of weighted programming. Of course we leverage an existing algorithm for our decision procedure for bisimulation equivalence \cite{10.1007/978-3-030-30942-8_18, wissmann2020efficient, dorsch2017efficientcoalgebraicpartitionrefinement}.
	\section{Conclusion}\label{conclS}
	We presented \wgkat, a weighted language inspired by Kleene Algebra with Tests, equipped with sequential composition, weighted branching over a large class of semirings, conditionals, Boolean tests, primitive actions, and guarded loops. We proposed an operational semantics and gave a sound and complete axiomatization of bisimilarity. Finally, we showed that bisimulation equivalence of \wgkat is efficiently decidable in $O(n^3\log^2(n))$ time. 
	
	One extension of this work would be to axiomatize the coarser notion of trace equivalence, which brings in new axioms. For example, left distributivity over guarded and weighted choice would seem to be sound under trace equivalence, given some conditions on the weights involved. The undecidability of equivalence of power series in the tropical semiring poses issues for a decision procedure for trace equivalence. However characterizing a decision procedure based on properties of the semiring would be an interesting result. 
		
	A practical direction for research would be to take these ideas of weightings and apply them to \netkat, a variant of \kat for verifying network properties. This would allow for examination of quantitative properties which may be of interest to network operators.
	
	Another direction for future research is to prove completeness without the semiring being refinement and positive. These conditions were needed for a bisimulation to exist, which is what made the completeness theorem possible. If it were possible to prove completeness without laying hands on a bisimulation, then behavioral equivalence could still be established without these properties. While we are only aware of pathological examples of semirings which fail to be refinement, allowing for non-positive semirings would include many familiar semirings. Finally, tackling the question of completeness without the uniqueness axiom which remains open for \gkat, \probgkat, and now \wgkat. 

\bibliography{icalp}

\end{document}